\documentclass{article}

\usepackage{arxiv}

\usepackage{natbib}
\bibpunct[, ]{(}{)}{,}{a}{}{,}%
\usepackage{amsmath}
\usepackage{amsthm}
\usepackage{amsfonts}
\usepackage{amssymb}
\usepackage{mathtools}
\usepackage{accents}
\usepackage{xfrac}
\usepackage{pifont} 

\usepackage{soul} 
\usepackage{comment}
\usepackage{hyperref}
\usepackage[dvipsnames]{xcolor}
\usepackage[subfigure]{tocloft} 
\usepackage{bookmark} 

\usepackage{booktabs}
\usepackage{multirow}
\usepackage{enumitem}
\usepackage{threeparttable}

\usepackage{tikz}
\usepackage{placeins}   
\usepackage{caption}
\usepackage[labelfont=sf]{subcaption}

\usepackage{algpseudocode}
\usepackage{algorithm}
\usepackage{listings}

\hypersetup{
	colorlinks = true,
	urlcolor = {black},
	linkcolor = {niceBlue},
	citecolor = {niceRed} 
}

\definecolor{niceRed}{RGB}{190,38,38}
\definecolor{niceBlue}{HTML}{0466a7}
\definecolor{niceGreen}{rgb}{0.2,0.8,0.2}

\newcommand{\Acal}{\mathcal{A}}

\newcommand{\Ncal}{\mathcal{N}}

\newcommand{\Xcal}{\mathcal{X}}
\newcommand{\R}{\mathbb{R}}
\newcommand{\Q}{\mathbb{Q}}
\newcommand{\N}{\mathbb{N}}
\DeclareMathOperator{\E}{\mathbb{E}}
\newcommand{\I}{\mathbb{I}}

\DeclarePairedDelimiter{\abs}{\lvert}{\rvert}

\newcommand{\rr}{\right)}
\renewcommand{\ll}{\left(}
\newcommand{\prob}{\mathbb{P}} 

\newcommand{\game}{\mathcal{G}} 
\newcommand{\nfgdiscr}{\game = \ll \players, \pures, \pay \rr} 

\newcommand{\players}{\mathcal{N}} 
\newcommand{\n}{n} 
\newcommand{\mi}{{-i}} 

\newcommand{\numact}{\ensuremath{m}} 

\newcommand{\pures}{\mathcal{A}} 

\newcommand{\pay}{u} 
\newcommand{\pot}{\phi} 

\newtheorem{theorem}{Theorem}
\newtheorem{lemma}{Lemma}
\newtheorem{remark}{Remark}
\newtheorem{example}{Example}

\newtheorem{definition}{Definition}
\newtheorem{proposition}{Proposition}

\newtheorem{corollary}{Corollary}
\newtheorem{result}{Result}

\newtheorem*{theorem*}{Theorem}
\newtheorem*{inf_theorem}{Informal theorem}
\newtheorem{reftheoreminner}{}
\newenvironment{reftheorem}[1]{%
	\IfBlankTF{#1}
	{\renewcommand{\thereftheoreminner}{\unskip}}
	{\renewcommand\thereftheoreminner{#1}}%
	\reftheoreminner
}{\endreftheoreminner}

\newcommand{\cmark}{\ding{51}} 
\newcommand{\xmark}{\ding{55}} 

\title{Online Optimization Algorithms in Repeated Price Competition: Equilibrium Learning and Algorithmic Collusion}

\author{
	Martin Bichler \\
	\small Department of Computer Science\\
	\small Technical University of Munich\\
	\small \texttt{m.bichler@tum.de} \\
	\And
	Julius Durmann\\
	\small Department of Computer Science\\
	\small Technical University of Munich \\
	\small \texttt{julius.durmann@tum.de}\\
	\And
	Matthias Oberlechner\\
	\small Department of Computer Science\\
	\small Technical University of Munich \\
	\small \texttt{matthias.oberlechner@tum.de}\\
}

\hypersetup{
	pdftitle={Online Optimization Algorithms in Repeated Price Competition},
	pdfsubject={cs.GT},
	pdfauthor={Bichler et al.},
	pdfkeywords={} ,
}

\renewcommand{\headeright}{}
\renewcommand{\undertitle}{}
\renewcommand{\shorttitle}{}

\begin{document}
	
\newgeometry{left=2cm, right=2cm, top=3cm, bottom=3cm}
	
\maketitle

\begin{abstract}
	This paper investigates whether online learning algorithms in pricing produce competitive outcomes or tacit collusion. This issue has drawn considerable attention from competition regulators as algorithmic pricing becomes more common in digital markets. Understanding when such algorithms lead to equilibrium or supra-competitive prices is critical for buyers, sellers, and policymakers.
	
	We study the behavior of multi-armed bandit (MAB) online learning algorithms in repeated price competition. These algorithms require little information to learn, making them realistic models of automated pricing. Our analysis shows that mean-based algorithms, a special variant of online learning algorithms, converge to correlated rationalizable actions. In the Bertrand environments considered, this implies convergence to the Nash equilibrium or adjacent prices. Numerical experiments reveal that most MAB algorithms, including those that are not mean-based, also converge. We observe supra-competitive prices only in specific cases where all sellers implement the same symmetric version of certain algorithms, such as UCB. This effect diminishes as the number of competitors increases.
	
	Our results suggest that, even in a stylized repeated Bertrand competition, sustained supra-competitive prices may be less of a concern when independent agents use different online learning algorithms.
	Our insights are relevant for regulators and managers considering the use of algorithmic pricing algorithms.
\end{abstract}

\vfill
\pagebreak

\setcounter{tocdepth}{2}
\tableofcontents

\vfill
\pagebreak

\section{Introduction}
Algorithmic pricing, where prices are determined by software agents, is becoming increasingly prevalent in online retail markets. 
\citet{chen2016empirical} estimated for 2015 that algorithms were involved in the pricing of approximately one-third of the roughly 1600 best-selling products on the Amazon marketplace. In 2018, the average product price on Amazon was reported to change once every ten minutes.\footnote{\url{https://www.businessinsider.com/amazon-price-changes-2018-8}} Since then, an industry has emerged that specializes in automated pricing software with over 130 companies offering a variety of 185 pricing algorithms in 2021 \citep{calzolari_pricing_2024}. 

From the point of view of a single seller, algorithmic pricing aims to solve an online optimization problem, where a firm's actions are the prices it sets, and the objective is the (average) profit it wants to maximize over time. 
When entering the market, sellers have little information about competitors' costs or the demand of buyers at different prices and need to find out over time which prices maximize profit. 

A key challenge is deciding whether to prioritize short-term revenue (by exploiting a known price with a high payoff) or invest in discovering better long-term pricing strategies (through exploration of different prices). Online optimization algorithms are designed to effectively balance this trade-off, and they scale effectively across large sets of products \citep{bubeck2011introduction}. 
In online exchanges, the algorithms have access to bandit feedback.
This means that after performing an action (setting a price), an agent learns the value of the objective function (his profit) for this specific action. 
This (multi-armed) bandit model in online optimization is particularly suitable for algorithmic pricing applications, which was recognized early on. Bandit algorithms were already suggested for pricing by \citet{rothschild1974two} long before digital platform markets emerged. Today, there is an extensive academic literature on multi-armed bandit algorithms for pricing \citep{trovo2015multi, den2015dynamic, bauer2018optimal, mueller2019low, elreedy2021novel, taywade_multi-armed_2023, qu24}, and there are many resources by practitioners on how to implement bandit algorithms for pricing.\footnote{\url{https://towardsdatascience.com/dynamic-pricing-with-multi-armed-bandit-learning\\-by-doing-3e4550ed02ac}, \url{https://www.griddynamics.com/blog/dynamic-pricing-algorithms}}

Online optimization algorithms target online problems where agents optimize against a stochastic process that is unknown and independent of their actions. Game-theoretical problems are different because the actions of one player impact the objectives of the others. In games, the Nash equilibrium (NE) describes a state where no agent wants to deviate unilaterally. Unlike the optimum in single-player optimization, this state may bring suboptimal payoffs to all agents. 
In spite of a long literature on learning in games, we don't have a comprehensive theory about which algorithms converge to a Nash equilibrium in which types of games \citep{fudenberg1999TheoryLearningGames, young2004strategic, cesa2006prediction}. 

Algorithmic pricing is a game-theoretical setting. An online retail market where multiple sellers compete via prices in a market for goods is naturally modeled as a Bertrand competition \citep{bertrand1883theorie}. This established oligopoly model allows for different assumptions about demand (e.g., all-or-nothing, linear, or logit demand models), and there is a long literature characterizing the emerging equilibrium prices. 
In the standard continuous benchmark with symmetric firms, the equilibrium price equals marginal cost, which maximizes consumer welfare.

Equilibrium analysis assumes rational agents that play their equilibrium strategy from the start. In practice, however, competitors lack information about each other's costs and the demand model, which are both needed to calculate such a price in advance. 
Additionally, changes in demand and supply over time would require players to recalculate their equilibrium strategy.

A growing number of articles shows that optimization and learning algorithms can lead to supra-competitive prices higher than the Nash equilibrium of the stage game \citep{waltman2008q, calvano_algorithmic_2019, abada_artificial_2023, klein_autonomous_2021, abada_collusion_2024, brown_competition_2023}.\footnote{Note that by the Folk Theorem, such a supra-competitive price could still be an equilibrium in an infinitely repeated game, if players have sufficient patience \citep{fudenberg1991game}.}  
The phenomenon is commonly referred to as \textit{algorithmic collusion} \citep{den2023mathematical}, and it raised substantial concerns among regulators \citep{oecd2017} as it decreases consumer welfare. Some states start to regulate the use of pricing algorithms, especially if they are used jointly between competitors \citep{aguiar-curry_ab-325_2025, vinson__elkins_llp_california_nodate}. \citet{abada2024algorithmic} provide a recent treatment of the subject, an overview of relevant literature, and policy implications. They define algorithmic collusion as persistent supra-competitive outcomes produced by learning algorithms without human design to do so, and we adhere to this definition. They also state that these algorithms should be "relevant for use in real market environments", and we argue that bandit-feedback online learning algorithms are exactly that further below.

The discussion on algorithmic collusion draws largely on numerical experiments of certain learning algorithms in a repeated Bertrand competition with a fixed set of sellers and a specific demand model. While real-world environments might be more complex, this model enables an analysis of convergence to the static Nash equilibrium. In particular, analysts can perform comparative statics based on the knowledge of costs and demand functions that are not available in empirical work. One might argue that if algorithmic collusion does not arise in repeated play in a Bertrand pricing game, it is unlikely to emerge in more complex environments where demand and supply change over time. 

There are a number of environments where algorithmic collusion was shown experimentally. Most articles focus on Q-learning \citep{Calvano.2020b, calvano_algorithmic_2021-1, klein_autonomous_2021}, but \citet{hansen_frontiers_2021} recently found evidence for supra-competitive prices with Upper Confidence Bound (UCB, \cite{auer2002bandit}), a multi-armed bandit algorithm, as well. Reward-punishment schemes that might be the result of reinforcement learning algorithms with large state spaces can be ruled out in these environments \citep{lambin2024less}. Our Appendix \ref{app:qlearning} includes Q-learning as a benchmark because it is a dominant baseline in prior algorithmic-collusion papers; our focus and contributions are on scalable online optimization/learning algorithms with bandit feedback.

We argue that it is important to study not only properties of algorithms but also to consider the structure of the games at hand when analyzing the convergence to equilibrium.
Decades of research on learning in games \citep{young2004strategic, foster1997CalibratedLearningCorrelated} has shown that convergence of learning algorithms to a Nash equilibrium typically requires strong assumptions. It is well-known that there are games that cannot be learned via any uncoupled learning dynamics \citep{hart2006stochastic, milionis2022nash}. Also, random normal-form games can lead to cycles or even chaotic dynamics \citep{sanders2018prevalence}. 
The Bertrand competition has more structure that we exploit in this paper to show the convergence of large classes of algorithms relevant to algorithmic pricing. 

\subsection{Contributions}

Prior studies on algorithmic pricing and algorithmic collusion analyzed specific algorithms in a repeated Bertrand competition with specific demand models. We make two main contributions: First, we prove that an important class of online optimization algorithms, mean-based algorithms, converge to correlated rationalizable strategies \citep{brandenburger1987rationalizability}, a solution concept that contains correlated equilibria \citep{aumann1987correlated} and Nash equilibria. Then, we show that in symmetric, discrete Bertrand competition with all-or-nothing, linear, and logit demand, mean-based algorithms converge almost surely to actions that are close to their Nash equilibria. This is because the set of correlated rationalizable strategies is equivalent or close to the set of Nash equilibria of these games. The connection between mean-based algorithms and correlated rationalizable strategies was unknown, and the finding adds to the literature on learning in games. The result has real-world relevance because it provides proof that widely used algorithms such as Exponential Weights (Exp3) converge to a Nash equilibrium in important models of pricing competition. Note that such convergence results of bandit algorithms to Nash equilibrium are rare in the literature on learning in games.

Not all online optimization algorithms suitable for algorithmic pricing are mean-based. 
In a second contribution, we provide extensive experiments with a variety of multi-armed bandit algorithms that have been suggested for algorithmic pricing and show that sustained supra-competitive prices are an exception limited to symmetric installations of UCB algorithms and Q-learning among a small number of sellers.   
We report experiments on Exp3, $\varepsilon$-greedy, UCB, and Thompson sampling, showing that UCB leads to high prices while the other algorithms converge to the Nash equilibrium. 
This holds for all-or-nothing, linear, and logit demand models for symmetric and asymmetric model parameterizations. In experiments in which two or more different algorithms (including UCB) are combined, we find that prices converge to the Nash equilibrium quickly. Supra-competitive pricing with UCB is also much reduced if there are more than three sellers. Overall, these results for a variety of standard Bertrand competition models indicate that non-competitive outcomes are less of a concern with this important class of online optimization algorithms. 

\subsection{Organization of the Paper}

In the next section, we discuss related work on online optimization, algorithmic collusion, and learning in games. In Section \ref{sec:model}, we introduce the relevant game-theoretical solution concepts and the Bertrand competition. Section \ref{sec:convergence-of-mean-based-algorithms} provides the central theoretical result of our paper on the convergence of mean-based algorithms to equilibrium, before we describe our experimental results in Section \ref{sec:experimental-analysis}. Section \ref{sec:conclusions} provides conclusions and an outlook.

\section{Related Work}

In what follows, we provide an overview of the relevant literature. We briefly discuss online optimization algorithms, summarize the literature on algorithmic collusion, and the key results from the literature on learning in games relevant to this paper.

\subsection{Online Optimization}

Online optimization is concerned with making sequential decisions in an unknown environment with the goal of optimizing a performance metric over time. 
At each stage $t = 1,2,\dots$ the agent chooses a strategy $x_t \in \Xcal $ from some set and gets a utility $u_t(x_t)$, which it seeks to maximize. In algorithmic pricing, the action could be the price or the quantity. 

Algorithms are commonly analyzed in two models, the adversarial model and the stochastic model. In the stochastic model, the objective is to minimize the expected regret over the distribution of the input, which is drawn independently and identically from some underlying distribution. In the adversarial model, the input can be chosen by an adversary. A standard performance measure of algorithms generating a sequence of strategies $x_t$ in both models is the notion of (external) regret. 
\begin{definition}[No-regret Algorithm]
	The (external) regret denotes the difference between the aggregated utilities after $T$ stages and the best fixed action in hindsight and is defined by 
	\begin{equation}\label{def:noregret}
		\text{Reg}(T) = \max_{x \in \Xcal} \sum_{t=1}^T u_t(x) - u_t(x_t).
	\end{equation}
	We say that an algorithm has \textit{no regret} if the regret $ \text{Reg}(T)$ grows sublinearly. This means that an algorithm performs on average at least as well as the best fixed action in hindsight in the long run.
\end{definition}

There are different types of no-regret online algorithms. One can distinguish such algorithms based on the feedback available to the agents. Some noteworthy feedback types include:

\begin{itemize}
	\item \emph{Bandit feedback}: The algorithm only receives partial feedback about the performance of its decisions. Typically, the feedback consists of a scalar reward or cost signal associated with the chosen action, but it does not reveal the potential rewards or costs of the other actions. 
	
	\item \emph{Gradient feedback}: The algorithm receives feedback about the first-order derivatives (e.g., gradients) of the cost or payoff function with respect to the chosen action. While it is typically unavailable to learning algorithms in the field, one can compute gradient feedback when the goal is to develop equilibrium solvers rather than to mimic real-world interactions. 
	
	\item \emph{Full feedback}: The algorithm receives feedback on the reward of all its possible actions or is granted access to the entire reward function. This information is given independent of which action was selected, thus providing access to counterfactual information.
\end{itemize}

We will largely focus on bandit feedback, which is a natural algorithm design principle for algorithmic pricing. Sellers set a price, and they learn their profit for this price in the next period. They often have no or only incomplete information about the demand model or the strategies used by all other sellers. This is particularly true on large online retail platforms where there are many substitutes for a good. We also include experiments with full feedback, and our theoretical results cover both bandit-feedback and full-feedback algorithms. 


Exponential Weights, Follow-the-Perturbed-Leader (FTPL), Online Gradient Descent, or Online Mirror Descent are all algorithms that satisfy the no-regret property (which does not make assumptions on the feedback model). Exponential Weights (or the Exp3 variant) uses bandit feedback.

Exp3 is also a \textit{mean-based algorithm} \citep{braverman_selling_2018}, a property of online optimization algorithms that will play an important role in our analysis. Informally, if the mean reward of action $a$ is significantly larger than the mean reward of action $b$, the learning algorithm will choose action $b$ with negligible probability. Apart from Exp3, also Multiplicative Weights Update (MWU), and FTPL are mean-based. We will show that such algorithms converge to a Nash equilibrium in versions of the Bertrand competition. 

Reinforcement learning (RL) algorithms have also been applied to algorithmic pricing \citep{rana2014real, kastius2022dynamic, deng2024algorithmic}. RL with states can be appropriate in environments where demand systematically depends on state variables, such as the day of the week or customer history \citep{abada_artificial_2023}. Most notably, prior work often uses tabular Q-learning with $\epsilon$-greedy exploration where the agents observe the last-period prices \citep{calvano_algorithmic_2021, calvano_algorithmic_2021-1, schaefer_emergence_2023, klein_autonomous_2021}. These algorithms are typically sample-inefficient: as the state and action space grow, they require extensive training rounds and adapt slowly to changing market conditions. This makes them less suitable for dynamic pricing environments, where rapid adaptation is crucial. By contrast, online learning algorithms with bandit feedback can adjust more quickly and scale better. 
As a benchmark, prior work on algorithmic collusion has often applied Q-learning. Its stateless version essentially resembles an $\varepsilon$-greedy bandit algorithm (see also \citet{lambin2024less}, who shows supra-competitive outcomes can emerge even with stateless Q-learning). 
Based on these arguments and discussions with practitioners, we therefore focus on online optimization algorithms with bandit feedback in our work, while reporting (state-based) Q-learning results for comparison in the appendix (Section \ref{app:qlearning}).


\subsection{Algorithmic Pricing and Collusion}

According to \citet{harrington_developing_2018}, "[c]ollusion is when firms use strategies that embody a reward–punishment scheme which rewards a firm for abiding by the supra-competitive outcome and punishes it for departing from it.” While this definition states that a reward-punishment scheme is a necessary mechanism for genuine collusion, \citet{abada_collusion_2024} argue that "[h]arm to consumers comes from supra-competitive prices, not the policy functions that produce those prices nor the learning algorithm that produces the policy function.” This is why some contributions from the computational literature also refer to supra-competitive outcomes as algorithmic "collusion".

Algorithmic collusion has drawn substantial attention from the research community and from policymakers, as we outlined in the introduction. 
Most of this literature analyzes specific algorithms such as Q-learning for specific model variations, i.e., Bertrand oligopolies with standard all-or-nothing demand, with linear, or logit demand \citep{bertrand1883book}. 
\citet{Calvano.2020} analyze a Bertrand competition with logit demand and constant marginal cost. They find that  Q-learning agents under self-play can lead to supra-competitive prices higher than the Nash equilibrium. 
A related sequential move pricing duopoly environment with linear demand (instead of the simultaneous move Bertrand model in \citet{Calvano.2020}) was analyzed by \citet{klein_autonomous_2021}, who also found collusion with Q-learning agents. 
\citet{asker_impact_2024} detect in their experiments on Bertrand competition with standard (all-or-nothing) demand that the outcome depends on specifics of the Q-learning algorithm  (e.g., synchronous vs. asynchronous updating). \citet{lambin2024less} analyzes a two-staged exploration scheme that helps explain how collusion can sometimes form with Q-learning under self-play, and \citet{schaefer_emergence_2023} empirically derives a probability boundary for the evolution of cooperation in a repeated Prisoner's dilemma. While there are measures that can be implemented against collusion, e.g., on a platform-level \citep{johnson_platform_2023}, these articles agree on the potential risk to consumer welfare posed by algorithmic sellers, although some raise doubts about the purposeful coordination between the algorithms.

In contrast, \citet{abada_collusion_2024} analyze Q-learning in Bertrand oligopolies and show that Q-learning algorithms with sufficiently large $\epsilon$-greedy exploration show no collusion. \citet{den_boer_artificial_2022} provide a detailed analysis of the inner workings of Q-learning and reveal that there is no immediate reason to believe that Q-learning would lead to collusion easily.
In addition, \citet{eschenbaum2022robust} criticize the claim that algorithms can be trained offline to successfully collude online in different market environments. The authors find that collusion breaks down when collusive reinforcement learning policies are extrapolated from a training environment to the market. 

Most of this experimental literature on algorithmic collusion is based on Q-learning, although there is no evidence that this algorithm is particularly important or widespread for algorithmic pricing.  
\citet{hansen_frontiers_2021} analyze the price levels that arise in a duopoly setting with UCB ("Upper Confidence Bound") agents. They run a series of experiments where these agents interact simultaneously in a Bertrand economy competition with linear demand functions. The agents observe a perturbed estimate of their revenues, which are a result of their prices and the corresponding demand. \cite{hansen_frontiers_2021} find that, under sufficiently small reward noise, agents eventually explore prices in a correlated manner, giving rise to supra-competitive outcomes.\footnote{Note that in this literature review, we focus on algorithms that were not designed to collude \citep{meylahn_learning_2022}.} \citet{douglas_naive_2024} provide more details on these results by analyzing UCB and an $\epsilon$-greedy bandit algorithm in a repeated Prisoner's dilemma. They argue that deterministic algorithms (UCB) always learn to cooperate ("collude") while randomized algorithms ($\epsilon$-greedy) never do so in this simple game.

A more extensive discussion of the literature is provided by \citet{abada2024algorithmic}. 
In summary, the literature on algorithmic collusion largely focuses on specific algorithms (mostly Q-learning) that are employed in specific types of oligopoly models. The same type of algorithm is used in a symmetric model. 
We analyze a variety of multi-armed bandit algorithms competing against the same, but also against different algorithms. We analyze standard all-or-nothing demand, linear, and logit demand models in a Bertrand competition. 

Recent work has emphasized that supra-competitive algorithmic pricing and algorithmic collusion should not be conflated. \citet{hartline2024regulation} propose a notion of plausible algorithmic non-collusion based on certifying unilateral competitive behavior, while \citet{Hartline2026ClarificationAlgorithmicCollusion} argues that some supra-competitive outcomes are better understood as unilateral non-competitive behavior rather than collusion. This distinction is important for regulatory interpretation: anti-competitive algorithmic pricing is the broader category, while collusion is a narrower case in which multiple firms’ behavior jointly sustains supra-competitive prices. We use “algorithmic collusion” in the broad outcome-based sense, merely indicating supra-competitive prices.

\subsection{Learning in Games} \label{sec:learning_games}

The literature on algorithmic collusion is inherently tied to that on learning in games. The latter has a long history and asks what type of equilibrium behavior (if any) may arise in the long run of a process of learning and adaptation, in which agents seek to maximize their payoff while learning about the actions of other agents through repeated interactions \citep{fudenberg1999TheoryLearningGames}. Cournot's study of duopoly \citep{cournot1838recherches} already introduced a Nash equilibrium and a particular learning process. Cournot's model of best reply dynamics appears unrealistic as a model for algorithmic pricing because players would need a lot of information about competitors. Later, \citet{brown1951IterativeSolutionGames} suggested fictitious play, which converges to equilibrium distributions for any two-player, finite-strategy, zero-sum game. However, \citet{shapley1964TopicsTwopersonGames} established that it can lead to cycles and that the frequency distribution does not necessarily converge. 
Many learning algorithms have been developed, ranging from iterative best-response to first-order online optimization algorithms in which agents follow their utility gradient \citep{mertikopoulos2019learning, bichler2023soda}.

In this general context, it is well known that the dynamics of learning agents do not always converge to a Nash equilibrium \citep{milionis2022nash, vlatakis2020no}: they may cycle, diverge, or be chaotic, even in zero-sum games, where the Nash equilibrium is tractable \citep{mertikopoulos2018cycles, bailey2018multiplicative}. 
While there is no comprehensive characterization of games that are ``learnable'' and one cannot expect that uncoupled dynamics lead to Nash equilibrium in all games \citep{hart2003uncoupled, milionis2022nash}, there are some important results regarding learners. A classical result is that the class of no-regret learning algorithms converges to the so-called \emph{[coarse] correlated equilibrium} ([C]CE) of a game \cite{fudenberg1999TheoryLearningGames}, a result that has recently been extended to Markov games for an algorithm called "V-learning" \citep{jin_v-learningsimple_2024}.
[C]CEs are supersets of Nash equilibria. However, CCEs can also contain dominated strategies and are a rather weak solution concept. \citet{wu_correlated_2008} and \citet{jann_correlated_2015} have shown that the Bertrand competition with "all-or-nothing" demand has a unique CE, and convergence guarantees to CE exist for no-internal-regret algorithms, like the scheme introduced by \cite{foster_regret_1999}. However, these schemes lack the simplicity of no-external-regret algorithms. Our results provide a complementary guarantee for mean-based algorithms. 

Less is known about conditions of games in which learning algorithms converge to a Nash equilibrium. In a landmark paper, \citet{monderer1996potential} introduced the class of \textit{potential games} and showed that Cournot oligopolies with linear price or cost functions are potential games. 
Potential games are guaranteed to have at least one pure Nash equilibrium. Importantly, it was shown that algorithms such as best and better response dynamics or fictitious play converge to a Nash equilibrium in these games. More recently, convergence of MWU \citep{palaiopanos2017multiplicative} and Exp3 \citep{cohen2017learning} in potential games was also shown.

Another important property of a game is supermodularity \citep{vives1999oligopoly, topkis1998supermodularity, milgrom1990rationalizability}. In supermodular games, the best response of each player is positively correlated with the strategies chosen by other players. 
\citet{milgrom1990rationalizability} showed that Bertrand oligopoly games are (log-)supermodular if each firm's elasticity of demand is a decreasing function of its competitors' prices. 
Later, \citet{milgrom1991adaptive} established that adaptive learning algorithms converge to the unique NE in a Cournot duopoly model and Bertrand oligopoly models with linear and logit demand. 
Adaptive learning algorithms include best-response dynamics and fictitious play. In the bandit algorithms that we analyze, an agent can only observe their profit for a particular action, but not necessarily the strategies of all other players, as is necessary for these earlier equilibrium-finding algorithms.

\citet{mertikopoulos2024unified} develops a general stochastic approximation template that unifies many online learning algorithms (gradient, multiplicative weights, optimistic, bandit, etc.) and analyzes their convergence properties across broad classes of games. Our paper complements this by focusing on repeated Bertrand competition, showing how such bandit dynamics manifest in price competition with direct implications for market design and regulation.

It is important to note the distinction between learning algorithms that have access to full payoff information (as in \citet{fudenberg1999TheoryLearningGames}) and bandit feedback, where players only observe realized payoffs from their chosen actions; while convergence results are more established in the former case, our work contributes to the latter, where such guarantees are harder to obtain.
Our work also differs from the extensive literature on dynamic pricing in the monopolist setting (see, e.g., \citet{den2015dynamic}), which focuses on single-seller learning of demand, whereas we study learning dynamics in competitive multi-seller environments where strategic interaction is central.

Finally, some recent papers aim to find an algorithm that, when employed by all agents, yields a Nash equilibrium. They require a common and coordinated exploration scheme. \citet{yang_competitive_2024} introduce an algorithm that learns the Nash equilibrium in a variety of Bertrand settings, based on a common exploration scheme that all participants need to follow. \citet{goyal_learning_2023} restrict their attention to the multinomial logit-demand case and show that the Nash equilibrium can be found with a specialized, decentralized online learning algorithm. Similar to our work, \citet{wang_learning_2022} focus on rationalizability in (coarse) correlated equilibria. They provide a number of algorithms that are able to achieve this under limited coupling requirements.

\section{Model}
\label{sec:model}

In the following, we will first introduce the necessary notation and definitions before we define different types of Bertrand competition models.

\subsection{Basic Definitions and Notation}

We begin with the concepts of a finite normal-form game, a Nash equilibrium, and a symmetric game. 
A normal-form game is a representation in game theory that defines the strategies available to each player, their corresponding payoffs, and the resulting outcomes in a simultaneous and strategic interaction. Formally, we have the following definition. 

\begin{definition}[(Finite) Normal-form Game\index{normal-form games}]\label{def:normalform}
	A \emph{normal-form game} with $\n$ players can be described as a tuple $\nfgdiscr$, with a finite set of players $i \in \players = \{1, \dots, \n\}$, an action space $\Acal= \Acal_1 \times \cdots \times \Acal_\n$ for all $i \in \players$, and payoff or utility functions $u = (u_1, \dots , u_\n)$ with $u_i:  \Acal \rightarrow \R$.
	We call the game a finite normal-form game if the action spaces are also finite, i.e., $\vert \Acal_i \vert < \infty$.
\end{definition}
Please note the slight abuse of notation in the subscripts: Previously, we used indices to refer to time (e.g., "$u_t(\dots)$"); now, we mostly differentiate players (e.g., "$u_i(\dots)$"). In the following, we will use letters $s$, $t$, and $T$ whenever we consider time, and $i$, $j$, and $n$ whenever we refer to players.

In a normal-form game, all agents $i = 1, \dots, \n$ submit their actions $a_i$ simultaneously to form an action profile $a = (a_1, \dots, a_\n)$. 
We often abbreviate the profile by $(a_i, a_{-i})$ where $a_{-i} = (a_1, \dots, a_{i - 1}, a_{i + 1}, \dots, a_\n)$. Similarly, we will sometimes index parts of combined spaces as follows: $\Acal_{-i} = \Acal_1 \times \dots \times \Acal_{i-1} \times \Acal_{i+1} \times \dots \times \Acal_{\n}$.
The agents may also play a distribution over actions, known as a mixed strategy, denoted by $x_i \in \Delta (\Acal_i)$, where $\Delta (\Acal_i) $ is the probability simplex over $\Acal_i$. 
For any joint distribution $x \in \Delta(\Acal)$ over action tuples, the expected utility of player $i$ is $u_i(x) = \E_{a \sim x} \left[ u_i(a_i, a_{-i}) \right]$. Similarly, we will write $u_i(a_i, x_{-i}) = \E_{a_{-i} \sim x_{-i}} \left[u_i(a_i, a_{-i}) \right]$ for the expected utility of some action $a_i$ of player $i$ against randomizing opponents.



A Nash equilibrium is a situation in a strategic interaction where each player's strategy is optimal given the strategies chosen by all other players, and no player has an incentive to unilaterally deviate from their chosen strategy.

\begin{definition}[Nash Equilibrium (NE)]
	In a normal-form game $\nfgdiscr$, a strategy profile $x^* = (x_1^*, \ldots, x_\n^*)$ is a \emph{Nash equilibrium} if, for every player $i \in \Ncal$, we have:
	\begin{equation}
		u_i(x_i^*, x_\mi^*) \geq u_i(x_i, x_\mi^*), \quad \forall x_i \in \Delta \Acal_i.
	\end{equation}
\end{definition}

It is well-known that computing Nash equilibria is PPAD-hard in the worst case \citep{daskalakis2009complexity}.
There are specific subsets of normal-form games where the utilities have more structure. This allows us to find an equilibrium with faster iterative algorithms. Potential games \citep{monderer1996potential} and supermodular games \citep{topkis1979submodular_games, milgrom1990rationalizability} have received significant attention. 

\begin{definition}[Potential Game] \label{def:potential}
	A game $\nfgdiscr$ is a \emph{potential game} if there exists a potential function $\pot : \Acal \rightarrow \R$ such that for every player $i \in \Ncal$ and every pair of actions $a_i, a'_i \in \Acal_i$ 
	\begin{equation}
		u_i(a_i, a_\mi) - u_i(a'_i, a_\mi) = \pot(a_i, a_\mi) - \pot(a'_i, a_\mi), \quad \forall a_\mi \in \Acal_\mi.
	\end{equation}
\end{definition}

It is known that best response dynamics, better response dynamics, Fictitious Play, Replicator Dynamics, and simultaneous gradient ascent all converge in potential games \citep{fudenberg1999TheoryLearningGames, monderer1996potential, sandholm2010population, swenson2018best}. \citet{cohen2017learning} could also show convergence of bandit algorithms such as Exp3 to a Nash equilibrium in potential games. 

Bertrand competitions with some demand models are \textit{supermodular games}, which were introduced by \citet{topkis1979submodular_games}. 
Supermodular games are defined on complete lattices and require utility functions that are supermodular and have "increasing differences" (see, e.g., \citep{milgrom1990rationalizability}). Since we only consider one-dimensional action spaces $\Acal \subset \R$, the utility functions trivially satisfy the supermodularity condition, and we can use the componentwise ordering to define a complete lattice on $\Acal_{i}$ and $\Acal_{-i}$. Then, the definition of supermodular games reduces to the following.
\begin{definition}[Supermodular Game]\label{def:supermodulargame}
	A normal-form game \( \nfgdiscr \) with $\Acal_i \subset \R$ is called a \textit{supermodular game} if it satisfies the following conditions for each player \( i \in \players \):
	\begin{enumerate}
		\item[(1)] $u_i$ is order upper semi-continuous in $a_i$ and continuous in $x_{-i}$ and has a finite upper bound.
		\item[(2)] $u_i$ has increasing differences in $a_i$ and $a_{-i}$, i.e., for any two action profiles $a=(a_i,a_{-i})$ and $a'=(a_i',a_{-i}')$ such that $a_i' \geq a_i$ and $a_j' \geq a_j$for all $j \neq i$, the following inequality holds:
		\begin{equation}
			u_i(a_i', a_{-i}') - u_i(a_i, a_{-i}') \geq u_i(a_i', a_{-i}) - u_i(a_i, a_{-i}).
		\end{equation}
	\end{enumerate}
\end{definition}
If we restrict ourselves to finite normal-form games, i.e., discrete actions, then we only have to check the increasing differences property (2) and can ignore the continuity requirements (1).

It is known that pure strategy Nash equilibria exist in supermodular games. The set of actions surviving iterated strict dominance has a greatest and least element, and both are Nash equilibria \citep{levin2006solution}. Consequently, if a game has a unique Nash equilibrium, then strict elimination of dominated strategies finds this equilibrium. 
However, the elimination of dominated strategies is not a learning method that can be used independently by all players of the game because it requires complete information about the payoff matrix of a game. Such information is not available in algorithmic pricing. Therefore, we will focus on uncoupled algorithms that can be used independently by the agents and that only learn from bandit feedback after each round.

\subsection{Bertrand Competition} 
\label{sec:bertrand-oligopoly}

The Bertrand pricing game or Bertrand competition \citep{bertrand1883theorie} is an economic model that describes the revenue of a firm according to the price it sets for its product and the resulting demand. 
The firm's action  $a_i \in \Acal_i$ is the price for a good it wants to produce. All firms produce the same good, and they all take action simultaneously. 
Firms compete for demand with their prices and affect each other's revenues: The demand depends on all agents' actions and is decreasing in the agent's own price. 
As is common, we assume that the cost functions are linear in the demand. In summary, the utility can be described by

\begin{equation}
	u_i(a_i, a_\mi) =  d_i(a) \cdot (a_i - c_i).
\end{equation}

\subsubsection{Demand Models}

We consider different demand models. The \textit{standard} or \textit{all-or-nothing demand} (also called Bertrand demand) is similar to an auction, where only the highest bidder gets the item. 
In this case, the firm with the lowest price gets all the demand. Different from auctions, where the valuation of an item does not depend on its price, the demand in this competition is additionally linearly decreasing in the agent's price. 
If multiple firms offer the lowest price, the demand is shared equally. This leads to a demand function given by
\begin{equation} \label{eq:demand_standard}
	d_i(a_i, a_{-i}) = \begin{cases} \frac{D}{n_{min}} (1- a_i) &\text{if } i \in \arg \min_{j \in \Ncal} a_j \\ 0 &\text{else} \end{cases}, 
	\tag{standard} 
\end{equation}
where $D > 0$ is the maximum total demand, and $n_{min} := \abs{\arg \min_{j \in \Ncal} a_j}$ is the number of firms with the lowest price.

Another model that was used in \cite{Calvano.2020} is the \textit{multinomial or logit demand}. Here, the demand is split between the $n$ agents/goods and some outside good (index $0$) according to
\begin{equation}\label{eq:logit_demand}
	d_i(a_i,a_\mi) = \dfrac{\exp \left( \frac{\alpha_i - a_i}{\mu} \right) }{ \exp \left( \frac{\alpha_0}{\mu} \right) + \sum_{j=1}^n \exp \left( \frac{\alpha_j - a_j}{\mu} \right)}. 
	\tag{logit}
\end{equation}
The parameters $\alpha_i > 0$ capture different product quality indices and $\mu > 0$ models a product differentiation between the goods, i.e., if $\mu \rightarrow 0$, the goods are perfect substitutes. (Please note the difference between $\alpha$ ("alpha") and $a$ in the formula.)

And finally, we also introduce the \textit{linear demand} that was considered in \citet{hansen_frontiers_2021}. The agents' demand function is given by 
\begin{equation}\label{eq:linear_demand}
	d_i(a_i, a_{-i}) = \alpha_i - \beta_i a_i + \tfrac{\gamma}{n-1} \sum_{j \neq i} a_j. 
	\tag{linear} 
\end{equation}
Similar to the standard model, the demand decreases with rate $\beta_i > 0$ in the price of the agent, but it also increases with rate $\gamma > 0$ in the sum of the prices of all other agents. 
We assume that the effect of the own price is greater than the effect of the average opponents' prices, i.e., $\beta_i > \gamma$. 

In a standard Bertrand model with homogeneous products and symmetric firms, prices equal marginal costs in equilibrium. However, this changes with asymmetries or product differentiation. 
With asymmetric costs between firms, the Nash equilibrium typically involves the low-cost firm pricing at or just below the marginal cost of the high-cost firm. With product differentiation (as in logit demand), equilibrium prices are typically above marginal costs due to firms having some market power. Independent of the demand model, we will use the Nash equilibrium as a baseline for the learning outcome.


\subsubsection{Game-Theoretical Properties}

Let us now discuss some properties that are already known about Bertrand pricing models.
\citet{milgrom1991adaptive} showed that Bertrand competition with linear and logit demand and continuous actions are (log-)supermodular games. As indicated earlier, iterated elimination of dominated strategies converges to a unique Nash equilibrium in this model. However, this is not the case for standard all-or-nothing demand, as we show (see Appendix, Proposition \ref{prop:super_standard}).
Potential games have even stronger properties, as we have seen. 
We show that the Bertrand competition with a linear demand function is not only supermodular, but also potential (see Appendix, Proposition \ref{prop:potential}). 
For Bertrand competitions with standard all-or-nothing and logit demand, we can show that these games are, in general, not potential games (see Appendix, Proposition \ref{prop:potentialstandard} and \ref{prop:potentiallogit}). 
An overview of these concepts for the different utility models is given in Table \ref{tab:overview_concepts}.

\subsection{Alternative Solution Concepts}

We explore alternative solution concepts to the Nash equilibrium. While the latter is computationally hard to find, there exist simple iterative algorithms that can identify strategy profiles corresponding to these solution concepts. 
In some games, the alternatives coincide with the Nash equilibrium, which is what we will leverage in our convergence proof for mean-based algorithms. 

\subsubsection{(Coarse) Correlated Equilibria}

Apart from Nash equilibria (NE), Correlated Equilibria (CE) and Coarse Correlated Equilibria (CCE) have received significant attention in the literature on learning in games. It is well known that $NE \subseteq CE \subseteq CCE$.
\begin{definition}[Correlated Equilibrium]
	A joint distribution $\sigma \in \Delta(\Acal)$ is a \emph{correlated equilibrium} if for every player $i \in \players$ and for all actions $a_i, a_i' \in \Acal_i$,
	\[
	\mathbb{E}_{a \sim \sigma}[u_i(a_i, a_{-i}) \mid a_i] \geq \mathbb{E}_{a \sim \sigma}[u_i(a_i', a_{-i}) \mid a_i].
	\]
\end{definition}

\begin{definition}[Coarse Correlated Equilibrium] \label{def:cce}
	A joint distribution $\sigma \in \Delta(\Acal)$ is a \emph{coarse correlated equilibrium} if for every player $i \in \players$ and for all actions $a_i' \in \Acal_i$,
	\[
	\mathbb{E}_{a \sim \sigma}[u_i(a_i, a_{-i})] \geq \mathbb{E}_{a \sim \sigma}[u_i(a_i', a_{-i})].
	\]
\end{definition}

Both CE and CCE capture a scenario in which a mediator recommends actions to the players, based on a joint action profile distribution, such that no player wants to deviate. For CEs, deviation may depend on the recommended action, while the deviation for CCEs does not take this information into account.
The corresponding CCE constraints form a polytope, and therefore, individual CCEs can be computed via a linear program. 
It was shown that algorithms with no internal regret converge to a CE and those with no external regret to a CCE \citep{foster1997CalibratedLearningCorrelated}. Popular bandit algorithms such as Exp3 are no-external-regret algorithms. However, CCEs are a weak solution concept, and they can contain dominated strategies \citep{viossat2013no}. Also, the set of CCEs can be a very large superset of the set of Nash equilibria. 
In Example \ref{ex:cce}, we computed different CCEs for Bertrand competition models with linear and all-or-nothing demand. 

\begin{example}[Coarse Correlated Equilibria in Bertrand Competitions] \label{ex:cce}
	We consider Bertrand competitions with 2 players, action spaces $\Acal = \{0.1, 0.2, \dots, 0.9 \}$, cost parameters $c_1 = c_2 = 0$, and two different demand functions, namely linear demand ($\alpha_i = 0.48, \, \beta_i = 0.9, \, \gamma = 0.6$ for $i \in \players$) and all-or-nothing demand ($D=1$). The computed CCEs are visualized in Figure \ref{fig:cce}.
	\begin{figure}[ht]
		\centering
		
		\hfill
		\begin{subfigure}{0.4\textwidth}
			\includegraphics[width=\textwidth]{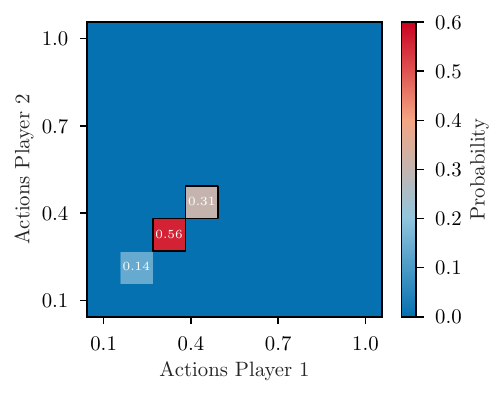}
			\subcaption{Linear Demand \label{fig:cce_linear}}
		\end{subfigure}
		\hfill
		\begin{subfigure}{0.4\textwidth}
			\includegraphics[width=\textwidth]{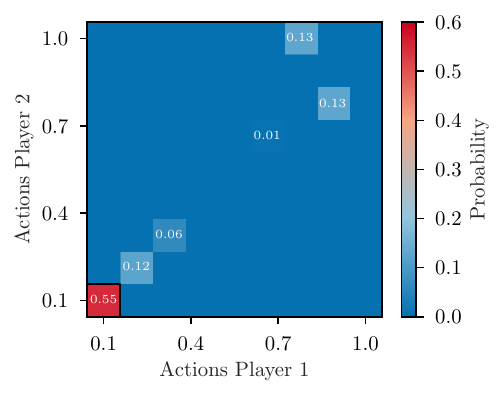}
			\subcaption{All-or-Nothing Demand\label{fig:cce_standard}}
		\end{subfigure}
		\hfill
		
		\caption{Coarse Correlated Equilibria for Different Bertrand Competitions}
		\label{fig:cce}
		
		\centering \small
		In the figures, we visualize exemplary CCEs and highlight the Nash equilibria (boxes with black frames). The CCEs were computed by solving an LP where the constraints come from the definition of a CCE. The LP objective is to maximize the Wasserstein distance to the Nash equilibria.
	\end{figure}
	For the linear demand model, the CCEs contain only actions near the Nash equilibrium. 
	However, for the all-or-nothing demand model, we can find CCEs with dominated actions. This is why the guarantee that no-regret learners converge to a CCE is not sufficient in a Bertrand competition with all-or-nothing demand.
\end{example}

\subsubsection{Serially Undominated Set}

Another important solution concept that can be computed by means of a simple, iterative algorithm is that of a serially undominated set. For this, let us introduce two types of dominance relationships. 

\begin{definition}[Dominated Actions]
	An action $a_i \in \Acal_i$ is 
	\begin{itemize}
		\item \emph{strongly dominated}, if there exists another action $\hat{a} \in \Acal_i$ such that $u_i(\hat{a}, a_{-i}) > u_i(a_i, a_{-i})$
		\item \emph{strictly dominated}, if there exists a mixed strategy $\hat{x} \in \Delta(\Acal_i)$ such that $u_i(\hat{x}, a_{-i}) > u_i(a_i, a_{-i})$ 
	\end{itemize}
	for all opponents' action profiles $ a_{-i} \in \Acal_{-i}$.
\end{definition}
By iterative removal of dominated actions, we end up with a serially undominated set. 
From \citep{milgrom1990rationalizability}, we know that in supermodular games, the iterative procedure of removing strongly dominated strategies leads to a set where the maximum and minimum elements are component strategies of pure Nash equilibria. This set includes all Nash equilibria, correlated equilibria, and rationalizable actions \citep{bernheim_rationalizable_1984, pearce_rationalizable_1984}.

Iteratively crossing out strictly dominated actions, also known as  \emph{Iterated Strict Dominance (ISD)}, leads to a smaller set of actions, which we will call \emph{strictly serially undominated set (SSU)}. A game is said to be \textit{strict dominance solvable} if ISD yields a unique strategy profile. In finite games, this strategy profile is necessarily the unique Nash equilibrium \citep{fudenberg1991game, Kolpin2009Strict}. Note that in order to perform ISD, we need complete information about the payoff matrix. This assumption is too strong for algorithmic pricing, which is why we will focus on bandit algorithms later. In finite games, there is a close relationship between actions that survive ISD and correlated rationalizable actions, which we will explore in the following subsection.

\subsubsection{Correlated Rationalizable Set}

Our main convergence statement is centered on the concept of \textit{correlated rationalizability} \citep{brandenburger1987rationalizability}. These are actions that are a best response against any joint mixed strategy of the opponents that is also correlated rationalizable. 


\begin{definition}[Correlated Rationalizable Set (CR)]
	\label{def:rationalizable-actions}
	Let $\bar{\Acal} = \bar{\Acal}_1 \times \dots \times \bar{\Acal}_{\n}$, $\bar{\Acal}_i \subseteq \Acal_i$, be the largest product set of actions such that, for each $i \in \players$ and $a_i \in \bar{\Acal}_i$
	\begin{equation*}
		\exists ~ \hat{x}_{-i} \in \Delta(\bar{\Acal}_{-i}), ~ \forall ~ a_i' \in \Acal_i : ~ u_i(a_i, \hat{x}_{-i}) \geq u_i(a_i', \hat{x}_{-i}).
	\end{equation*}
	The set $\bar{\Acal}$ is called the \emph{correlated rationalizable set} (CR).
\end{definition}
This means that an action $a_i $ is in $\bar{\Acal}_i$ ("correlated rationalizable") if it maximizes the player's utility for some probability $\hat{x}_{-i}$ with support on $\bar{\Acal}_{-i}$. The definition is recursive, which ensures that only actions are in the support of $\hat{x}_{-i}$ which are correlated rationalizable themselves. 

Different from the definition of \textit{rationalizable} actions (\cite{bernheim_rationalizable_1984}, \cite{pearce_rationalizable_1984}), correlated rationalizability allows the mixed strategies $\hat{x}_{-i}$ to be dependent/correlated probability distributions. The set of rationalizable actions is a subset of the set of correlated rationalizable actions, but the subset relation may not be strict.
\citet{bernheim_rationalizable_1984} has shown that all Nash equilibria are rationalizable, so they are also correlated rationalizable. 
In finite games, iterated strict dominance and correlated rationalizability give the same solution set \citep{levin2006solution}. While this set (SSU) can be computed with ISD having the complete payoff matrix available, we will show that uncoupled mean-based algorithms also find specific instances of the correlated rationalizable set.

\subsubsection{Summary}

We summarize the relation between different solution concepts for games in Figure \ref{fig:solution-concepts-venn-diagram}. 
It is well known that $NE \subseteq CE \subseteq CCE$ in finite games. 
The relationship between correlated rationalizable sets (or strictly serially undominated sets) and [coarse] correlated equilibria is more complex. 
We can show that actions contained in the support of correlated equilibria are also correlated rationalizable.
Additionally, while it is known that CCE may contain dominated strategies \citep{viossat2013no}, we can show that actions supported by a CCE are not a superset of correlated rationalizable actions. 
The following propositions formalize these relations. Please refer to the appendix for both proofs.
\begin{proposition}
	\label{prop:CE-is-correlated-rationalizable}
	Given a finite normal-formal game, any action profile that is in the support of a correlated equilibrium is also in the correlated rationalizable set.
\end{proposition}
\begin{proposition}\label{prop:SSU-not-in-CCE}
	There exists a finite-normal form game where an action from the strictly serially undominated set is not in the support of any CCE.
\end{proposition}

\begin{figure}[ht]
	\centering
	\begin{tikzpicture}
		\draw[line width=0.33mm, rounded corners, dashed] (-6, -1) rectangle (4, 2); 
		\draw[line width=0.33mm, rounded corners] (-4, -0.75) rectangle (5, 1.75); 
		\draw[line width=0.33mm, rounded corners, dashed] (-2, -0.5) rectangle (3, 1.5); 
		\draw[line width=0.33mm, rounded corners] (0, -0) rectangle (2, 1); 
		\node at (1, 0.5) {NE}; 
		\node at (-1, 0.75) {CE}; 
		\node at (-3, 1) {CR=SSU};   
		\node at (-5, 1.25) {CCE};
		\node at (4.5, 0.5) {}; 
	\end{tikzpicture}
	\caption{Relation Between Solution Concepts}
	\label{fig:solution-concepts-venn-diagram}
	
	\centering \small
	The displayed relations are between the support sets of the solution concepts.\\
	CCE = Coarse Correlated Equilibrium, SSU = Strictly Serially Undominated, CR = Correlated rationalizable, CE = Correlated Equilibrium, NE = Nash equilibrium.
\end{figure}

As indicated earlier, algorithms that have no (external) regret (see Equation \eqref{def:noregret}) converge to a CCE, while algorithms with no internal regret (such as regret matching) converge to CEs \citep{foster1997CalibratedLearningCorrelated}.
ISD computes the strictly serially undominated set (SSU) and, thus, also the correlated rationalizable set (CR), but this algorithm requires full access to the payoff matrix. 

Table \ref{tab:overview_concepts} summarizes properties of games that are satisfied by Bertrand competitions with different demand models (all-or-nothing, linear, and logit demand) and algorithms that are known to converge to a Nash equilibrium if certain of the game properties hold. 

\begin{table}[h]
	\centering
	\footnotesize
	
	\begin{tabular}{ l l l l r }
		\hline
		\multirow{2}{*}{\textbf{Game properties}} & \multicolumn{3}{c}{\textbf{Demand Model}} & \multirow{2}{*}{\textbf{Convergence to NE}}  \\ 
		& All-or-Nothing & Linear & Logit &  \\ 
		\hline
		Potential   & \xmark ~ (Prop.~\ref{prop:potentialstandard})  & \cmark ~ (Prop.~\ref{prop:potential}) & \xmark ~ (Prop. \ref{prop:potentiallogit})  & FP, BR, Exp3 \\
		Supermodular \& unique NE   & \xmark ~ (Prop. \ref{prop:super_standard}) & \cmark ~ (MR90)  & \cmark ~ (MR90) & ISD, FP, BR (MR90)\\
		Unique CR  & \cmark ~ (Prop. \ref{prop:standard-demand-correlated-rationalizable}) & \cmark\textsuperscript{1}~ (Prop. \ref{prop:linear-demand-correlated-rationalizable-actions}) & \cmark\textsuperscript{1} ~ (Prop. \ref{prop:logit-demand-correlated-rationalizable-actions}) & ISD, \textit{MB} (Thm. \ref{thm:bertrand}) \\ 
		\hline
	\end{tabular}
	\caption{Overview of Games and Properties. }
	\label{tab:overview_concepts}
	
	\raggedright \small
	Fictitious play (FP), best response (BR), Exponential Weights (Exp3), iterated strict dominance (ISD), mean-based (MB); (MR90) refers to \citet{milgrom1990rationalizability}.\\
	\textsuperscript{1} The set of CR is almost unique. We show that only two neighboring actions per agent can be in CR. Depending on the discretization, either one or both are NE. 
\end{table}

The Bertrand competition with linear demand is a potential game for which we know that even bandit algorithms such as Exp3 converge. However, other demand types do not lead to a potential game as we showed earlier. \citet{milgrom1990rationalizability} showed that Bertrand competition with continuous actions and linear demand is supermodular and has a unique Nash equilibrium, while with logit demand, it is log-supermodular. We don't know of convergence proofs for bandit algorithms, but we know that ISD, Fictitious Play, and best response dynamics find the unique Nash equilibrium in such supermodular games. ISD also finds correlated rationalizable sets, because they are equivalent to serially undominated sets. Our main analytical contribution is to show that mean-based algorithms, with Exp3 as the leading example, converge to the correlated rationalizable set. Importantly, in a Bertrand competition with all-or-nothing or linear demand, this set is either a singleton, meaning it must be the Nash equilibrium, or consists of only two adjacent actions. 

\section{Convergence of Mean-Based Algorithms}
\label{sec:convergence-of-mean-based-algorithms}




Let us now focus on mean-based algorithms. These are online optimization algorithms that pick actions with low average rewards with low probability \citep{braverman_selling_2018, deng2022nash}. 

\begin{definition}[Mean-based Algorithm]
	Let $\alpha_t(a)$ be the average reward of action $a \in \Acal$ in the first $t-1$ rounds: $\alpha_t(a) = \frac{1}{t-1} \cdot \sum_{s=1}^{t-1} u_s(a)$. 
	An algorithm is $\gamma_t$-mean-based if, for any $a \in \Acal$, whenever there exists $a' \in \Acal$ such that $\alpha_t(a^\prime)-\alpha_t(a) > \gamma_t$, the probability that the algorithm picks $a$ at round $t$ is at most $\gamma_t$. An algorithm is \emph{mean-based} if it is $\gamma_t$-mean-based for some decreasing sequence $(\gamma_t)_{t=1}^\infty$ such that $\gamma_t \rightarrow 0$ as $t \rightarrow \infty$.
\end{definition}

Many no-regret algorithms, such as Exp3, MWU, and FTPL, are also mean-based. We know from \citet{kolumbus_auctions_2022} that, if a mean-based algorithm converges to a CCE, then the CCE is co-undominated, which avoids outcomes with dominated strategies as illustrated in Figure \ref{fig:cce_standard}. 
\citet{deng2022nash} and \citet{feng_convergence_2021} analyzed the convergence of mean-based algorithms in first-price and second-price auctions, and they found convergence in the complete-information games. We focus on Bertrand competitions, which differ due to the various demand models. 

Below, we provide proof that mean-based algorithms converge to strategy profiles in the correlated rationalizable set. 
We analyze time-average and last-iterate convergence \citep{anagnostides2024convergence}. While the former measures the fraction of times that all players have chosen actions from a reference set, the latter makes a statement about the current strategy of the players.

\begin{definition}[Time-average Convergence]
	A sequence of action profiles $\{a_s\}_{t = 1}^{\infty}$ converges in time-average to a set of action profiles $\bar{\Acal} \subseteq \Acal$ if
	$\lim_{t \to \infty} \frac{1}{t} \sum_{s = 1}^t \I[a_{s} \in \bar{\Acal}] = 1$.
	The sequence converges in time-average to a single action profile $\bar a$ if $\bar \Acal = \{ \bar a \}$.
\end{definition}

\begin{definition}[Last-iterate Convergence]
	A sequence of action profiles $\{a_s\}_{t = 1}^{\infty}$ converges in last-iterate to a set of action profiles $\bar{\Acal} \subseteq \Acal$ if 
	$\lim_{t \to \infty} \mathbb{P}_t(a_t \in \bar{\Acal}) = 1$.
	The sequence converges in last-iterate to a single action profile $\bar a$ if $\bar \Acal = \{ \bar a \}$.
\end{definition}


With these definitions at hand, we state our central convergence theorem. For brevity of notation, we restrict ourselves to an informal formulation and provide a mathematically thorough version in our Appendix \ref{sec:formal-theorem}. 
The proof, building on techniques introduced by \citet{feng_convergence_2021} and \citet{deng2022nash}, along with the referenced propositions below, can also be found in Appendix \ref{sec:proof of the theorems}. 

\begin{reftheorem}{Theorem \ref{thm:mean-based-convergence} (Informal version)}
	\label{thm:mean-based-convergence-informal}
	If all players use mean-based algorithms in a finite normal-form game, their empirical frequency of play converges almost surely to the set of correlated rationalizable actions (time-average convergence). This also holds if players $i$ observe $U_{it}$ at time $t$, which is a stochastic version of the average reward $u_i$, as long as the random variables $U_{it}(a_i, a_{-i})$ are bounded, sampled independently for each time $t$, and have expected value $u_i$. Moreover, agents may enter the game at different points in time $\tau_i$.
\end{reftheorem}

This theorem states that players will eventually play correlated rationalizable actions in almost every round. But we can also say something about the actual (mixed) strategy profile of the players, i.e., about last-iterate convergence. 
We formalize this result in the corollary below and provide its proof in Appendix \ref{sec:last-iterate-convergence}.
\begin{corollary}[Last-iterate convergence]
	\label{cor:last-iterate-convergence}
	Under the same assumptions as in Theorem \ref{thm:mean-based-convergence}, the probability that the agents select correlated rationalizable actions approaches 1 for $t \to \infty$.
\end{corollary}
If the correlated rationalizable actions correspond to the unique Nash equilibrium, we thus get last-iterate convergence to the Nash equilibrium. 
Due to the relation between correlated rationalizable actions and serially undominated actions, we can also state the following corollary.
\begin{corollary}
	Mean-based algorithms converge to strategy profiles in SSU.
\end{corollary}

Next, we turn to the class of Bertrand competition games more specifically. Let us present our second main result informally.
\begin{inf_theorem}\label{thm:bertrand}
	In symmetric Bertrand competition games with standard demand, linear demand, and logit demand and discrete actions, the empirical play of mean-based algorithms converges almost surely to actions that are close to their (discrete) Nash equilibria, and the probability of playing such actions approaches 1.
\end{inf_theorem}

For this, we just need to show that the correlated rationalizable sets are close to the set of Nash equilibria for the three demand models, which we do in Appendix \ref{app:correlated-rationalizable-bertrand}. The proofs also show that, at least for standard and linear demand, the discrete Nash equilibria are close to their continuous counterparts.

\section{Experimental Analysis}
\label{sec:experimental-analysis}

In our experimental analysis, we analyze the behavior of various bandit algorithms in Bertrand oligopolies with different demand models. The focus variables of our experiments are the converged prices and the convergence behavior. 
We focus on scenarios that our theoretical contribution does not cover. We use algorithms that are well-known but (with the exception of Exp3) not known to be mean-based, and we simulate their behavior in a variety of Bertrand oligopolies. 

In our numerical experiments, most algorithms consistently reach Nash equilibrium prices within a few iterations, indicating that even without tens of thousands of repeated interactions and without the mean-based property, supra-competitive prices are rarely established. By including different demand functions and both symmetric and asymmetric environments, we show that our findings are robust across a wide range of treatments.

\subsection{Selected Algorithms}

In our study, we include $\epsilon$-greedy, UCB, Thompson Sampling, and Exp3, which are widely used and cited in the literature \citep{bubeck2011introduction}. 
All bandit algorithms share similarities. First, an action is selected according to some potentially random action-selection rule that is based on past experiences. Then, the agent receives a corresponding reward feedback with which it subsequently updates its internal state. We describe this generic scheme in the Algorithm \ref{alg:bandit-feedback-algorithm} below. All algorithms that we analyze follow this generic scheme while implementing distinct rules for selecting actions and updating their beliefs. Details of the algorithms can be found in appendix \ref{app:algorithms}.


\begin{algorithm}
	\caption{Bandit Algorithm Scheme}\label{alg:bandit-feedback-algorithm}
	\begin{algorithmic}
		\State Initialization: Algorithm parameters $\theta$
		\For{$t = 1, 2, \dots$}
		\State $a_t \leftarrow$ \texttt{get\_action($\theta$, $t$)} \hfill \textit{// Agents choose action.}
		\State $u_t \leftarrow$ \texttt{environment($a_t$)} \hfill \textit{// Agents submit actions and observe own utility.}
		\State $\theta \leftarrow$ \texttt{update($\theta$, $a_t$, $u_t$, $t$)} \hfill \textit{// Agents update their strategy.}
		\EndFor
	\end{algorithmic}
\end{algorithm}

\subsection{Experiment Setup}
\label{sec:experiment-setup}

We demonstrate that even in treatments where no analytical convergence results are available, convergence to competitive prices can be observed consistently with many bandit algorithms. \textit{Supra-competitive ("collusive") prices}, which are prices between the Nash equilibrium and the joint profit-maximizing prices, only occur under symmetric combinations of certain algorithms.
We characterize supra-competitive prices via the following metrics. Let $(\bar{a}_i)_{i = 1}^{\n}$ be the largest component strategies of any pure Nash equilibrium, and let $(a^M_i)_{i = 1}^{\n}$ be the actions played in the monopoly. We define the \textit{price-competition index} of player $i$ and the \textit{profit-competition index} of all players by 
\begin{equation}
	CI_{price, i}(a_i) = \frac{a_i - \bar{a}_i}{a^M_i - \bar{a}_i}, ~ i \in [\n]
	\qquad \text{and} \qquad 
	CI_{profit}(a) = \frac{\sum_{i = 1}^{\n} ( u_i(a) - u_i(\bar{a}) ) }{\sum_{i = 1}^{\n} ( u_i(a^M) - u_i(\bar{a}) ) }.
\end{equation}
A price-competition index of zero captures that player $i$ plays $\bar{a}_i$, and thus is competitive. A value of one means that the agent plays its monopoly price, indicating non-competitive pricing. The profit-competition index represents the fraction of monopolistic payoff surplus that all agents could achieve in total.

Let us briefly discuss the choice of metrics. In essence, we want to capture a deviation from competitive play (Nash equilibrium) to cooperation. A Nash equilibrium naturally lends itself as the lower bound of our metric's values, but the choice of the upper deserves discussion. As noted by \citet{loots_datadriven_2023}, multiple notions of collusive prices are reasonable, and joint profit maximization might not always be the adequate choice. In line with previous research \citep{Calvano.2020}, we nonetheless select these prices as our upper reference points, because in our analyzed settings, joint profit maximization is beneficial to all firms. Being a joint maximum, it is Pareto-optimal for sellers, too. As indicated by \citet{loots_datadriven_2023}, joint profit maximization also exhibits the highest price increase and the worst effect on consumer welfare, on average.

In the following, we introduce the treatment variables of our experiments, such as the oligopoly configurations and the bandit algorithms, as well as the experiment setup. We then state the conclusions drawn from our experiments and present the empirical evidence supporting them.

In our experiments, we vary the games that the agents compete in, the algorithms they use, and the number of competitors they face. An experiment is formed by a unique combination of these factors. 
Table \ref{tab:experiment-treatment-variables} provides an overview of the treatment and focus variables. Below, we describe the individual treatments in more detail.

\begin{table}[H]
	\centering
	
    \begin{tabular}{l l}
        \hline
        \textbf{Treatment variables} & \textbf{Variations} \\
        \hline
        Oligopoly model & Demand models (standard, linear, logit), Symmetric/asymmetric parameterization \\
        Algorithms & UCB-T, Exp3-$\epsilon$, $\epsilon$-Greedy, TS, (UCB1, Q-Learning, MWU)\\
        Number of competitors & 2 - 10 \\
        \hline
        \textbf{Focus variables} & \textbf{Description} \\
        \hline
        Prices & Prices after convergence \\
        Convergence speed & Time to convergence and convergence behavior \\
        \hline
    \end{tabular}
    \caption{Treatment and Focus Variables of the Experiments}
    \label{tab:experiment-treatment-variables}
\end{table}

\paragraph{Oligopoly Model.}
Our oligopolies model the demand and the costs that firms face. They are based on the games described in Section \ref{sec:bertrand-oligopoly}. We implemented models with \ref{eq:demand_standard}, \ref{eq:linear_demand}, and \ref{eq:logit_demand} demand functions, where the games with linear and logit demand are inspired by \citep{hansen_frontiers_2021} and \citep{Calvano.2020}, respectively. Depending on the parameter choices, the utility functions are the same for all agents (symmetric game) or not (asymmetric game). Our configurations are summarized in Table \ref{tab:bertrand-parameters}.
The asymmetric settings (\textit{O2'}, \textit{O3'}) are a step towards more realistic scenarios where firms might face different costs and demand structures. 
The price ranges are chosen to span the Nash equilibrium prices and the prices that maximize the joint profit, including a small margin. The agents can choose from 21 evenly distributed prices within this range. 
For example, an agent in \textit{O3} can select actions from the set $\{1.00, 1.05, 1.10, \dots, 1.95, 2.00\}$. 
Other discretizations, and in particular asymmetric discretizations for the players, did not substantially change our results in preliminary experiments.

\begin{table}[h]
	\centering
	\footnotesize
	
	\begin{tabular}{c | c c c c | c c}
		\toprule
		Model & Demand & Costs $(c_1, c_2)$ & Other parameters & Price range & Nash & Monopoly \\
		\midrule
		\textbf{O1} & standard & $(0.0, 0.0)$ & - & $[0.05, 1.00]$ & $(0.05, 0.05)$ & $(0.48, 0.48)$ \\
		\textbf{O2} & linear & $(0.0, 0.0)$ & $\alpha = 0.48$, $\beta = 0.9$, $\gamma = 0.6$ & $[0.00, 1.00]$ & $(0.40, 0.40)$ & $(0.80, 0.80)$ \\
		\textbf{O3} & logit & $(1.0, 1.0)$ & $a_1 = a_2 = 2.0$, $a_0 = 0.0$, $\mu = 0.25$  & $[1.00, 2.00]$ & $(1.50, 1.50)$ & $(1.90, 1.90)$ \\
		\textbf{O2'} & linear & $(0.0, 0.2)$ & $\alpha = 0.48$, $\beta = 0.9$, $\gamma = 0.6$ & $[0.00, 1.00]$ & $(0.45, 0.50)$ & $(0.80, 0.90)$  \\
		\textbf{O3'} & logit & $(0.5, 1.0)$ & $a_1 = 1.5$, $a_2 = 2.0$, $\mu = 0.25$, $a_0 = 0.0$ & $[1.00, 2.00]$ & $(0.95, 1.48)$ & $(1.40, 1.93)$  \\
		\bottomrule
	\end{tabular}
	
	\caption{Parameterization of the Bertrand oligopoly models.}
	\label{tab:bertrand-parameters}
	
	\raggedright \small
	The (pure) Nash equilibria and monopoly prices are stated for the discretized game. In case of multiple equilibria, the maximum prices were selected.
\end{table}

\paragraph{Algorithms.}
We investigate the behavior of four widespread bandit algorithms. First, we evaluate a variant of the UCB-Tuned (\textit{UCB-T}) algorithm that resolves ties randomly and does not eliminate actions. Next, we run experiments with the simple $\epsilon$-greedy algorithm (\textit{$\epsilon$-Greedy}) and an Exp3 (\textit{Exp3-$\epsilon$}) algorithm with fixed exploration rate $\epsilon$. We note that this version of Exp3 is not mean-based. Known mean-based versions of Exp3 suffer from a very slowly decaying exploration rate, which makes them impractical for experiments and applications. Finally, we implemented an algorithm (\textit{TS}) that performs Thompson Sampling with Gaussian priors, likelihoods, and posteriors. We note that our experiment results are stable against reasonable modifications of the parameters. In one result on mean-based algorithms, we also analyze Multiplicative Weights Update (MWU) and the mean-based version of Exp3 introduced by \citet{braverman_selling_2018}. In the appendix, we also report results on Q-learning to provide a comparison. However, given the multitude of hyperparameters available in this algorithm, we restrict our analysis to a few experiments replicating prior results. 

\paragraph{Number of competitors.}
Some of our empirical results are based on the number of firms that compete in an oligopoly. The symmetric games (\textit{O1}, \textit{O2}, \textit{O3}) can be easily extended to an arbitrary number of players. We investigate numbers between two and ten for an oligopoly based on \textit{O2}.

\vspace{\baselineskip}

We run each experiment ten times with random seeds from 0 to 9. Fixing the seeds allows us to reproduce results despite the random nature of the algorithms. Each run consists of 250,000 iteration steps during which the agents first submit their actions simultaneously and independently to the market. Then, the demands are evaluated, and the agents observe their respective rewards. The agents update their beliefs at the end of each step. We found that almost all of our experiments converged within the provided time frame.

\subsection{Results}

Let us now report the main results of our experimental analysis. 

\subsubsection{Convergence Speed with Mean-based Algorithms}

Let us first focus on markets with mean-based algorithms where our convergence guarantee applies. Speed of convergence is unknown, however. We consider the mean-based versions of MWU and Exp3 that have been described by \cite{braverman_selling_2018}. 
Figure \ref{fig:price-evolution-mb} shows the evolution of the price-competition index over the course of the experiments. We observe convergence to very low index values. While the MWU algorithm reaches the final outcome within a few iterations, the Exp3 algorithm requires noticeably more iterations. This is expected as MWU has access to full feedback while Exp3 only receives bandit feedback. Below, we use a different version of Exp3 that is faster than its mean-based counterpart.

\begin{figure}
	\centering
	
	\hfill
	\begin{subfigure}{0.44\textwidth}
		\includegraphics[width=\textwidth]{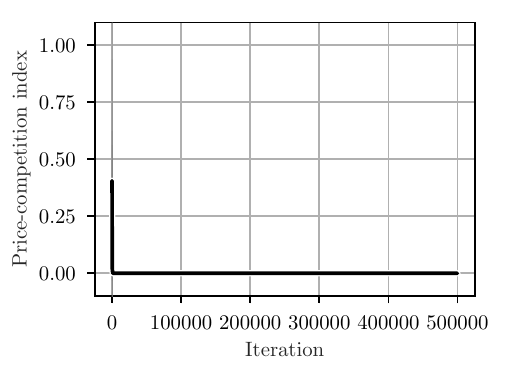}
		\subcaption{Multiplicative Weights}
	\end{subfigure}
	\hfill
	\begin{subfigure}{0.44\textwidth}
		\includegraphics[width=\textwidth]{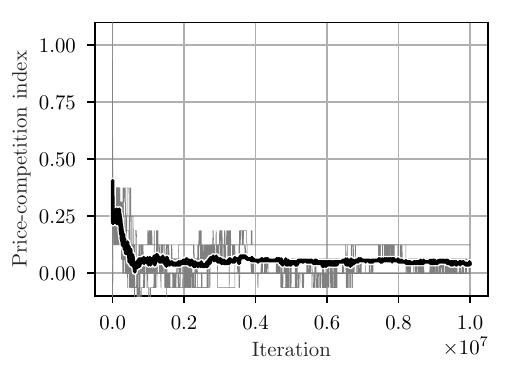}
		\subcaption{Exp3 (mean-based)}
	\end{subfigure}
	\hfill
	
	\caption{Evolution of Prices During Training for Mean-based Algorithms}
	\label{fig:price-evolution-mb}
	
	\centering \small
	The experiments are based on the symmetric environment with linear demand (\textit{O2}). The figures display the competition indices based on the charged prices (running medians over 1000 steps) for ten runs of the experiment in thin lines. The thick line shows the average of all runs and thus indicates the overall trend.
\end{figure}

\subsubsection{Algorithmic Pricing in Duopolies}

Next, we analyze duopolies with different algorithms that are not mean-based anymore. 
We first report the prices at the end of the experiments in two-player games. After 250,000 iterations, the algorithms consistently settle at a price level (see next section), so we can assume convergence. 
In Appendix \ref{sec:convergence-speed}, we provide an analysis of the convergence speed.
Our first result makes a statement about the price levels agents eventually reach.

\begin{result}
	Supra-competitive pricing only evolves with UCB algorithms under self-play. Other algorithms consistently price close to the Nash equilibrium prices. In particular, combinations of diverse algorithms rarely display non-competitive behavior. This happens for all demand models (standard, linear, logit) and for symmetric and asymmetric settings.
\end{result}

We visualize these results in Figure \ref{fig:collusion-in-duopolies}, which displays the profit- and price-competition indices found in our experiments. We average the indices over settings, runs, and - if applicable - over agents. In settings with two identical competing algorithms (diagonal entries), we only observe non-competitive behavior with the \textit{UCB-T} algorithm. In this case, prices \textit{and} profits are high, meaning that cooperating algorithms indeed benefit from their non-competitive play. The combined settings of any two different algorithms show copmetition indices close to zero, with a combination of \textit{UCB-T} and \textit{$\epsilon$-Greedy} resulting in only slightly supra-competitive prices.

Our experiments also demonstrate that non-competitive behavior with Bandit algorithms does not develop in a stable manner. While we could observe high prices with \textit{UCB-T} for all demand functions, and in symmetric as well as asymmetric games, these experiments usually suffer from a high variance between runs, making the result unpredictable in advance. In contrast, competitive algorithms performed consistently in all environments and all runs.

\begin{figure}
	\centering
	
	\hfill
	\begin{subfigure}{0.4\textwidth}
		\includegraphics[width=\textwidth]{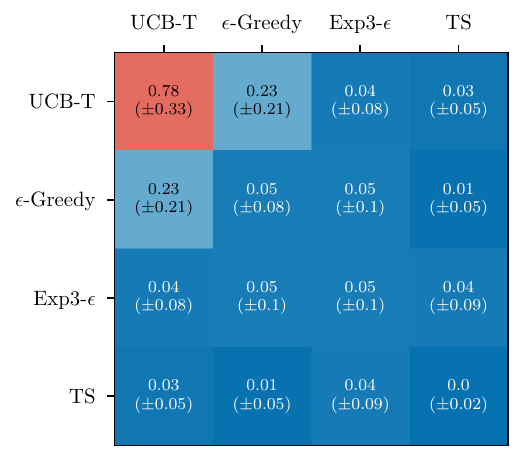}
		\subcaption{Price-competition indices \label{fig:price-collusion-in-duopolies}}
	\end{subfigure}
	\hfill
	\begin{subfigure}{0.4\textwidth}
		\includegraphics[width=\textwidth]{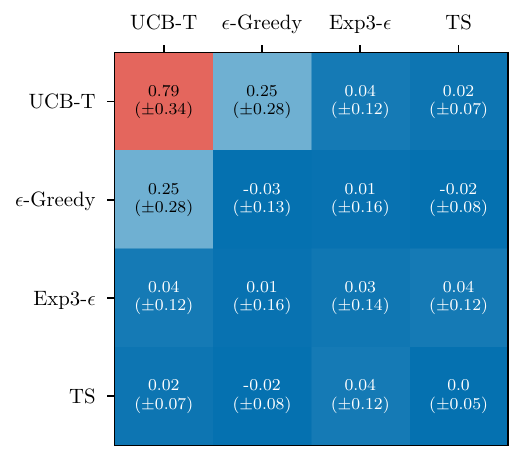}
		\subcaption{Profit-competition indices \label{fig:profit-collusion-in-duopolies}}
	\end{subfigure}
	\hfill
	
	\caption{Pricing in Duopolies}
	\label{fig:collusion-in-duopolies}
	
	\centering \small
	We visualize the competition indices, based on a) the median prices and b) the mean profits in the last 1000 steps, of several combinations of algorithms. The displayed numbers are the numeric values of the competition indices; the values in brackets display their corresponding standard deviation. Averages and standard deviations were computed over two agents, ten runs, and the five settings \textit{O1} - \textit{O3'}.
\end{figure}

While experiments with two players in simplified oligopolies are useful to extract the essence of algorithmic interaction, we also investigated extensions of our games with more players and richer demand functions. We provide our results in Appendix \ref{sec:additional-experiments}.

\section{Conclusions}
\label{sec:conclusions}

The discussion of algorithmic collusion is largely based on the experimental results of specific algorithms and versions of the Bertrand competition. In general, learning algorithms do not converge to an equilibrium. However, we prove that important algorithms do so in repeated Bertrand pricing competition with all-or-nothing, linear, and logit demand models. The convergence of mean-based algorithms to correlated rationalizable strategies is of independent interest and might be useful for the analysis of different games as well. The result closes a gap in the literature on learning in games, which largely focused on algorithms with no internal or external regret, which converge to correlated or coarse correlated equilibria, respectively. Mean-based algorithms such as Exp3 only require bandit feedback after each round, a realistic assumption for algorithms used in algorithmic pricing. The fact that correlated rationalizable strategy profiles coincide with the set of Nash equilibria in Bertrand competition games is an important insight, and it shows that such games can be learned even with very simple algorithms relevant in practice.

In our experiments, persistent supra-competitive pricing is rare for the bandit algorithms we study and appears mainly in special symmetric installations, such as symmetric UCB among a few sellers. Whether such outcomes should be classified as collusion in the stricter legal or economic sense depends on the underlying mechanism and on whether multiple firms’ choices are implicated in sustaining the outcome.

Such insights are important for regulators when they want to identify non-competitive behavior in online markets. While multi-armed bandit algorithms are an important and widely used class of algorithms used for algorithmic pricing, one cannot guarantee that sellers don't use other algorithms. In future research, it will be valuable to analyze and classify alternative algorithms and models. The properties that we showed for the Bertrand competition will also be useful for subsequent studies.

\section*{Acknowledgements}

This project has received funding from the European Research Council (ERC) under the European Union’s Horizon Europe research and innovation programme (grant agreement No 101198689).\\    
This project was funded by the Deutsche Forschungsgemeinschaft (DFG, German Research Foundation) - GRK 2201/2 - Project Number 277991500 and BI 1057/9.

\vfill
\pagebreak


\bibliographystyle{ACM-Reference-Format}
\bibliography{literature.bib}

@article{abada_artificial_2023,
	title = {Artificial {Intelligence}: {Can} {Seemingly} {Collusive} {Outcomes} {Be} {Avoided}?},
	volume = {69},
	OPT_issn = {0025-1909},
	shorttitle = {Artificial {Intelligence}},
	OPT_url = {https://pubsonline.informs.org/doi/full/10.1287/mnsc.2022.4623},
	OPT_doi = {10.1287/mnsc.2022.4623},
	number = {9},
	urldate = {2024-09-25},
	journal = {Management Science},
	author = {Abada, Ibrahim and Lambin, Xavier},
	month = sep,
	year = {2023},
	pages = {5042--5065},
}

@article{mertikopoulos2024unified,
  title={A unified stochastic approximation framework for learning in games},
  author={Mertikopoulos, Panayotis and Hsieh, Ya-Ping and Cevher, Volkan},
  journal={Mathematical Programming},
  volume={203},
  number={1},
  pages={559--609},
  year={2024},
  publisher={Springer}
}

@article{meylahn_learning_2022,
	title = {Learning to {Collude} in a {Pricing} {Duopoly}},
	volume = {24},
	OPT_issn = {1523-4614},
	OPT_url = {https://pubsonline.informs.org/doi/abs/10.1287/msom.2021.1074},
	OPT_doi = {10.1287/msom.2021.1074},
	number = {5},
	urldate = {2023-06-27},
	journal = {Manufacturing \& Service Operations Management},
	author = {Meylahn, Janusz M. and V. den Boer, Arnoud},
	month = {sep},
	year = {2022},
	pages = {2577--2594}
}

@article{lambin2024less,
  title={Less than meets the eye: simultaneous experiments as a source of algorithmic seeming collusion},
  author={Lambin, Xavier},
  journal={Available at SSRN 4498926},
  year={2024}
}

@article{waltman2008q,
  title={Q-learning agents in a Cournot oligopoly model},
  author={Waltman, Ludo and Kaymak, Uzay},
  journal={Journal of Economic Dynamics and Control},
  volume={32},
  number={10},
  pages={3275--3293},
  year={2008},
  publisher={Elsevier}
}

@article{deng2024algorithmic,
  title={Algorithmic Collusion in Dynamic Pricing with Deep Reinforcement Learning},
  author={Deng, Shidi and Schiffer, Maximilian and Bichler, Martin},
  journal={Proceedings of Wirtschaftsinformatik 2024},
  year={2024}
}

@article{brown_competition_2023,
	title = {Competition in {Pricing} {Algorithms}},
	volume = {15},
	OPT_issn = {1945-7669},
	OPT_url = {https://www.aeaweb.org/articles?id=10.1257/mic.20210158},
	OPT_doi = {10.1257/mic.20210158},
	language = {en},
	number = {2},
	urldate = {2023-06-27},
	journal = {American Economic Journal: Microeconomics},
	author = {Brown, Zach Y. and MacKay, Alexander},
	month = {may},
	year = {2023},
	pages = {109--156}
}

@article{abada2024algorithmic,
  title={Algorithmic Collusion: Where Are We and Where Should We Be Going?},
  author={Abada, Ibrahim and Harrington Jr, Joseph E and Lambin, Xavier and Meylahn, Janusz M},
  journal={Available at SSRN 4891033},
  year={2024}
}

@article{abada_collusion_2024,
	title = {Collusion by mistake: {Does} algorithmic sophistication drive supra-competitive profits?},
	volume = {318},
	OPT_issn = {0377-2217},
	shorttitle = {Collusion by mistake},
	OPT_url = {https://www.sciencedirect.com/science/article/pii/S037722172400434X},
	OPT_doi = {10.1016/j.ejor.2024.06.006},
	number = {3},
	urldate = {2024-09-25},
	journal = {European Journal of Operational Research},
	author = {Abada, Ibrahim and Lambin, Xavier and Tchakarov, Nikolay},
	month = nov,
	year = {2024},
	pages = {927--953},
}

@article{hansen_frontiers_2021,
	title = {Frontiers: {Algorithmic} {Collusion}: {Supra}-competitive {Prices} via {Independent} {Algorithms}},
	volume = {40},
	OPT_issn = {0732-2399},
	shorttitle = {Frontiers},
	OPT_url = {https://pubsonline.informs.org/doi/abs/10.1287/mksc.2020.1276},
	OPT_doi = {10.1287/mksc.2020.1276},
	number = {1},
	urldate = {2023-06-27},
	journal = {Marketing Science},
	author = {Hansen, Karsten T. and Misra, Kanishka and Pai, Mallesh M.},
	month = {jan},
	year = {2021},
	pages = {1--12}
}

@article{klein_autonomous_2021,
	title = {Autonomous algorithmic collusion: {Q}-learning under sequential pricing},
	volume = {52},
	copyright = {© 2021 The Authors. The RAND Journal of Economics published by Wiley Periodicals LLC on behalf of The RAND Corporation},
	OPT_issn = {1756-2171},
	shorttitle = {Autonomous algorithmic collusion},
	OPT_url = {https://onlinelibrary.wiley.com/doi/abs/10.1111/1756-2171.12383},
	OPT_doi = {10.1111/1756-2171.12383},
	language = {en},
	number = {3},
	urldate = {2023-06-27},
	journal = {The RAND Journal of Economics},
	author = {Klein, Timo},
	year = {2021},
	pages = {538--558}
}

@misc{den_boer_artificial_2022,
	OPT_address = {Rochester, NY},
	type = {{SSRN} {Scholarly} {Paper}},
	title = {Artificial {Collusion}: {Examining} {Supracompetitive} {Pricing} by {Q}-{Learning} {Algorithms}},
	shorttitle = {Artificial {Collusion}},
	OPT_url = {https://papers.ssrn.com/abstract=4213600},
	OPT_doi = {10.2139/ssrn.4213600},
	language = {en},
	urldate = {2023-06-27},
	author = {den Boer, Arnoud V. and Meylahn, Janusz M. and Schinkel, Maarten Pieter},
	month = {dec},
	year = {2022}
}

@misc{calvano_algorithmic_2021,
	OPT_address = {Rochester, NY},
	type = {{SSRN} {Scholarly} {Paper}},
	title = {Algorithmic {Collusion}, {Genuine} and {Spurious}},
	OPT_url = {https://papers.ssrn.com/abstract=3928672},
	language = {en},
	urldate = {2023-06-27},
	author = {Calvano, Emilio and Calzolari, Giacomo and Denicolò, Vincenzo and Pastorello, Sergio},
	month = {jul},
	year = {2021}
}

@article{calvano_algorithmic_2019,
	title = {Algorithmic {Pricing} {What} {Implications} for {Competition} {Policy}?},
	volume = {55},
	OPT_issn = {1573-7160},
	OPT_url = {https://doi.org/10.1007/s11151-019-09689-3},
	OPT_doi = {10.1007/s11151-019-09689-3},
	language = {en},
	number = {1},
	urldate = {2023-09-01},
	journal = {Review of Industrial Organization},
	author = {Calvano, Emilio and Calzolari, Giacomo and Denicolò, Vincenzo and Pastorello, Sergio},
	month = {aug},
	year = {2019},
	pages = {155--171}
}

@inproceedings{kolumbus_auctions_2022,
	OPT_address = {New York, NY, USA},
	series = {{WWW} '22},
	title = {Auctions between {Regret}-{Minimizing} {Agents}},
	OPT_isbn = {978-1-4503-9096-5},
	OPT_url = {https://dl.acm.org/doi/10.1145/3485447.3512055},
	OPT_doi = {10.1145/3485447.3512055},
	urldate = {2023-10-17},
	booktitle = {Proceedings of the {ACM} {Web} {Conference} 2022},
	publisher = {Association for Computing Machinery},
	author = {Kolumbus, Yoav and Nisan, Noam},
	month = {apr},
	year = {2022},
	pages = {100--111}
}

@article{pearce_rationalizable_1984,
	title = {Rationalizable {Strategic} {Behavior} and the {Problem} of {Perfection}},
	volume = {52},
	OPT_issn = {0012-9682},
	OPT_url = {https://www.jstor.org/stable/1911197},
	OPT_doi = {10.2307/1911197},
	number = {4},
	urldate = {2024-07-31},
	journal = {Econometrica},
	author = {Pearce, David G.},
	year = {1984},
	pages = {1029--1050}
}

@article{bernheim_rationalizable_1984,
	title = {Rationalizable {Strategic} {Behavior}},
	volume = {52},
	OPT_issn = {0012-9682},
	OPT_url = {https://www.jstor.org/stable/1911196},
	OPT_doi = {10.2307/1911196},
	number = {4},
	urldate = {2024-07-31},
	journal = {Econometrica},
	author = {Bernheim, B. Douglas},
	year = {1984},
	pages = {1007--1028}
}

@inproceedings{auer_gambling_1995,
	title = {Gambling in a rigged casino: {The} adversarial multi-armed bandit problem},
	shorttitle = {Gambling in a rigged casino},
	OPT_doi = {10.1109/SFCS.1995.492488},
	booktitle = {Proceedings of {IEEE} 36th {Annual} {Foundations} of {Computer} {Science}},
	author = {Auer, P. and Cesa-Bianchi, N. and Freund, Y. and Schapire, R.E.},
	month = {oct},
	year = {1995},
	pages = {322--331}
}

@misc{russo_tutorial_2020,
	title = {A {Tutorial} on {Thompson} {Sampling}},
	OPT_url = {http://arxiv.org/abs/1707.02038},
	OPT_doi = {10.48550/arXiv.1707.02038},
	urldate = {2023-09-21},
	publisher = {arXiv},
	author = {Russo, Daniel and Van Roy, Benjamin and Kazerouni, Abbas and Osband, Ian and Wen, Zheng},
	month = {jul},
	year = {2020}
}

@article{thompson_likelihood_1933,
	title = {On the {Likelihood} that {One} {Unknown} {Probability} {Exceeds} {Another} in {View} of the {Evidence} of {Two} {Samples}},
	volume = {25},
	OPT_issn = {0006-3444},
	OPT_url = {https://www.jstor.org/stable/2332286},
	OPT_doi = {10.2307/2332286},
	number = {3/4},
	urldate = {2023-09-28},
	journal = {Biometrika},
	author = {Thompson, William R.},
	year = {1933},
	pages = {285--294}
}

@article{thompson_theory_1935,
	title = {On the {Theory} of {Apportionment}},
	volume = {57},
	OPT_issn = {0002-9327},
	OPT_url = {https://www.jstor.org/stable/2371219},
	OPT_doi = {10.2307/2371219},
	number = {2},
	urldate = {2023-09-28},
	journal = {American Journal of Mathematics},
	author = {Thompson, William R.},
	year = {1935},
	pages = {450--456}
}

@article{bertrand1883theorie,
	title = {Th{\'e}orie Math{\'e}matique de La Richesse Sociale},
	author = {Bertrand, Joseph},
	year = {1883},
	journal = {Journal des Savants},
	volume = {67},
	number = {1883},
	pages = {499--508},
	publisher = {{Paris}}
}

@inproceedings{trovo2015multi,
	title = {Multi-armed bandit for pricing},
	author = {Trovo, Francesco and Paladino, Stefano and Restelli, Marcello and Gatti, Nicola and others},
	booktitle = {Proceedings of the 12th European Workshop on Reinforcement Learning},
	pages = {1--9},
	year = {2015}
}

@misc{levin2006solution,
	type = {Notes},
	title = {Solution {Concepts}},
	OPT_url = {https://web.stanford.edu/~jdlevin/Econ%20286/Solution%20Concepts.pdf},
	language = {en},
	urldate = {2024-09-23},
	author = {Levin, Jonathan},
	month = apr,
	year = {2006},
}

@article{brandenburger1987rationalizability,
	title = {Rationalizability and correlated equilibria},
	author = {Brandenburger, Adam and Dekel, Eddie},
	journal = {Econometrica: Journal of the Econometric Society},
	pages = {1391--1402},
	year = {1987},
	publisher = {JSTOR}
}

@book{sandholm2010population,
	OPT_address = {Cambridge, Mass.},
	series = {Economic learning and social evolution},
	title = {Population games and evolutionary dynamics},
	OPT_isbn = {978-0-262-19587-4},
	language = {eng},
	publisher = {MIT Press},
	author = {Sandholm, William H.},
	year = {2010},
}

@book{topkis1998supermodularity,
	OPT_address = {Princeton, N.J},
	series = {Frontiers of economic research},
	title = {Supermodularity and complementarity},
	OPT_isbn = {978-0-691-03244-3},
	language = {eng},
	publisher = {Princeton Univ. Press},
	author = {Topkis, Donald M.},
	year = {1998},
}

@book{vives1999oligopoly,   
	OPT_address = {Cambridge, Mass.},
	title = {Oligopoly Pricing: Old Ideas and New Tools},
	OPT_isbn = {978-0-262-72040-3 978-0-262-22060-6},
	shorttitle = {Oligopoly pricing},
	language = {eng},
	publisher = {MIT Press},
	author = {Vives, Xavier},
	year = {2001},
}

@article{bulow1985multimarket,
	title = {Multimarket oligopoly: Strategic substitutes and complements},
	author = {Bulow, Jeremy I and Geanakoplos, John D and Klemperer, Paul D},
	journal = {Journal of Political economy},
	volume = {93},
	number = {3},
	pages = {488--511},
	year = {1985},
	publisher = {The University of Chicago Press}
}

@article{rothschild1974two,
	title = {A two-armed bandit theory of market pricing},
	author = {Rothschild, Michael},
	journal = {Journal of Economic Theory},
	volume = {9},
	number = {2},
	pages = {185--202},
	year = {1974},
	publisher = {Elsevier}
}

@article{qu24,
	author = {Qu, Jiaming},
	year = {2024},
	month = {02},
	pages = {160-165},
	title = {Survey of dynamic pricing based on Multi-Armed Bandit algorithms},
	volume = {37},
	journal = {Applied and Computational Engineering},
	OPT_doi = {10.54254/2755-2721/37/20230497}
}

@article{mueller2019low,
	title = {Low-rank bandit methods for high-dimensional dynamic pricing},
	author = {Mueller, Jonas W and Syrgkanis, Vasilis and Taddy, Matt},
	journal = {Advances in Neural Information Processing Systems},
	volume = {32},
	year = {2019}
}

@article{den2015dynamic,
	title = {Dynamic pricing and learning: historical origins, current research, and new directions},
	author = {den Boer, Arnoud V},
	journal = {Surveys in operations research and management science},
	volume = {20},
	number = {1},
	pages = {1--18},
	year = {2015},
	publisher = {Elsevier}
}

@article{bauer2018optimal,
	title = {Optimal pricing in e-commerce based on sparse and noisy data},
	author = {Bauer, Josef and Jannach, Dietmar},
	journal = {Decision support systems},
	volume = {106},
	pages = {53--63},
	year = {2018},
	publisher = {Elsevier}
}

@misc{taywade_multi-armed_2023,
	title = {Multi-armed {Bandit} {Algorithms} for {Cournot} {Games}},
	OPT_url = {https://www.researchsquare.com/article/rs-2928787/v1},
	OPT_doi = {10.21203/rs.3.rs-2928787/v1},
	urldate = {2024-09-25},
	author = {Taywade, Kshitija and Goldsmith, Judy and Harrison, Brent and Bagh, Adib},
	month = may,
	year = {2023},
	OPT_issn = {2693-5015},
    note = {under review}
}

@article{elreedy2021novel,
	title = {Novel pricing strategies for revenue maximization and demand learning using an exploration--exploitation framework},
	author = {Elreedy, Dina and Atiya, Amir F and Shaheen, Samir I},
	journal = {Soft Computing},
	volume = {25},
	number = {17},
	pages = {11711--11733},
	year = {2021},
	publisher = {Springer}
}

@article{anagnostides2024convergence,
	title = {On the convergence of no-regret learning dynamics in time-varying games},
	author = {Anagnostides, Ioannis and Panageas, Ioannis and Farina, Gabriele and Sandholm, Tuomas},
	journal = {Advances in Neural Information Processing Systems},
	volume = {36},
	year = {2024}
}

@misc{den2023mathematical,
	OPT_address = {Rochester, NY},
	type = {{SSRN} {Scholarly} {Paper}},
	title = {Algorithmic {Collusion}: {A} {Mathematical} {Definition} and {Research} {Agenda} for the {OR}/{MS} {Community}},
	shorttitle = {Algorithmic {Collusion}},
	OPT_url = {https://papers.ssrn.com/abstract=4636488},
	OPT_doi = {10.2139/ssrn.4636488},
    language = {en},
	urldate = {2024-09-25},
	author = {den Boer, Arnoud V.},
	month = nov,
	year = {2023},
}

@book{fudenberg1991game,
	OPT_address = {Cambridge (Mass.)},
	title = {Game theory},
	OPT_isbn = {978-0-262-06141-4},
	language = {eng},
	publisher = {the MIT press},
	author = {Fudenberg, Drew and Tirole, Jean},
	year = {1991},
}

@article{auer2002bandit,
	title = {The Nonstochastic Multiarmed Bandit Problem},
	author = {Auer, Peter and Cesa-Bianchi, Nicol\`{o} and Freund, Yoav and Schapire, Robert E.},
	year = {2002},
	journal = {SIAM Journal on Computing},
	volume = {32},
	number = {1},
	pages = {48--77},
	OPT_doi = {10.1137/S0097539701398375},
	OPT_url = {https://doi.org/10.1137/S0097539701398375},
	eprint = {https://doi.org/10.1137/S0097539701398375}
}

@article{bertrand1883book,
	title = {Book review of theorie mathematique de la richesse social and of recherches sur les principes mathematiques de la theorie des richesses},
	author = {Bertrand, Joseph},
	year = {1883},
	journal = {Journal des Savants}
}

@inproceedings{braverman_selling_2018,
	OPT_address = {New York, NY, USA},
	series = {{EC} '18},
	title = {Selling to a {No}-{Regret} {Buyer}},
	OPT_isbn = {978-1-4503-5829-3},
	OPT_url = {https://dl.acm.org/doi/10.1145/3219166.3219233},
	OPT_doi = {10.1145/3219166.3219233},
	urldate = {2024-01-23},
	booktitle = {Proceedings of the 2018 {ACM} {Conference} on {Economics} and {Computation}},
	publisher = {Association for Computing Machinery},
	author = {Braverman, Mark and Mao, Jieming and Schneider, Jon and Weinberg, Matt},
	month = jun,
	year = {2018},
	pages = {523--538},
}

@article{brown1951IterativeSolutionGames,
	title = {Iterative Solution of Games by Fictitious Play},
	author = {Brown, George W},
	year = {1951},
	journal = {Activity analysis of production and allocation},
	publisher = {{New York}},
	volume = {13},
	number = {1},
	pages = {374--376},
	ids = {brown1951iterative,brownIterativeSolutionGames1951}
}

@article{Calvano.2020,
	title = {{Artificial Intelligence, Algorithmic Pricing, and Collusion}},
	author = {Calvano, Emilio and Calzolari, Giacomo and Denicol{\`o}, Vincenzo and Pastorello, Sergio},
	year = {2020},
	month = {oct},
	journal = {{American Economic Review}},
	volume = {110},
	number = {10},
	pages = {3267--3297},
	OPT_doi = {10.1257/aer.20190623},
	OPT_issn = {0002-8282},
	OPT_url = {https://www.aeaweb.org/articles?id=10.1257/aer.20190623},
	urldate = {2023-06-27},
	language = {en}
}

@article{Calvano.2020b,
	title = {{Protecting consumers from collusive prices due to AI}},
	author = {Calvano, Emilio and Calzolari, Giacomo and Denicol{\`o}, Vincenzo and Harrington, Joseph E. and Pastorello, Sergio},
	year = {2020},
	journal = {{Science}},
	volume = {370},
	number = {6520},
	pages = {1040--1042},
	OPT_doi = {10.1126/science.abe3796}
}

@article{kastius2022dynamic,
	title = {Dynamic pricing under competition using reinforcement learning},
	author = {Kastius, Alexander and Schlosser, Rainer},
	journal = {Journal of Revenue and Pricing Management},
	volume = {21},
	number = {1},
	pages = {50--63},
	year = {2022},
	publisher = {Springer}
}

@article{rana2014real,
	title = {Real-time dynamic pricing in a non-stationary environment using model-free reinforcement learning},
	author = {Rana, Rupal and Oliveira, Fernando S},
	journal = {Omega},
	volume = {47},
	pages = {116--126},
	year = {2014},
	publisher = {Elsevier}
}

@inproceedings{chen2016empirical,
	title = {An empirical analysis of algorithmic pricing on amazon marketplace},
	author = {Chen, Le and Mislove, Alan and Wilson, Christo},
	year = {2016},
	month = {apr},
	booktitle = {Proceedings of the 25th international conference on World Wide Web},
	publisher = {{International World Wide Web Conferences Steering Committee}},
	OPT_address = {Republic and Canton of Geneva, Switzerland},
	series = {{WWW} '16},
	pages = {1339--1349},
	OPT_doi = {10.1145/2872427.2883089},
	OPT_isbn = {978-1-4503-4143-1},
	OPT_url = {https://dl.acm.org/doi/10.1145/2872427.2883089},
	urldate = {2023-09-01}
}

@article{daskalakis2009complexity,
	title = {The complexity of computing a Nash equilibrium},
	author = {Daskalakis, Constantinos and Goldberg, Paul W and Papadimitriou, Christos H},
	year = {2009},
	month = {jan},
	journal = {SIAM Journal on Computing},
	publisher = {SIAM},
	volume = {39},
	number = {1},
	pages = {195--259},
	OPT_doi = {10.1137/070699652},
	OPT_issn = {0097-5397}
}

@inproceedings{deng2022nash,
	title = {Nash convergence of mean-based learning algorithms in first price auctions},
	author = {Deng, Xiaotie and Hu, Xinyan and Lin, Tao and Zheng, Weiqiang},
	year = {2022},
	booktitle = {Proceedings of the ACM Web Conference 2022},
	pages = {141--150}
}

@article{vlatakis2020no,
  title   = {No-regret learning and mixed nash equilibria: They do not mix},
  author  = {Vlatakis-Gkaragkounis, Emmanouil-Vasileios and Flokas, Lampros and Lianeas, Thanasis and Mertikopoulos, Panayotis and Piliouras, Georgios},
  year    = 2020,
  journal = {Advances in Neural Information Processing Systems},
  volume  = 33,
  pages   = {1380--1391}
}

@article{foster1997CalibratedLearningCorrelated,
	title = {Calibrated Learning and Correlated Equilibrium},
	author = {Foster, Dean P and Vohra, Rakesh V},
	year = {1997},
	journal = {Games and Economic Behavior},
	volume = {21},
	number = {1-2},
	pages = {40}
}

@book{fudenberg1999TheoryLearningGames,
  title     = {The Theory of Learning in Games},
  author    = {Fudenberg, Drew and Levine, David K.},
  year      = 1999,
  publisher = {{MIT Press}},
  address   = {{Cambridge}},
  series    = {{{MIT Press}} Series on Economic Learning and Social Evolution},
  volume    = 2,
  isbn      = {978-0-262-06194-0},
  ids       = {fudenbergTheoryLearningGames1998},
  edition   = {2.}
}

@inproceedings{bailey2018multiplicative,
	title = {Multiplicative weights update in zero-sum games},
	author = {Bailey, James P and Piliouras, Georgios},
	year = {2018},
	booktitle = {Proceedings of the 2018 ACM Conference on Economics and Computation},
	pages = {321--338}
}

@article{milgrom1990rationalizability,
	title = {Rationalizability, learning, and equilibrium in games with strategic complementarities},
	author = {Milgrom, Paul and Roberts, John},
	year = {1990},
	journal = {Econometrica: Journal of the Econometric Society},
	publisher = {JSTOR},
	pages = {1255--1277}
}

@article{monderer1996potential,
	title = {Potential games},
	author = {Monderer, Dov and Shapley, Lloyd S},
	year = {1996},
	journal = {Games and economic behavior},
	publisher = {Elsevier},
	volume = {14},
	number = {1},
	pages = {124--143}
}

@article{milgrom1991adaptive,
	title = {Adaptive and sophisticated learning in normal form games},
	author = {Milgrom, Paul and Roberts, John},
	year = {1991},
	journal = {Games and Economic Behavior},
	publisher = {Elsevier},
	volume = {3},
	number = {1},
	pages = {82--100}
}

@book{cournot1838recherches,
	title = {Recherches sur les principes math{\'e}matiques de la th{\'e}orie des richesses},
	author = {Cournot, Antoine Augustin},
	year = {1838},
	publisher = {L. Hachette}
}

@inproceedings{mertikopoulos2018cycles,
	title = {Cycles in adversarial regularized learning},
	author = {Mertikopoulos, Panayotis and Papadimitriou, Christos and Piliouras, Georgios},
	year = {2018},
	booktitle = {Proceedings of the twenty-ninth annual ACM-SIAM symposium on discrete algorithms},
	pages = {2703--2717},
	organization = {SIAM}
}

@article{hart2003uncoupled,
	title = {Uncoupled dynamics do not lead to Nash equilibrium},
	author = {Hart, Sergiu and Mas-Colell, Andreu},
	year = {2003},
	journal = {American Economic Review},
	volume = {93},
	number = {5},
	pages = {1830--1836}
}

@article{hart2006stochastic,
	title = {Stochastic uncoupled dynamics and Nash equilibrium},
	author = {Hart, Sergiu and Mas-Colell, Andreu},
	year = {2006},
	journal = {Games and economic behavior},
	publisher = {Elsevier},
	volume = {57},
	number = {2},
	pages = {286--303}
}

@article{swenson2018best,
	title = {On best-response dynamics in potential games},
	author = {Swenson, Brian and Murray, Ryan and Kar, Soummya},
	journal = {SIAM Journal on Control and Optimization},
	volume = {56},
	number = {4},
	pages = {2734--2767},
	year = {2018},
	publisher = {SIAM}
}

@article{Kolpin2009Strict,
	title = {Strict dominance solvability without equilibrium},
	author = {Van Kolpin},
	journal = {Economics Bulletin},
	year = {2009},
	volume = {29},
	pages = {51-55},
	OPT_doi = {}
}

@article{cohen2017learning,
	title = {Learning with bandit feedback in potential games},
	author = {Cohen, Johanne and Heliou, Am{\'e}lie and Mertikopoulos, Panayotis},
	year = {2017},
	journal = {Advances in Neural Information Processing Systems},
	volume = {30}
}

@book{lattimore2020bandit,
	title = {Bandit algorithms},
	author = {Lattimore, Tor and Szepesv{\'a}ri, Csaba},
	year = {2020},
	month = {jul},
	publisher = {Cambridge University Press},
	OPT_doi = {10.1017/9781108571401},
	OPT_isbn = {978-1-108-57140-1 978-1-108-48682-8},
	OPT_url = {https://www.cambridge.org/core/product/identifier/9781108571401/type/book},
	urldate = {2023-03-20},
	edition = {1},
	language = {en}
}

@article{mertikopoulos2019learning,
	title = {Learning in games with continuous action sets and unknown payoff functions},
	author = {Mertikopoulos, Panayotis and Zhou, Zhengyuan},
	year = {2019},
	journal = {Mathematical Programming},
	publisher = {Springer},
	volume = {173},
	number = {1-2},
	pages = {465--507}
}

@article{milionis2022nash,
	title = {Nash, Conley, and Computation: Impossibility and Incompleteness in Game Dynamics},
	author = {Milionis, Jason and Papadimitriou, Christos and Piliouras, Georgios and Spendlove, Kelly},
	year = {2022},
	journal = {arXiv preprint arXiv:2203.14129},
	publisher = {National Acad Sciences},
	volume = {120},
	number = {41},
	pages = {e2305349120}
}

@techreport{oecd2017,
	title = {Algorithms and {Collusion}: {Competition} {Policy} in the {Digital} {Age}},
	OPT_url = {http://www.oecd.org/competition/algorithms-collusion-competition-policy-in-the-digital-age.htm},
	urldate = {2024-09-10},
	institution = {OECD},
	author = {{OECD}},
	year = {2017},
}

@article{aumann1987correlated,
	title = {Correlated equilibrium as an expression of Bayesian rationality},
	author = {Aumann, Robert J},
	journal = {Econometrica: Journal of the Econometric Society},
	pages = {1--18},
	year = {1987},
	publisher = {JSTOR}
}

@book{cesa2006prediction,
	OPT_address = {Cambridge},
	title = {Prediction, learning, and games},
	OPT_isbn = {978-0-511-19178-7},
	language = {en},
	publisher = {Cambridge University Press},
	author = {Cesa-Bianchi, Nicolò and Lugosi, Gábor},
	year = {2006},
}

@article{sanders2018prevalence,
	title = {The prevalence of chaotic dynamics in games with many players},
	author = {Sanders, James BT and Farmer, J Doyne and Galla, Tobias},
	year = {2018},
	journal = {Scientific {R}eports},
	publisher = {Nature Publishing Group},
	volume = {8},
	number = {1},
	pages = {1--13}
}

@article{shapley1964TopicsTwopersonGames,
	title = {Some Topics in Two-Person Games},
	author = {Shapley, Lloyd},
	year = {1964},
	journal = {Advances in game theory},
	volume = {52},
	pages = {1--29},
	ids = {shapley1964some}
}

@article{bichler2023soda,
	title = {Computing {Bayes}–{Nash} {Equilibrium} {Strategies} in {Auction} {Games} via {Simultaneous} {Online} {Dual} {Averaging}},
	author = {Bichler, Martin and Fichtl, Maximilian and Oberlechner, Matthias},
	year = {2023},
	journal = {Operations Research},
	OPT_doi = {10.1287/opre.2022.0287}
}

@article{viossat2013no,
	title = {No-regret dynamics and fictitious play},
	author = {Viossat, Yannick and Zapechelnyuk, Andriy},
	year = {2013},
	journal = {Journal of Economic Theory},
	publisher = {Elsevier},
	volume = {148},
	number = {2},
	pages = {825--842}
}

@book{young2004strategic,
	OPT_address = {Oxford},
	edition = {Repr},
	series = {The {Arne} {Ryde} memorial lectures},
	title = {Strategic learning and its limits},
	OPT_isbn = {978-0-19-926918-1},
	language = {eng},
	number = {2002},
	publisher = {Oxford Univ. Pr},
	author = {Young, H. Peyton},
	year = {2010}
}

@misc{bubeck2011introduction,
	OPT_address = {Princeton University, Department of Operations Reserach and Financial Engineering},
	type = {Lecture notes},
	title = {Introduction to {Online} {Optimization}},
	OPT_url = {http://sbubeck.com/BubeckLectureNotes.pdf},
	language = {en},
	urldate = {2024-09-25},
	author = {Bubeck, Sebastien},
	month = dec,
	year = {2011}
}

@article{eschenbaum2022robust,
	title = {Robust algorithmic collusion},
	author = {Eschenbaum, Nicolas and Mellgren, Filip and Zahn, Philipp},
	journal = {arXiv preprint arXiv:2201.00345},
	year = {2022}
}

@article{topkis1979submodular_games,
	author = {Topkis, Donald M.},
	title = {Equilibrium Points in Nonzero-Sum n-Person Submodular Games},
	year = {1979},
	issue_date = {Nov 1979},
	publisher = {Society for Industrial and Applied Mathematics},
	OPT_address = {USA},
	volume = {17},
	number = {6},
	OPT_issn = {0363-0129},
	OPT_url = {https://doi.org/10.1137/0317054},
	OPT_doi = {10.1137/0317054},
	journal = {SIAM J. Control Optim.},
	month = {nov},
	pages = {773–787},
	numpages = {15}
}

@article{slade1994,
    title = {What Does an {{Oligopoly Maximize}}?},
    author = {Slade, Margaret E.},
    year = {1994},
    journal = {The Journal of Industrial Economics},
    volume = {42},
    number = {1},
    pages = {45--61},
    publisher = {{Wiley}},
    OPT_issn = {0022-1821},
    OPT_doi = {10.2307/2950588},
}

@inproceedings{palaiopanos2017multiplicative,
    author = {Palaiopanos, Gerasimos and Panageas, Ioannis and Piliouras, Georgios},
    booktitle = {Advances in Neural Information Processing Systems},
    editor = {I. Guyon and U. Von Luxburg and S. Bengio and H. Wallach and R. Fergus and S. Vishwanathan and R. Garnett},
    pages = {},
    publisher = {Curran Associates, Inc.},
    title = {Multiplicative Weights Update with Constant Step-Size in Congestion Games:  Convergence, Limit Cycles and Chaos},
    volume = {30},
    year = {2017}
}

@article{calvano_algorithmic_2021-1,
	title = {Algorithmic collusion with imperfect monitoring},
	volume = {79},
	urldate = {2023-10-04},
	journal = {International Journal of Industrial Organization},
	author = {Calvano, Emilio and Calzolari, Giacomo and Denicoló, Vincenzo and Pastorello, Sergio},
	month = {dec},
	year = {2021},
	pages = {102712},
	OPT_issn = {Field (line: 33, key: `issn`): `0167-7187`},
	OPT_url = {Field (line: 34, key: `url`): `https://www.sciencedirect.com/science/article/pii/S0167718721000059`},
	OPT_doi = {Field (line: 35, key: `doi`): `10.1016/j.ijindorg.2021.102712`}
}

@misc{calzolari_pricing_2024,
	address = {Rochester, NY},
	type = {{SSRN} {Scholarly} {Paper}},
	title = {Pricing algorithms out of the box: a study of the repricing industry},
	shorttitle = {Pricing algorithms out of the box},
	language = {en},
	urldate = {2024-12-17},
	publisher = {Social Science Research Network},
	author = {Calzolari, Giacomo and Hanspach, Philip},
	month = {apr},
	year = {2024},
	OPT_url = {Field (line: 52, key: `url`): `https://papers.ssrn.com/abstract=4871394`},
	OPT_doi = {Field (line: 53, key: `doi`): `10.2139/ssrn.4871394`}
}

@article{asker_impact_2024,
	title = {The impact of artificial intelligence design on pricing},
	volume = {33},
	copyright = {© 2023 Wiley Periodicals LLC.},
	language = {en},
	number = {2},
	urldate = {2024-12-17},
	journal = {Journal of Economics \& Management Strategy},
	author = {Asker, John and Fershtman, Chaim and Pakes, Ariel},
	year = {2024},
	pages = {276--304},
	OPT_issn = {Field (line: 69, key: `issn`): `1530-9134`},
	OPT_url = {Field (line: 70, key: `url`): `https://onlinelibrary.wiley.com/doi/abs/10.1111/jems.12516`},
	OPT_doi = {Field (line: 71, key: `doi`): `10.1111/jems.12516`}
}

@article{johnson_platform_2023,
	title = {Platform {Design} {When} {Sellers} {Use} {Pricing} {Algorithms}},
	volume = {91},
	copyright = {© 2023 The Econometric Society},
	language = {en},
	number = {5},
	urldate = {2024-12-17},
	journal = {Econometrica},
	author = {Johnson, Justin P. and Rhodes, Andrew and Wildenbeest, Matthijs},
	year = {2023},
	pages = {1841--1879},
	OPT_issn = {Field (line: 108, key: `issn`): `1468-0262`},
	OPT_url = {Field (line: 109, key: `url`): `https://onlinelibrary.wiley.com/doi/abs/10.3982/ECTA19978`},
	OPT_doi = {Field (line: 110, key: `doi`): `10.3982/ECTA19978`}
}

@article{nevo_measuring_2001,
	title = {Measuring {Market} {Power} in the {Ready}-to-{Eat} {Cereal} {Industry}},
	volume = {69},
	doi = {10.1111/1468-0262.00194},
	number = {2},
	journal = {Econometrica},
	author = {Nevo, Aviv},
	year = {2001},
	pages = {307--342},
	url_OPT = {https://onlinelibrary.wiley.com/doi/abs/10.1111/1468-0262.00194}
}

@misc{berry_automobile_1993,
	title = {Automobile {Prices} in {Market} {Equilibrium}: {Part} {I} and {II}},
	doi = {10.3386/w4264},
	publisher = {National Bureau of Economic Research},
	author = {Berry, Steven and Levinsohn, James and Pakes, Ariel},
	year = {1993},
	url_OPT = {https://www.nber.org/papers/w4264}
}

@misc{conlon_pyblp_2025,
	title = {{PyBLP}},
	author = {Conlon, Christopher and Gortmaker, Jeff},
	year = {2025},
	url_OPT = {https://pyblp.readthedocs.io/en/latest/index.html}
}

@article{conlon_best_2020,
	title = {Best practices for differentiated products demand estimation with {PyBLP}},
	volume = {51},
	doi = {10.1111/1756-2171.12352},
	number = {4},
	journal = {The RAND Journal of Economics},
	author = {Conlon, Christopher and Gortmaker, Jeff},
	year = {2020},
	pages = {1108--1161},
	url_OPT = {https://onlinelibrary.wiley.com/doi/abs/10.1111/1756-2171.12352}
}

@article{nevo_practitioners_2000,
	title = {A {Practitioner}'s {Guide} to {Estimation} of {Random}-{Coefficients} {Logit} {Models} of {Demand}},
	volume = {9},
	doi = {10.1111/j.1430-9134.2000.00513.x},
	number = {4},
	journal = {Journal of Economics \& Management Strategy},
	author = {Nevo, Aviv},
	year = {2000},
	pages = {513--548},
	url_OPT = {https://onlinelibrary.wiley.com/doi/abs/10.1111/j.1430-9134.2000.00513.x}
}

@article{jin_v-learningsimple_2024,
	title = {V-{Learning}—{A} {Simple}, {Efficient}, {Decentralized} {Algorithm} for {Multiagent} {Reinforcement} {Learning}},
	volume = {49},
	doi = {10.1287/moor.2021.0317},
	number = {4},
	journal = {Mathematics of Operations Research},
	author = {Jin, Chi and Liu, Qinghua and Wang, Yuanhao and Yu, Tiancheng},
	year = {2024},
	pages = {2295--2322},
	url_OPT = {https://pubsonline.informs.org/doi/abs/10.1287/moor.2021.0317}
}

@misc{wang_learning_2022,
	title = {Learning {Rationalizable} {Equilibria} in {Multiplayer} {Games}},
	doi = {10.48550/arXiv.2210.11402},
	publisher = {arXiv},
	author = {Wang, Yuanhao and Kong, Dingwen and Bai, Yu and Jin, Chi},
	year = {2022},
	url_OPT = {http://arxiv.org/abs/2210.11402}
}

@article{douglas_naive_2024,
	title = {Naive {Algorithmic} {Collusion}: {When} {Do} {Bandit} {Learners} {Cooperate} and {When} {Do} {They} {Compete}?},
	journal = {ICIS 2024 Proceedings},
	author = {Douglas, Connor and Provost, Foster and Sundararajan, Arun},
	year = {2024},
	url_OPT = {https://aisel.aisnet.org/icis2024/aiinbus/aiinbus/15}
}

@article{harrington_developing_2018,
	title = {Developing {Competition} {Law} for {Collusion} by {Autonomous} {Artificial} {Agents}},
	volume = {14},
	doi = {10.1093/joclec/nhy016},
	number = {3},
	journal = {Journal of Competition Law \& Economics},
	author = {Harrington, Joseph E},
	year = {2018},
	pages = {331--363},
	url_OPT = {https://doi.org/10.1093/joclec/nhy016}
}

@misc{schaefer_emergence_2023,
	title = {On the {Emergence} of {Cooperation} in the {Repeated} {Prisoner}'s {Dilemma}},
	doi = {10.48550/arXiv.2211.15331},
	publisher = {arXiv},
	author = {Schaefer, Maximilian},
	year = {2023},
	url_OPT = {http://arxiv.org/abs/2211.15331}
}

@article{yang_competitive_2024,
	title = {Competitive {Demand} {Learning}: a {Non}-cooperative {Pricing} {Algorithm} with {Coordinated} {Price} {Experimentation}},
	doi = {10.1177/10591478231224912},
	journal = {Production and Operations Management},
	author = {Yang, Yongge and Lee, Yu-Ching and Chen, Po-An},
	year = {2024},
	pages = {10591478231224912},
	url_OPT = {https://doi.org/10.1177/10591478231224912}
}

@misc{goyal_learning_2023,
	title = {Learning to {Price} {Under} {Competition} for {Multinomial} {Logit} {Demand}},
	doi = {10.2139/ssrn.4572453},
	publisher = {Social Science Research Network},
	author = {Goyal, Vineet and Li, Shukai and Mehrotra, Sanjay},
	year = {2023},
	url_OPT = {https://papers.ssrn.com/abstract=4572453}
}

@article{loots_datadriven_2023,
	title = {Data‐driven collusion and competition in a pricing duopoly with multinomial logit demand},
	volume = {32},
	doi = {10.1111/poms.13919},
	number = {4},
	journal = {Production and Operations Management},
	author = {Loots, Thomas and den Boer, Arnoud V.},
	year = {2023},
	pages = {1169--1186},
	url_OPT = {https://doi.org/10.1111/poms.13919}
}

@article{barfuss_intrinsic_2023,
	title = {Intrinsic fluctuations of reinforcement learning promote cooperation},
	volume = {13},
	doi = {10.1038/s41598-023-27672-7},
	number = {1},
	journal = {Scientific Reports},
	author = {Barfuss, Wolfram and Meylahn, Janusz M.},
	year = {2023},
	pages = {1309},
	url_OPT = {https://www.nature.com/articles/s41598-023-27672-7}
}

@article{feng_convergence_2021,
	title = {Convergence {Analysis} of {No}-{Regret} {Bidding} {Algorithms} in {Repeated} {Auctions}},
	volume = {35},
	doi = {10.1609/aaai.v35i6.16680},
	number = {6},
	journal = {Proceedings of the AAAI Conference on Artificial Intelligence},
	author = {Feng, Zhe and Guruganesh, Guru and Liaw, Christopher and Mehta, Aranyak and Sethi, Abhishek},
	year = {2021},
	pages = {5399--5406},
	url_OPT = {https://ojs.aaai.org/index.php/AAAI/article/view/16680}
}

@article{foster_regret_1999,
	title = {Regret in the {On}-{Line} {Decision} {Problem}},
	volume = {29},
	doi = {10.1006/game.1999.0740},
	number = {1},
	journal = {Games and Economic Behavior},
	author = {Foster, Dean P. and Vohra, Rakesh},
	year = {1999},
	pages = {7--35},
	url_OPT = {https://www.sciencedirect.com/science/article/pii/S0899825699907406}
}

@misc{aguiar-curry_ab-325_2025,
	title = {{AB}-325 {Cartwright} {Act}: violations.},
	author = {Aguiar-Curry and Ward},
	year = {2025},
	url_OPT = {https://leginfo.legislature.ca.gov/faces/billNavClient.xhtml?bill_id=202520260AB325}
}

@misc{vinson__elkins_llp_california_nodate,
	title = {California {Looks} to {Crack} {Down} on {Algorithmic} {Pricing} and {Clarify} {Antitrust} {Pleading} {Standards}},
	journal = {JD Supra},
    year = {2025},
	author = {{Vinson \& Elkins LLP} and Ballard, Dylan and Costello, Kevin and Lo, Madison and Scarborough, Mike},
	url_OPT = {https://www.jdsupra.com/legalnews/california-looks-to-crack-down-on-3793901/}
}

@article{milgrom_monotone_1994,
	title = {Monotone {Comparative} {Statics}},
	volume = {62},
	doi = {10.2307/2951479},
	number = {1},
	journal = {Econometrica},
	publisher = {[Wiley, Econometric Society]},
	author = {Milgrom, Paul and Shannon, Chris},
	year = {1994},
	pages = {157--180},
	url_OPT = {https://www.jstor.org/stable/2951479}
}

@article{jann_correlated_2015,
	title = {Correlated equilibria in homogeneous good {Bertrand} competition},
	volume = {57},
	doi = {10.1016/j.jmateco.2015.01.005},
	journal = {Journal of Mathematical Economics},
	author = {Jann, Ole and Schottmüller, Christoph},
	year = {2015},
	pages = {31--37},
	url_OPT = {https://www.sciencedirect.com/science/article/pii/S0304406815000130}
}

@inproceedings{wu_correlated_2008,
	title = {Correlated {Equilibrium} of {Bertrand} {Competition}},
	doi = {10.1007/978-3-540-92185-1_24},
	booktitle = {Internet and {Network} {Economics}},
	publisher = {Springer},
	author = {Wu, John},
	editor = {Papadimitriou, Christos and Zhang, Shuzhong},
	year = {2008},
	pages = {166--177}
}

@misc{Hartline2026ClarificationAlgorithmicCollusion,
	author       = {Jason Hartline},
	title        = {Clarification of `Algorithmic Collusion without Threats'},
	year         = {2026},
	eprint       = {2602.22232},
	archivePrefix = {arXiv},
	primaryClass = {cs.GT},
	note         = {Also listed under econ.TH}
}

@inproceedings{hartline2024regulation,
	title={Regulation of algorithmic collusion},
	author={Hartline, Jason D and Long, Sheng and Zhang, Chenhao},
	booktitle={Proceedings of the 2024 Symposium on Computer Science and Law},
	pages={98--108},
	year={2024}
}

\vfill
\pagebreak

\appendix
\renewcommand{\theHsection}{A\arabic{section}}
\small
\setlength{\parindent}{0pt}

\section{Proofs}
\label{sec:proofs}

\subsection{Proofs of Game-Theoretical Properties}\label{app:proof1}
\begin{proposition}\label{prop:super_standard}
	The Bertrand competition with all-or-nothing demand is not supermodular.
\end{proposition}

\begin{proof}
	To prove that the Bertrand-competition with standard all-or-nothing demand is not supermodular, we show that the property of increasing differences, i.e., $ u_i(a_i', a_{-i}') - u_i(a_i, a_{-i}') \geq u_i(a_i', a_{-i}) - u_i(a_i, a_{-i}) $ with $a_i' \geq a_i$ and $ a_{j}' \geq a_{j} $ for all $ j \neq i $, is not satisfied.
	To that end, we consider two cases in the simple $n=2$ player case:
	\begin{align*}
		\intertext{ For $a_2 < a_1 < a_2' < a_1' $, the inequality is violated:}
		u_1(a_1', a_2') - u_1(a_1, a_2') &= 0 - u_1(a_1, a_2') < 0 \\
		u_1(a_1', a_2) - u_1(a_1, a_2) &= 0 - 0 = 0
		\intertext{For $a_1 < a_2 < a_1' < a_2' $ with $1-(a_1-c_1) > a_1'$, the inequality is satisfied:}
		u_1(a_1', a_2') - u_1(a_1, a_2') &= (1-a_1')(a_1'-c_1) - (1-a_1)(a_1 - c_1) > 0 \\
		u_1(a_1', a_2) - u_1(a_1, a_2) & = 0 - u_1(a_1, a_2) < 0
	\end{align*}
	The condition $1-(a_1-c_1) > a_1'$ simply means that $a_1$ is closer to $c_1$ than $a_1'$ is to 1, which is necessary to ensure that $ u_1(a_1', a_2') > u_1(a_1, a_2')$.
	Since the inequalities is strictly satisfied and violated, we cannot define a different order on the opponent's action set such that the increasing differences property is satisfied. Therefore, the game cannot be supermodular. Note that this argument also works for continuous actions.
\end{proof}

\begin{proposition} \label{prop:potential}
	The Bertrand competition with linear demand is a potential game.
\end{proposition}

\begin{proof}
	To show that the Bertrand competition with linear demand is a potential game, we prove that a function $\phi: \Xcal \rightarrow \R$ satisfies the definition of a potential. Our choice of $\phi$ is similar to the potential function of the Cournot setting in \cite{slade1994} and is given by
	\begin{equation*}
		\phi(a) = \sum_k \alpha_k a_k - \sum_k \beta_k a_k (a_k - c_k) + \tfrac{\gamma}{n-1} \sum_{k < l} a_k a_l.
	\end{equation*}
	It is easy to check that this function satisfies the condition of an exact potential:
	\begin{align*}
		\phi(a_i, a_\mi) - \phi(a'_i, a_\mi) 
		&= \alpha_i [a_i - a'_i] -  \beta_i [a_i (x_i - c_i) - a'_i (a'_i - c_i)] \\
		& \quad + \tfrac{\gamma}{n-1} \sum_{l > i} [a_i - a'_i] a_l  + \tfrac{\gamma}{n-1} \sum_{k < i} a_k (a_i - a'_i). 
		\intertext{Note that in the sum over $k,l$, we have the two cases where either $k=i$ or $l=i$, which leads to the two sums above. We can combine them and get:}
		&= \alpha_i (a_i - a'_i) -  \beta_i [a_i (a_i - c_i) - a'_i (a'_i - c_i)] + \tfrac{\gamma}{n-1} \sum_{j \neq i} a_j (a_i - a'_i)\\
		&= d_i(a_i,a_\mi) \cdot (a_i - c_i) - d_i(a'_i,a_\mi) \cdot (a'_i - c_i) = u_i(a_i,a_\mi) - u_i(a'_i,a_\mi).
	\end{align*}
	This argument works for both continuous and discrete actions.
\end{proof}

\begin{proposition} \label{prop:potentialstandard}
	The Bertrand competition with all-or-nothing demand is not a potential game.
\end{proposition}

\begin{proof}
	To show that the Bertrand competition with standard demand is not a potential game, it is sufficient to show that there exists an example, where the closed path property is not satisfied \citep[Theorem~2.8]{monderer1996potential}. For simplicity, we focus on the two player version, but the argument can be extended to more agents. Consider the actions $ 0 < a_1 < a_2 < a_1' < a_2' < 1$ and the closed path $\gamma$ of length 4 given by the action profiles $(a_1, a_2) \rightarrow (a_1', a_2) \rightarrow (a_1', a_2') \rightarrow (a_1, a_2') \rightarrow (a_1,a_2)$.
	We know that our game is a potential game if and only if $I(\gamma,u)=0$. In this example, $I(\gamma,u)$ is given by 
	\begin{align*}
		I(\gamma,u) 
		&:= [u_1(a_1',a_2) - u_1(a_1,a_2)] + [u_2(a_1',a_2') - u_2(a_1', a_2)] \\ 
		&\quad + [u_1(a_1,a_2') - u_1(a_1',a_2')] + [u_2(a_1,a_2) - u_2(a_1, a_2')] \\
		\intertext{Using the definition of the standard demand model, we get}
		&= [0 - D(1-a_1)(a_1-c_1)] + [0 - D(1-a_2)(a_2-c_2)] \\
		&\quad+ [D(1-a_1)(a_1-c_1) - D(1-a_1')(a_1'-c_1)] + [0 - 0] \\
		&= - D(1-a_2)(a_2-c_2) - D(1-a_1')(a_1'-c_1) 
	\end{align*}
	If we further assume that both agents bid above their marginal costs, i.e., $c_1 < a_1 < 1$ and  $c_2 < a_2$ the last line is strictly negative and thereby $I(\gamma,u) \neq 0$, which implies that the game cannot be a potential game.
\end{proof}

\begin{proposition} \label{prop:potentiallogit}
	The Bertrand competition with logit demand is in general not a potential game.
\end{proposition}

\begin{proof}
	Since the utility function in the Bertrand competition with logit demand is twice differentiable, we can use that the game is a potential game if and only if  $\tfrac{\partial^2u_i}{\partial a_i \partial a_j} = \tfrac{\partial^2u_j}{\partial a_i \partial a_j}$ \citep[Theorem~4.5]{monderer1996potential}. Again, we focus on a simplified version to show that the game is in general not a potential game. Let us assume that we have $n=2$ agents with costs $c_1, c_2 > 0$ and $\alpha_0 = \alpha_1 = \alpha_2 = \alpha$, and $\mu_0 = \mu_1 = \mu_2 = \mu$. If we compute the second derivatives we get
	\begin{align*}
		&\quad \frac{\partial^2u_1(a_1, a_2)}{\partial a_1 \partial a_2} =\frac{\partial^2u_2(a_1, a_2)}{\partial a_1 \partial a_2} \\
		\Leftrightarrow \quad &\mu \left(e^{\frac{\alpha}{\mu}} + e^{\frac{- a_{1} + \alpha}{\mu}} + e^{\frac{- a_{2} + \alpha}{\mu}}\right) e^{\frac{- a_{1} - a_{2} + 2 \alpha}{\mu}} + \left(- a_{1} + c_1\right) \left(e^{\frac{\alpha}{\mu}} + e^{\frac{- a_{1} + \alpha}{\mu}} + e^{\frac{- a_{2} + \alpha}{\mu}}\right) e^{\frac{- a_{1} - a_{2} + 2 \alpha}{\mu}} + 2 \left(a_{1} - c_1\right) e^{\frac{- 2 a_{1} - a_{2} + 3 \alpha}{\mu}} \\
		&= \mu \left(e^{\frac{\alpha}{\mu}} + e^{\frac{- a_{1} + \alpha}{\mu}} + e^{\frac{- a_{2} + \alpha}{\mu}}\right) e^{\frac{- a_{1} - a_{2} + 2 \alpha}{\mu}} + \left(- a_{2} + c_{2}\right) \left(e^{\frac{\alpha}{\mu}} + e^{\frac{- a_{1} + \alpha}{\mu}} + e^{\frac{- a_{2} + \alpha}{\mu}}\right) e^{\frac{- a_{1} - a_{2} + 2 \alpha}{\mu}} + 2 \left(a_{2} -c_2\right) e^{\frac{- a_{1} - 2 a_{2} + 3 \alpha}{\mu}} \\
		\Leftrightarrow \quad & \left(a_{1} - c_1\right) \left(- \left(e^{\frac{\alpha}{\mu}} + e^{\frac{- a_{1} + \alpha}{\mu}} + e^{\frac{- a_{2} + \alpha}{\mu}}\right) e^{\frac{- a_{1} - a_{2} + 2 \alpha}{\mu}} + 2 e^{\frac{- 2 a_{1} - a_{2} + 3 \alpha}{\mu}}\right) \\
		&=\left(a_{2} - c_2\right) \left(- \left(e^{\frac{\alpha}{\mu}} + e^{\frac{- a_{1} + \alpha}{\mu}} + e^{\frac{- a_{2} + \alpha}{\mu}}\right) e^{\frac{- a_{1} - a_{2} + 2 \alpha}{\mu}} + 2 e^{\frac{- a_{1} - 2 a_{2} + 3 \alpha}{\mu}}\right) \\
		\intertext{It is easy to see that if both agents bid for instance $c_1$, only the left side is equal to zero while the right side is generally not.
			And even if agents are symmetric with respect to the costs, i.e., $c_1 = c_2 = c$, and assume that agents bid $a_1 = 2c$ and $a_2 = 3c$ we get}
		\Rightarrow \quad & c \left(e^{\frac{3 c}{\mu}} + 3 e^{\frac{c}{\mu}} - 3\right) e^{\frac{3 \alpha - 8 c}{\mu}} = 0, \text{ which is obviously not true in general.}
	\end{align*}
	Therefore the condition from the theorem is not satisfied and the Bertrand competition with logit demand is in general not a potential game. 
\end{proof}

\subsection{Remarks on Correlated Rationalizability}
\label{sec:remarks-correlated-rationalizability}

In this section, we define a series of rational response sets, $(\Acal^i)_{i = 1}^{K}$, together with a game-dependent number, the competition constant $\delta$, that will play an essential role in our proof of convergence. 

\begin{definition}[Correlated Rational Responses]
	\label{def:correlated-rational-responses}
	Given a product set $\hat{\Acal} = \hat{\Acal}_1 \times \dots \times \hat{\Acal}_{\n}$ of action sets $\hat{\Acal}_i \subseteq \Acal_i$, we define player $i$'s correlated rational responses $R_i(\hat{\Acal})$ as 
	\begin{equation*}
		R_i(\hat{\Acal}) = \{ a_i \in \Acal_i ~\vert ~ \exists ~ \hat{x}_{-i} \in \Delta(\hat{\Acal}_{-i}), ~ \forall ~ a_i' \in \Acal_i : ~ u_i(a_i, \hat{x}_{-i}) \geq u_i(a_i', \hat{x}_{-i}) \}.
	\end{equation*}
\end{definition}
In other words, the correlated rational responses are the actions which maximize the utility for some probability distribution $\hat{x}_{-i}$ with support on $\hat{\Acal}_{-i}$.
For all actions $a_i \not\in R_i(\hat{\Acal})$, we find that there always exists a superior action for every probability distribution $\hat{x}_{-i} \in \Delta(\hat{\Acal}_{-i})$:
\begin{equation}
	\label{eq:irrational-actions}
	\forall ~ \hat{x}_{-i} \in \Delta(\hat{\Acal}_{-i}), ~ \exists ~ a_i' \in \Acal_i: \quad u_i(a_i, \hat{x}_{-i}) < u_i(a_i', \hat{x}_{-i}).
\end{equation}

Based on the concept of correlated rational responses from Definition \ref{def:correlated-rational-responses}, we can define an iteration of product sets as follows. 
Let $\Acal^0 = \Acal_1 \times \dots \times \Acal_\n$ be the cartesian product of the action sets of all players and let $R(\hat{\Acal}) = R_1(\hat{\Acal}) \times \dots \times R_{\n}(\hat{\Acal})$ be an operator on the sets of all players' actions that combines the correlated rational responses of all players. For $k \geq 1$, we iteratively define $\Acal^k = R(\Acal^{k-1})$. We say that an action is \emph{stable under $R$ or $R_i$} for a given $\hat{\Acal}$ if it is contained in $R_i(\hat{\Acal})$.
Furthermore, we denote the \emph{competition constant} $\delta > 0$ of the game as the largest number for which the utility difference for non-rational responses to $\Acal^k$ is at least $\delta$ for all $k$:
\begin{equation*}
	\forall k \geq 0, ~ a_i \not\in R_i(\Acal^k) : \quad \forall ~ \hat{x}_{-i} \in \Delta(\Acal^k_{-i}), ~ \exists ~ a_i' \in \Acal_i : ~ u_i(a_i', \hat{x}_{-i}) - u_i(a_i, \hat{x}_{-i}) \geq \delta .
\end{equation*}

\begin{proposition}
	\label{prop:competition-constant-exists}
	The competition constant $\delta > 0$ exists for any finite normal-form game in which the set of correlated rationalizable actions is a strict subset of the action spaces: $\bar{\Acal} \subset \Acal$.
\end{proposition}

\begin{proof}
	Let $k \geq 0$ be an iteration such that $\exists a_i \not\in R_i(\Acal^k)$. Since $a_i$ is not a correlated rational response, we know that
	\begin{equation*}
		\forall ~ \hat{x}_{-i} \in \Delta(\Acal^k_{-i}), ~ \exists ~ a_i' \in \Acal_i : ~ u_i(a_i', \hat{x}_{-i}) - u_i(a_i, \hat{x}_{-i}) > 0.
	\end{equation*}
	This is equivalent to stating that 
	\begin{equation*}
		\forall ~ \hat{x}_{-i} \in \Delta(\Acal^k_{-i}), \qquad g(\hat{x}_{-i}) > 0,
	\end{equation*}
	where we define $g(\hat{x}_{-i}) := \max_{a_i' \in \Acal_i} u_i(a_i', \hat{x}_{-i}) - u_i(a_i, \hat{x}_{-i})$. $g$ is the maximum of a finite number of linear functions (the utility differences), so it is real-valued and continuous. Additionally, the space $\Delta(\Acal_{-i}^k)$ is compact, which implies that the minimum of $g$ on this space is attained according to the extreme value theorem. Since we know that $g(\hat{x}_{-i}) > 0$, the minimum must be strictly positive, yielding
	\begin{equation*}
		\forall ~ \hat{x}_{-i} \in \Delta(\Acal^k_{-i}), \qquad g(\hat{x}_{-i}) \geq \delta_k,
	\end{equation*}
	for some $\delta_k > 0$. This again implies that
	\begin{equation*}
		\forall ~ \hat{x}_{-i} \in \Delta(\Acal^k_{-i}), ~ \exists ~ a_i' \in \Acal_i : ~ u_i(a_i', \hat{x}_{-i}) - u_i(a_i, \hat{x}_{-i}) \geq \delta_k.
	\end{equation*}
	Now take $\delta = \min_{\{k \geq 0 \text{ s.t. } \exists a_i \not\in R_i(\Acal^k)\}} \{\delta_k\} > 0$, a minimum over a finite number of values, to obtain the result.
\end{proof}

\begin{remark}
	\label{rem:rationalizable-fix-points}
	The sets of correlated rationalizable actions $\bar{\Acal}$ are fix points under $R(\cdot)$, which follows directly from Definition \ref{def:rationalizable-actions}.
\end{remark}

\begin{remark}
	The operator $R(\cdot)$ may cycle for some games and some inputs, meaning that after some number of iterations greater than 1, the same set is obtained: $(R \circ \dots \circ R) (\hat{\Acal}) = \hat{\Acal}$. 
	However, this can only happen if there are some correlated rational responses to $\hat{\Acal}$ that are not within $\hat{\Acal}$. If we start applying $R$ at the full action sets $\Acal^0$ as defined above, we will never end up in such a cyclic relation.
	Instead, we will reach a fixed point of $R$ that is a correlated rationalizable set. The following proposition makes this statement more formal.
\end{remark}

\begin{proposition}
	\label{prop:rationalizable-subset-relation}
	As before, let $\Acal^0 = \Acal = \Acal_1 \times \dots \times \Acal_n$ and $\Acal^{k} = R(\Acal^{k-1})$, $k = 1, \dots, K$, define the sequence of correlated rational responses. 
	The sets $\Acal^k_i$ are (non-strict) subsets of their predecessors $\Acal^{k-1}_i$: $ \Acal^k_i \subseteq \Acal^{k-1}_i$.
\end{proposition}

\begin{proof}
	We prove this statement by induction.
	\begin{description}
		\setlength{\itemindent}{0pt}
		\item[\emph{Induction start, $k = 1$:}] Since $\Acal^0 = \Acal$, we get $\Acal^1 \subseteq \Acal^0$.
		\item[\emph{Induction step, $k \to k + 1$:}] 
		Assume that $\Acal^{k'} \subseteq \Acal^{k'-1}$ holds for all $k' \leq k$. We show that this also holds for $k + 1$, i.e., $\Acal_{i}^{k+1} \subseteq \Acal_i^k$.
		Assume, towards contradiction, that $\Acal^{k+1}_i$ contains an element $a \not \in \Acal^k_i$. For this element, there is a mixed opponent strategy $\hat{x}_{-i} \in \Delta(\Acal^k_{-i})$ such that $a$ maximizes $u_i(a, \hat{x}_{-i})$. Since $\Acal^k_{-i} \subseteq \Acal^{k-1}_{-i}$, $\hat{x}_{-i} \in \Delta(\Acal^{k-1}_{-i})$, and thus $a$ must also be in $\Acal^k_i$, leading to contradiction.
	\end{description}
\end{proof}


Finally, we also provide proofs for Proposition \ref{prop:CE-is-correlated-rationalizable} and \ref{prop:SSU-not-in-CCE}.

\begin{proof}[Proof of Proposition \ref{prop:CE-is-correlated-rationalizable}]
	Consider an action $a_i \in \Acal_i$ that is supported in some CE, meaning that there is some probability distribution $p$ on the cartesian product of action sets $\Acal = \Acal_1 \times \dots \times \Acal_\n$ such that for all $a_i' \in \Acal_i$
	\begin{equation*}
		\E_{a \sim p} [u_i(a) \vert a_i] \geq \E_{a \sim p} [u_i(a_i',a_{-i}) \vert a_i].
	\end{equation*}
	The existence of $p$ implies that there exists a joint distribution $\hat{x}_{-i}$ on the opponents' actions $\Acal_{-i}$, namely the posterior distribution of $p$ given $a_i$, such that for all $a_i'$
	\begin{equation*}
		u_i(a_i, \hat{x}_{-i})  \geq u_i(a_i',\hat{x}_{-i}).
	\end{equation*}
	This shows that $a_i$ is a correlated rational response to $\hat{x}_{-i}$. Because this holds for all players in a CE, all players can only choose correlated rational responses. Mutually correlated rational responses show that the actions are indeed correlated rationalizable.
\end{proof}

\begin{proof}[Proof of Proposition \ref{prop:SSU-not-in-CCE}]
	We show that not every strictly serially undominated action is supported by some CCE. To that end we consider the following matrix game with two agents $\Ncal = \{1, 2\}$ and two actions $\Acal_1 = \Acal_2 = \{a_1, a_2\}$ each. The utilities are given by the payoff matrices
	\[
	P_1 = \begin{pmatrix}
		0 & 2 \\ 2 & 1
	\end{pmatrix}, \quad P_2 =  \begin{pmatrix}
		2 & 0 \\ 0 & 0
	\end{pmatrix},
	\]
	where $u_i(a_k, a_l) = (P_i)_{kl}$ with $i \in \Ncal$ and $k,l \in \{1,2\}$. 
	It is easy to see that no action is strictly dominated. Therefore the set of serially undominated actions is the whole action set. We will now show that there is no CCE $\sigma = (\sigma_{11}, \sigma_{12}, \sigma_{21}, \sigma_{22}) \in \Delta(\Acal)$ where the action $a_1$ of agent 1 has a positive probability, i.e., we show that $\sigma_{11} = \sigma_{12} = 0$, even though $a_1$ is undominated.
	The constraints for the CCE are given by the following constraints (see Definition \ref{def:cce}):
	\begin{align*}
		\sigma_{21}(2-0) + \sigma_{22}(1-2) & \geq 0, \\
		\sigma_{11}(0-2) + \sigma_{12}(2-1) & \geq 0, \\ 
		\sigma_{12}(0-2) + \sigma_{22}(0-0) & \geq 0, \\ 
		\sigma_{11}(2-0) + \sigma_{21}(0-0) & \geq 0.
	\end{align*}
	Since $\sigma$ has to be a probability measure over $\Acal$, we also have $\sigma_{11} + \sigma_{12} + \sigma_{21} + \sigma_{22} = 1$ and $\sigma_{11}, \sigma_{12}, \sigma_{21}, \sigma_{22}  \geq 0 $
	From the third constraint we immediately get $\sigma_{12} \leq 0$ and thereby $\sigma_{12} = 0$. Using this in the second constraint gives us $\sigma_{12} \geq 2 \sigma_{11}$. Therefore we have $\sigma_{11} = \sigma_{12} = 0$, which means that action 1 of agent 1 is not supported in any CCE even though it is a undominated action $\Rightarrow$ $SSU \not\subseteq CCE$.
\end{proof}

\subsection{Correlated Rationalizable Set in Bertrand Oligopolies}
\label{app:correlated-rationalizable-bertrand}

\begin{proposition}
	\label{prop:standard-demand-correlated-rationalizable}
	In the symmetric Bertrand competition game with standard demand and action space $\Acal_i = \{a^1, \dots, a^\numact\}$, $i = 1, \dots, \n$, with $c < a^1 < \dots < a^\numact \leq 1$ and $a^{k-1} - (a^k - a^{k-1}) \geq c$, the set of correlated rationalizable actions is $\bar{\Acal}_i = \{a^1\}$ for every player, which is equivalent to the unique Nash equilibrium of the game.
\end{proposition}

\begin{proof}
	First, we show that $a_i^* = a^1, \, \forall i \in \players$ is the unique pure Nash equilibrium. It is easy to see that it is a NE, since it is only possible to deviate to a higher action which yields a utility of $0$ instead of $\tfrac{D}{n}(1-a^1)(c-a^1) > 0$.
	It remains to show that there is no other Nash equilibrium. Other candidates have to be symmetric, i.e., $a_i^* = a^k, \, \forall i \in \players$ for some $k \in \{2, \dots, \numact\}$, since otherwise there is at least one agent pricing higher than some other agents, which means that the player has a utility of zero which could be strictly improved by deviating to $a^1$. 
	Now assuming that all agents play $a^k$, we want to show that deviating to $a^{k-1}$ strictly increases an agent's utility, showing that $a_i^* = a^k$ for all $i \in \players$ cannot be a NE.
	\begin{align*}
		u_i(a^{k-1}, a^k_{-i}) -  u_i(a^k, a^k_{-i}) 
		&= D(1-a^{k-1})(a^{k-1}-c) - \tfrac{D}{n}(1-a^k)(a^k-c) \\
		&= \tfrac{D}{n} \left[ n(1-a^{k-1})(a^{k-1}-c) - (1-a^{k-1})(a^k-c) + (a^k-a^{k-1})(a^k-c) \right] \\
		&\geq \tfrac{D}{n} \left[ (1-a^{k-1})(2a^{k-1}-2c) - (1-a^{k-1})(a^k-c) + (a^k-a^{k-1})(a^k-c) \right] \\
		&=\tfrac{D}{n} \left[ (1-a^{k-1})(a^{k-1} - (a^k-a^{k-1}) - c)+ (a^k-a^{k-1})(a^k-c) \right]  \\
		&\geq \tfrac{D}{n} (a^k-a^{k-1})(a^k-c) > 0.
	\end{align*}
	In the last step, we use the assumption on the discretization, i.e., $(a^{k-1} - (a^k-a^{k-1}) - c) \geq 0$. 
	Therefore, we have uniqueness of our NE.    
	
	Next, we want to show that no other action than $a^1$ can be correlated rationalizable. Towards contradiction, assume we have supports $\hat{\Acal}_1, \dots, \hat{\Acal}_\n$ for all players where at least one player can play an action unequal $a^1$. We show that there must be some player $i$ and one action for which Equation \eqref{eq:irrational-actions} holds, meaning that the support will change under $R_i$.
	W.l.o.g. let player $i$ support actions up to $a^k, k \in \{2, \dots, \numact\}$, the highest action of any support. (There could be multiple such players.)	The other players have a joint distribution $x_{-i}$ with support only on actions up to $a^{max} = a^k $. 
	Let us  define $\Delta := \max_a u_i(a) $ and $p \in [0, 1]$, which denotes the probability that all opponents play $a^k$ under $x_{-i}$.
	Depending on $x_{-i}$ (in particular $p$), we will show that either $a^1$ or $a^{k-1}$ can be chosen to get inequality (\ref{eq:irrational-actions}).
	\begin{enumerate}[label=(\roman*)]
		\item If $p \leq \frac{1}{2 \Delta} (1-a^1)(a^1-c)$, we choose $a_i' = a^1$:
		\begin{align*}
			\E_{-i} \left[ u_i(a^1, a_{-i}) - u_i(a^k, a_{-i}) \right] 
			&= p  \left[ u_i(a^1, a^k_{-i}) - u_i(a^k, a^k_{-i}) \right] + (1-p) \left[ u_i(a^1, a^1_{-i}) + \dots \right]\\
			&\geq p \tfrac{D}{n}  \left[ (n-1) (1-a^1)(a^1-c) - (1-a^k)(a^k-c) \right]  + \tfrac{D}{n} (1-a^1)(a^1-c) \\
			&> - p \tfrac{D}{n} \Delta + \tfrac{D}{n} (1-a^1)(a^1-c) \\
			&> \tfrac{D}{2n} (1-a^1)(a^1-c) = \frac 1 2 u_i(a^1,a^1_{-i}) > 0.
		\end{align*}
		For the last line we used that $(n-1) (1-a^1)(a^1-c) > 0$, $- (1-a^k)(a^k-c) \geq -\Delta$, and $p < \frac{1}{2 \Delta} (1-a^1)(a^1-c)$.
		\item If $p > \frac{1}{2 \Delta} (1-a^1)(a^1-c)$, we choose $a_i' = a^{k-1}$:
		\begin{align*}
			\E_{-i} \left[ u_i(a^{k-1}, a_{-i}) - u_i(a^k, a_{-i}) \right] 
			&\geq p \left[  u_i(a^{k-1}, a^k_{-i}) - u_i(a^k, a^k_{-i}) \right] 
			\geq p \tfrac{D}{n} (a^k-a^{k-1})(a^k-c) \\
			&> \tfrac{D}{2 n \Delta} (1-a^1)(a^1-c) (a^k-a^{k-1})(a^k-c) \\
			&\geq \frac 1 2 u_i(a^1_i,a^1_{-i}) \cdot \frac 1 \Delta \min_k (a^k-a^{k-1})(a^k-c) \\
			&> \frac 1 2 u_i(a^1,a^1_{-i}) \cdot \frac 1 \Delta \min_k (a^k-a^{k-1})^2 > 0.
		\end{align*}
		In the first line, we ignored instances where all opponents might have played $a^{k-1}$ and used the inequality we derived above to show uniqueness of the NE. In the last line we used $a^k-c > a^{k-1} - c \geq a^k - a^{k-1}$.
	\end{enumerate} 
	Thus, for all mixed strategies of the opponents within the given support, there always exists an action $a_i'$ for which the utility is strictly increased. As a consequence, $a^k$ is not stable under $R_i$ and the supports cannot be correlated rationalizable.
\end{proof}

\begin{proposition}
	\label{prop:linear-demand-correlated-rationalizable-actions}    
	For a symmetric Bertrand competition game with linear demand, $\gamma < \beta$, and action space $\Acal = \{a^1, \dots, a^\numact \}$ where $0 \leq a^1 \leq \dots, < a^{k^*} \leq p^* < a^{k^* + 1} \dots < a^\numact$ with $p^* = \frac{\alpha + \beta c}{2 \beta - \gamma}$, all correlated rationalizable sets are subsets of $\{a^{k^*}, a^{k^* + 1}\}$ for every player $i$.
\end{proposition}

\begin{proof}
	We will use the Nash equilibrium of the continuous game $p^* = \tfrac{\alpha + \beta c}{2 \beta - \gamma}$ as a reference point and show that only the neighboring points of $p^*$, i.e., $a^{k^*}$ or $a^{k^*+1}$ can be stable under $R$.
	Using that the utility $u_i(a_i,a_{-i})$ is linear in $a_{-i}$, we can write the difference in the utility between some action $a$ and $a'$ for agent $i$ as 
	\begin{align*}
		\E_{a_{-i} \sim x_{-i}} \left[ u_i(a', a_{-i}) - u_i(a, a_{-i}) \right]
		&=  u_i(a', \E_{a_{-i} \sim x_{-i}} \left[ a_{-i} \right] )- u_i(a, \E_{a_{-i} \sim x_{-i}} \left[ a_{-i} \right] )\\
		&= (\alpha - \beta a' + \gamma \bar a_{-i})(a'-c) - (\alpha - \beta a + \gamma \bar a_{-i})(a-c) \\
		&= (a'- a)(\alpha + \beta c - \beta(a' + a) + \gamma  \bar a_{-i}),
	\end{align*}
	where $\bar a_{-i} := \E_{a_{-i} \sim x_{-i}} [ \tfrac{1}{n-1} \sum_{j\neq i} a_j ] $ is the expected average action of the opponents. Using the definition of $p^*$, we obtain
	\begin{equation} \label{eq:diff_utility_linear_demand}
		u_i(a', x_{-i}) - u_i(a, x_{-i}) = (a'-a) \ll \beta(p^*-a') + \beta(p^* - a) - \gamma (p^* - \bar a_{-i}) \rr.
	\end{equation}
	
	We show that deviating from an action $a  \notin \{ a^{k^*}, a^{k^*+1} \}$ to one of these actions is always better if $a$ is the lowest or highest action of any support. Therefore, an action outside of $\{ a^{k^*}, a^{k^*+1} \}$ can never be correlated rationalizable. We have to consider two different cases:
	Given an action profile $a$
	\begin{enumerate}[label=(\roman*)]
		\item let agent $i$ be the one with the highest support up to action $a$ with $a > a^{k^*+1}>p^*$:
		\begin{align}
			u_i(a^{k^*+1},x_{-i}) - u_i(a,x_{-i})  
			&\stackrel{\eqref{eq:diff_utility_linear_demand}}{=}  (a^{k^*+1}-a) \ll \beta(p^*-a^{k^*+1}) + \beta(p^* - a) - \gamma (p^* - \bar a_{-i}) \rr \notag \\ 
			&\geq (a^{k^*+1}-a)(\beta(p^*-a^{k^*+1}) + (\beta-\gamma)(p^*-a)) \notag \\
			&\geq (a^{k^*+1}-a^{k^*+2})\beta(p^*-a^{k^*+1}) > 0. \label{eq:ineq1_linear_demand}
		\end{align}
		\item let agent $i$ be the one with the lowest support up to action $a$ with $a < a^{k^*} \leq p^*$:
		\begin{align}
			u_i(a^{k^*},x_{-i}) - u_i(a,x_{-i})  
			&\stackrel{\eqref{eq:diff_utility_linear_demand}}{=} (a^{k^*}-a) \ll \beta(p^*-a^{k^*}) + \beta(p^* - a) - \gamma (p^* - \bar a_{-i}) \rr \notag \\
			&\geq (a^{k^*}-a)(\beta(p^*-a^{k^*}) + (\beta-\gamma)(p^*-a)) \notag \\
			&\geq(a^{k^*}-a^{k^*-1})(\beta - \gamma)(p^*-a^{k^*-1}) > 0. \label{eq:ineq2_linear_demand}
		\end{align}    
	\end{enumerate}
	In summary, this shows that for any choice of, potentially mixed, opponent strategy with support of at most/at least $a$, there is an action that is better than $a$. Thus, $a$ is not stable under $R_i$ and cannot be part of any correlated rationalizable set.
	
\end{proof}

\begin{proposition}
	\label{prop:logit-demand-correlated-rationalizable-actions}    
	For a symmetric Bertrand competition game with logit demand where $0 \leq a^1 \leq \dots, < a^{k^{*}} < \dots < a^{k^{**}} < \dots < a^\numact$, all correlated rationalizable sets are subsets of $\{a^{k^*}, \dots, a^{k^{**}}\}$ for every player $i$.
\end{proposition}

\begin{proof}
	The continuous Bertrand competition with logit demand is log-supermodular, i.e., if we transform the utilities by the application of a logarithm, the transformed game is supermodular \citep{milgrom1990rationalizability}. 	
	When we discretize the game, this property remains valid in the sense of \cite{milgrom1990rationalizability}.
	The log-transformed game is thus a game with strategic complementarities, a generalization of supermodularity, as defined in \cite{milgrom_monotone_1994}. 
	Since the concept of strategic complementarities is only based on ordinal relations between the actions and rewards and because the log-transformation is monotone, the original game also has this property. 
	\cite[Thm.~12]{milgrom_monotone_1994} shows that the game thus has smallest ($a^{k^*}$) and largest ($a^{k^{**}}$) serially strictly undominated strategies, which are both Nash equilibria. For this reason, the set of correlated rationalizable actions (being equivalent to SSU) is tight to the set of discrete equilibria.
\end{proof}

\subsection{Convergence with Mean-based Algorithms}
\label{sec:proof of the theorems}

In this section, we finally provide the proof for our main statement: the convergence of mean-based algorithms to correlated rationalizable actions. Before we start with the formal derivation, we briefly discuss the steps of our proof.

\subsubsection{Proof Sketch}

We want to show that all agents play only correlated rationalizable actions with probability close to 1 as time goes to infinity. Towards this result, we do the following:
\begin{description}
	\item[Step 1.] (Proposition \ref{prop:individually-rational-responses}): We show that every mean-based agent tends to choose correlated rational responses to their competitors' past actions with high probability. In the context of Bertrand competition games, this means that an individual agent undercuts the opponents' prices if the prices are too high and increases prices if the prices are too low. 
	\item[Step 2.] (Proposition \ref{prop:mean-based-convergence}): By an inductive argument, we find that agents will settle on correlated rationalizable actions. Assume that all players have been playing a subset of the action space for a prolonged time. Every mean-based player will thus select actions that are correlated rational responses to their competitors' past actions. This will continue for an extended number of steps, which changes the set of past actions. As a result, the players adapt their strategy again by selecting correlated rational responses to the new history. By noting that such a behavior of repeated correlated rational responses will lead to a set of correlated rationalizable responses, we can establish our result. Since mean-based algorithms do not "exploit" all of the time but sometimes choose "exploratory" actions at random, we must rely on probabilistic arguments to show that this happens with probability one.
	\item[Step 3.] 
	(Propositions \ref{prop:individually-rational-responses-stochastic} and \ref{prop:staggered-entry}): Finally, we extend these results by providing variants of Proposition \ref{prop:individually-rational-responses} that extend the results to noisy or stochastic payoffs and staggered entry.
\end{description}

\subsubsection{Notation and Formal Setting}

Our argument is of probabilistic nature. Let $(\Omega, \mathcal{F}, \prob)$ be a probability space where $\Omega$ contains all possible sequences of actions and rewards, including counterfactual rewards, $\mathcal{F}$ is the power set of these sequences, and $\prob$ is a probability measure that is induced by the agents' algorithms and the reward distributions\footnote{We cover stochastic rewards in Section \ref{sec:stochastic-utility}.}. On this probability space, we consider a stochastic process $\{a(t, i, \omega): t \in \mathbb{N}, ~ i = 1, \dots, n, ~ \omega \in \Omega\}$, in which $a(t, i, \omega)$ represents the action of player $i$ at time $t$ for some random outcome $\omega$. For a simple and consistent notation, we let $a_{i,t} = a(t, i)$ in the sections below. 
The exclusion of rewards in the stochastic process has technical reasons (see Section \ref{sec:stochastic-utility}), but it does not limit us in our ability to make statements about the convergence of the actions. Note that, in the setting that we cover first and where rewards are deterministic, (realized and counterfactual) rewards can be recovered exactly from the actions of all players. 

Based on this stochastic process, we define a filtration $\{\mathcal{F}_t\}_{t = 0}^{\infty}$ as follows: $\mathcal{F}_0 = \{\emptyset, \Omega\}$ and $\mathcal{F}_t = \sigma\{(a_{1,s}, \dots, a_{n, s}), s = 1, \dots, t\}$ is the $\sigma$-algebra generated by the actions up to time $t$. The stochastic process thus is adapted to this filtration. Naturally, we define the $\mathcal{F}_{t - 1}$-measurable random variable $H_{t - 1} = ((a_{1, s}, \dots, a_{n, s}))_{s = 1}^{t - 1}$ to be the history of all actions up to time $t - 1$. 
In summary, our derivations will rely on the filtered probability space $(\Omega, \mathcal{F}, \{\mathcal{F}_t\}, \prob)$.

Recall that we defined the series of correlated rational response sets $\Acal^0, \dots, \Acal^k, \dots$ as $\Acal^{k} = R(\Acal^{k-1})$ for all $k \geq 1$ and $\Acal^0 = \Acal$. Let $K$ be the number for which $\Acal^{K}$ is a strict subset of $\Acal^{K-1}$ and $\Acal^K = \Acal^{K+1} = \bar{\Acal}$. Because of remarks \ref{prop:rationalizable-subset-relation} and \ref{rem:rationalizable-fix-points} this $K$ always exists. To prove Theorem \ref{thm:mean-based-convergence}, we show that mean-based algorithms converge along this series of support sets with high probability, approaching $\Acal^K = \bar{\Acal}$ in the limit. Also, let $\numact$ be the maximum cardinality of any action space $\Acal_1, \dots, \Acal_\n$.

\subsubsection{Step 1: Individual Players Play Correlated Rational Responses}
\label{section:proof-step-mean-based-properties}
In the first step, we reason that individual players with mean-based algorithms will choose correlated rational responses.
To do so, we look at the history of past actions $H_{t-1}$ of all players. 
If the opponents have mainly selected actions within $\Acal^k_{-i}$, we show that the player is likely to choose actions within $\Acal^{k+1}_i = R_i(\Acal^k)$.

We denote the empirical frequency of opponents ($-i$) selecting actions in $\Acal^k_{-i}$ up to time $t$ by
\begin{equation*}
	P_{i,t}(\Acal^k) := \frac{1}{t} \sum_{s = 1}^{t} \I \left[\forall j \neq i: ~ a_{j, s} \in \Acal^k_j \right].
\end{equation*}
Now we can state the proposition:
\begin{proposition}
	\label{prop:individually-rational-responses}
	Given $k \in \{0, \dots, K-1 \}$ and $\gamma_t < \frac{\delta}{2}$ where $\delta$ is the competition constant of the game.
	If $P_{i, t-1}(\Acal^k)$ is sufficiently large for a history of action $h = ((a_{1,s}, \dots, a_{n, s}))_{s = 1}^{t - 1}$, then $\prob_{i, t}(a \mid H_{t-1} = h) \leq \gamma_t$ for all $a \not\in \Acal^{k+1}_i$.
\end{proposition}
\noindent By $\prob_{i,t}(a)$ we denote the probability that player $i$ with $\gamma_t$-mean-based algorithm chooses action $a$ at time $t$.

\begin{proof}
	We say that the empirical frequency $P_{i, t-1}(\Acal^k) $ is sufficiently large if 
	\begin{equation} \label{eq:p_tilde}
		P_{i, t-1}(\Acal^k) \geq \tilde p := \frac{1}{2} \ll \frac{\Delta u}{\delta + \Delta u} + 1 \rr \in (0, 1),
	\end{equation}
	where $\delta$ is the competition constant of the game and $\Delta u$ is the maximum absolute utility difference between any two actions of player $i$ for any opponent profile $a_{-i}$.
	Let $a' \in \Acal^{k+1}_i$ be an action which increases player $i$'s utility by at least $\delta$ compared to playing action $a \not\in \Acal^{k+1}_i$ if opponents play a mixed strategy with support on $\Acal^k_{-i}$. 
	Consider the counterfactual advantage of action $a'$ over $a$, which we can compute based on the actions given in $h$:
	\begin{equation*}
		\begin{aligned}
			\alpha_{i, t}(a') - \alpha_{i, t}(a)
			&:= \frac{1}{t-1} \sum_{s = 1}^{t-1} u_i(a', a_{-i,s}) - u_i(a, a_{-i,s}) \\
			&= \frac{1}{t-1} \sum_{\substack{1 \leq s \leq t -1: \\ a_{-i, s} \not \in \Acal_{-i}^k}} \underbrace{u_i(a', a_{-i,s}) - u_i(a, a_{-i,s})}_{\geq - \Delta u} 
			+ \frac{1}{t-1} \sum_{\substack{1 \leq s \leq t - 1: \\ a_{-i,s} \in \Acal_{-i}^k}} u_i(a', a_{-i,s}) - u_i(a, a_{-i,s}) \\
			&\geq (1 - P_{i, t-1}(\Acal^k_{-i})) \cdot (- \Delta u)
			+  P_{i, t-1}(\Acal^k) \cdot \underbrace{\E_{a_{-i} \sim x_{-i}} \left[u_i(a', a_{-i}) - u_i(a, a_{-i}) \right]}_{> \delta} \\
			&> (1 - P_{i, t-1}(\Acal^k)) (- \Delta u) + \delta P_{i, t-1}(\Acal^k) 
			= P_{i, t-1}(\Acal^k) (\delta + \Delta u) - \Delta u \\
			&\geq \tilde{p} (\delta + \Delta u) - \Delta u 
			= \frac{1}{2} (2\Delta u + \delta) - \Delta u 
			= \frac{\delta}{2} 
			> \gamma_t.
		\end{aligned}
	\end{equation*}
	The third line interprets the sum over all times for which the opponents play $a_{-i} \in \Acal^k_{-i}$ as an expectation over a joint action distribution $x_{-i} \in \Delta(\Acal^{k}_{-i})$ with support on $\Acal^k_{-i}$. We choose $a' \in \Acal^{k+1}_i$ such that it performs better than $a$ by a margin of at least $\delta$. For $a \not \in \Acal^{k+1}_i$, this action exists for any such mixed strategy, according to Equation \eqref{eq:irrational-actions}.
	The last line implies that the probability $\prob_{i, t}(a \mid h)$ that the mean-based algorithm chooses actions $a \not\in \Acal^{k+1}_i$ is small ($\leq \gamma_t$).
\end{proof}

\subsubsection{Step 2: High-probability Induction}
\label{section:proof-step-high-prob-induction}

In the second step, we use the insights from Proposition \ref{prop:individually-rational-responses} to make a statement about all players. 
Our inductive argument is based on the following observation. Suppose all agents played actions in $\Acal^{k}$ at time $t$ with high frequency. 
Then, each of the competitors will tend to play actions within $\Acal^{k+1}$ in the next step, thus increasing $P_{i, t-1}(\Acal^k)$ further (recall that $\Acal^{k+1} \subseteq \Acal^k$). At the same time, $P_{i, t-1}(\Acal^{k+1})$ will also increase.
Eventually, $P_{i, t^\prime-1}(\Acal^{k+1})$ will exceed some threshold at time $t^\prime$ with high probability. 
This induces the next step towards $\Acal^{k+2}$, and so on, until the algorithms reach $\Acal^K$.

Instead of looking at $P_{i, t-1}(\Acal^k)$ for individual players, we will now move on to looking at the frequency at which all players (not only the opponents) played within $\Acal^k$. 
We capture that this frequency is high by a series of events
\begin{equation*}
	A_k := \left[ \frac{1}{T_k} \sum_{t = 1}^{T_k} \I \left[ \exists i: ~ a_{i, t} \not\in \Acal^k_i \right] \leq \beta \right] = \left[ \frac{1}{T_k} \sum_{t = 1}^{T_k} \I \left[ \forall i: ~ a_{i, t} \in \Acal^k_i \right] \geq 1 - \beta \right].
\end{equation*}
These events are defined for some small $\beta > 0$ and for a corresponding series of time steps $T_0 < \dots < T_k < \dots < T_K$, and they are $\mathcal{F}_{T_k}$-measurable. 
In words, they capture that only in rare cases (fraction $\beta$ of all timesteps up to $T_k$), any agent chooses an action outside $\Acal^k$. Equivalently, this means that in most cases (fraction $1 - \beta$), all agents played actions from $\Acal^k$. 

We note the meaning of the event $A_K$ with respect to time-average convergence. For $\beta \to 0$, this event captures that all agents played within $\Acal^K$, i.e., the set of rationalizable actions, for a fraction of the time close to 1. 
If $\beta \to 0$ as $T_K \to \infty$, \textit{this implies time-average convergence.} 
The following proposition makes a statement about the probability that this will happen.

\begin{proposition}
	\label{prop:mean-based-convergence}
	Assume that all agents follow a $\gamma_t$-mean-based algorithm with respect to a monotonously decreasing sequence $\gamma_t$ such that $\gamma_t \to 0$ as $t \to \infty$.
	For every sufficiently small $\beta > 0$, there exists a rational constant $ c > 1$ such that
	\begin{equation*}
		\prob(A_K) \geq 1 - \sum_{k = 0}^{K - 1} \sum_{j = 1}^J \exp \left( \frac{- (T^j_k - T^{j-1}_k) \beta^2}{32} \right),
	\end{equation*}
	where $T_k^j = c^{j + Jk} T_0$ and $T_k^0 = T_k$ with a sufficiently large $T_0 \in \N$ and $J = \lceil \log_{c} \tfrac 2 \beta \rceil$.\\
	Moreover, $\prob(A_K)$ approaches 1 as $T_0 \to \infty$.
\end{proposition}

\begin{proof}
	For a sufficiently small $\beta$, i.e., $0 < \beta < 1 - \tilde p$ (as defined in Eq. \ref{eq:p_tilde}), we choose the constant $c \in \Q$ such that $1 < c < \tfrac{1-\beta}{\tilde p} $ and a sufficiently large $T_0$ such that $\gamma_t < \min \{\tfrac{\delta}{2}, \frac{\beta}{4nm}\}$ for all $t \geq T_0$, where $\delta$ is the competition constant of the game (see Section \ref{sec:remarks-correlated-rationalizability}), and $\Delta u$ is the maximium absolute utility difference between any two action profiles.
	Additionally, (for technical reasons) we choose $T_0$ such that $T^j_k = c^{j + Jk} T_0 \in \N$ for all $j = 0,\dots, J$ and $k = 0,\dots,K-1$.
	
	Before we can prove the result, we have to introduce some sub-events we use in the following:
	\begin{align}
		B_k^j &:= \left[ \frac{1}{T^j_k - T^{j-1}_k} \sum_{t = T^{j-1}_k+1}^{T^j_k} \I[\exists i: ~ a_{i, t} \not\in \Acal^{k+1}_i ] \leq \frac{\beta}{2} \right] \label{eq:subevent},\\ 
		A_k^j &:= (A_0, \dots, A_{k}, B_{k}^1, \dots, B_{k}^j).\label{eq:combined_events}
	\end{align}
	Note that $B_k^j$ is a somehow stricter version of the events $A_k$ since $\tfrac \beta 2 < \beta$ and $\Acal^{k+1} \subseteq \Acal^k$. Again, the events $B_k^j$ are $\mathcal{F}_{T_k^j}$-measurable. We can also interpret the events $A_k$, $A_k^j$, and $B_k^j$ as sets of "good" action sequences in the context of our filtered probability space.
	
	Using these definitions and some intermediate results, which we will prove afterwards, we get:
	{
		\begin{align*}
			\prob(A_K)  
			&\geq \prob(A_K, A_{K-1}, \dots, A_{1}, A_{0}) \\
			&= \underbrace{\prob(A_{0})}_{=1} \cdot \prod_{k = 0}^{K - 1} \prob(A_{k+1} \mid A_{k}, \dots, A_{0}) \tag{Chain Rule}\\
			&\geq \prod_{k = 0}^{K - 1} \prob(A_{k}^J \mid A_{k}, \dots, A_{0})  \tag{Lemma \ref{lem:implication}}\\
			&= \prod_{k = 0}^{K - 1}  \prob(B_{k}^1, \dots, B_{k}^J \mid A_{k}, \dots, A_{0}) \tag{Equation \ref{eq:combined_events}}\\
			&= \prod_{k = 0}^{K - 1} \prod_{j = 1}^J \prob(B_{k}^j \mid  A_{k}, \dots, A_{0}, B_{k+1}^1, \dots,  B_{k+1}^{j-1} ) \tag{Chain Rule} \\
			&= \prod_{k = 0}^{K - 1} \prod_{j = 1}^J \prob(B_{k}^j \mid A_{k}^{j-1}) \tag{Equation \ref{eq:combined_events}} \\
			&\geq \prod_{k = 0}^{K - 1} \bigg[ 1 - \sum_{j = 1}^J \exp \bigg( \frac{-(T^j_k - T^{j-1}_k) \beta^2}{32} \bigg) \bigg] \tag{Weierstrass Product Inequality and Lemma \ref{lem:B-events-conditioned-on-A-events}} \\
			&\geq 1 - \sum_{k = 0}^{K - 1} \sum_{j = 1}^J \exp \bigg( \frac{- (T^j_k - T^{j-1}_k) \beta^2}{32} \bigg) \tag{Weierstrass Product Inequality}
		\end{align*}
	}
	
	Since $(T^j_k - T^{j-1}_k) = (c-1) c^{j - 1 + Jk} T_0 $, the term in the last line converges to $1$ for $T_0 \rightarrow \infty$.
	
\end{proof}

In the remaining part of this subsection, we will show the necessary steps used in the proof of Proposition \ref{prop:mean-based-convergence}.
\begin{lemma} \label{lem:implication}
	$A_{k}^J$ implies $A_{k+1}$, which means $\prob(A_{k+1}) \geq \prob(A_{k}^J)$
\end{lemma}

\begin{proof}
	If $A_{k}^J$ holds we know that $B^1_{k}, \dots, B^J_{k}$ hold, which means that the agents have played actions outside of $\Acal^{k+1}$ for a fraction of at most $\frac{\beta}{2}$ during $\{T_k^0+1, \dots, T_k^1\}, \dots, \{T_k^{J-1}+1, \dots, T_k^J \}$. Recall that $T_k^0 = T_k$ and $T_k^J = T_{k+1}$.
	Using this, we can bound the relevant sum for $A_{k+1}$:
	\begin{align*}
		\frac{1}{T_{k+1}} \sum_{t = 1}^{T_{k+1}} \I[\exists i: a_{i,t} \not\in \Acal^{k+1}_i] 
		&= \frac{1}{T_{k+1}} \sum_{t = 1}^{T_{k}} \I[\exists i: a_{i,t} \not\in \Acal^{k+1}_i] + \frac{1}{T_{k+1}} \sum_{t = T_k + 1}^{T_{k+1}} \I[\exists i: a_{i,t} \not\in \Acal^{k+1}_i] \\
		&\leq \frac{T_k}{T_{k+1}} + \frac{1}{T_{k+1}} \sum_{j=1}^J (T_k^j - T_k^{j-1}) \frac{\beta}{2} 
		= \frac{T_k}{T_{k+1}} +  \frac{T_{k+1} - T_k}{T_{k+1}} \frac{\beta}{2}  \\
		&= c^{-J} + (1-c^{-J}) \frac \beta 2 
		\leq c^{- \log_c\left( \frac{2}{\beta} \right) } + \frac{\beta}{2} = \beta.
	\end{align*}
	Therefore $A_{k}^J$ implies $A_{k+1}$.
\end{proof}

\begin{lemma}
	\label{lem:B-events-conditioned-on-A-events}
	\begin{equation*}
		\begin{aligned}
			\prob(B_k^j \mid A_k^{j-1}) \geq 1 - \exp \bigg( \frac{- (T^j_k - T^{j-1}_k) \beta^2}{32} \bigg), \qquad \qquad j = 1, \dots, J.
		\end{aligned}
	\end{equation*}
\end{lemma}

\begin{proof}
	Recall that $B_k^j$ defined in Equation \eqref{eq:subevent} is the event that agents play only a small fraction outside of $\Acal^{k+1}$ during the time steps $\{ T^{j-1}_k+1, \dots, T^j_k \}$ given some $j \in J$. To get a bound on the probability, we first introduce random variables that count if an agent plays outside this set:
	\begin{equation*}
		X_t := \I \left[ \exists i: ~ a_{i, t} \not\in \Acal^{k+1}_i \right], \qquad t \in \{ T^{j-1}_k+1, \dots, T^j_k \}.
	\end{equation*}
	We can show that the expectation of $ X_t$ conditioned on a specific history $h \in A_{k}^{j-1}$ can be bounded by
	\begin{align*}
		\E[X_t \mid H_{t-1} = h] &= \prob(\exists i: ~ a_{i, t} \not\in \Acal^{k+1}_i \mid H_{t-1} = h) \\
		&\leq  n \cdot \prob(a_{i, t} \not\in \Acal^{k+1}_i \mid H_{t-1} = h) \tag{Union Bound Agents} \\
		&\leq  n \cdot \numact \cdot \prob(a_{i,t} = a \notin \Acal^{k+1}_i \mid H_{t-1} = h) \tag{Union Bound Actions} \\
		&\leq n \cdot \numact \cdot \gamma_t \leq n \cdot \numact \cdot  \gamma_{T_k^{j-1}}. \tag{Proposition \ref{prop:individually-rational-responses}} 
	\end{align*}
	Note that Proposition \ref{prop:individually-rational-responses} is applicable because of Lemma \ref{lem:condition-holds-in-gamma-intervals} (below).  Recall that $\numact$ was the number of elements in the largest action set, so this bound holds for any player.	
	The bound conditioned on the individual histories also implies that 
	\begin{equation*}
		\E[X_t \mid H_{t-1} \in A_k^{j - 1}] \leq n \cdot \numact \cdot  \gamma_{T_k^{j-1}}.
	\end{equation*}
	
	With this upper bound, we can use the random variables $ X_t$ to define
	\begin{equation}
		Z_t := \sum_{s = T_k^{j-1} + 1}^{t} \left( X_s - n \numact \gamma_{T_k^{j-1}} \right) ,  \qquad t \in \{ T^{j-1}_k+1, \dots, T^j_k \}.
	\end{equation}
	With $Z_{T_k^{j - 1}} := 0$ and under the condition of $A_{k}^{j-1}$, the value of $Z_t$ decreases with $t$ in expectation as we always subtract an upper bound of the expectation of each summand $X_s$, i.e., $\E[Z_t \vert Z_{t-1}, A_{k}^{j-1}] \leq Z_{t-1}$.
	This makes $Z_t$ a super-martingale and allows us to apply Azuma's inequality.\\
	First, we choose $T_0$ sufficiently large, such that $ \gamma_{T_k^{j-1}} \leq \beta \tfrac{1}{4\n\numact}$, where $n$ is the number of agents and $\numact$ is the maximum size of any action space. 
	This gives us $\vert Z_t - Z_{t-1} \vert = \vert X_t - \n \numact \gamma_{T_k^{j-1}} \vert \leq 1$ since $X_t \in \{0, 1\}$ and $0 \leq \n \numact \gamma_{T_k^{j-1}} \leq \beta \tfrac{\n\numact}{4\n\numact} \leq 1$.
	By  Azuma's inequality we get for all $\varepsilon > 0$
	\begin{align*}
		&\prob(Z_{T_k^{j}} - Z_{T_k^{j-1}} \geq \epsilon \mid A_{k}^{j-1}) \leq \exp \bigg( \frac{-\epsilon^2}{2 (T^j_k - T^{j-1}_k)} \bigg)\\
		\Leftrightarrow \quad &\prob(Z_{T_k^{j}} - Z_{T_k^{j-1}} < \epsilon \mid A_{k}^{j-1}) \geq 1 - \exp \bigg( \frac{-\epsilon^2}{ 2  (T^j_k - T^{j-1}_k)} \bigg).
	\end{align*}
	Now we choose $\varepsilon = \frac{\beta}{4}(T^j_k - T^{j-1}_k) $ such that we end up with the event described by $B_k^j$:
	\begin{align*}
		Z_{T_k^{j}} - Z_{T_k^{j-1}} &< \varepsilon \\
		\Leftrightarrow \quad \sum_{t = T_k^{j-1} + 1}^{T_k^j} \left(  \I \left[ \exists i: ~ a_{i, t} \not\in \Acal^{k+1}_i \right] - n \numact \gamma_{T_k^{j-1}} \right)  &< \frac{\beta}{4}(T^j_k - T^{j-1}_k) \\
		\Leftrightarrow \quad \frac{1}{T^j_k - T^{j-1}_k}\sum_{t = T_k^{j-1} + 1}^{T_k^j}  \I \left[ \exists i: ~ a_{i, t} \not\in \Acal^{k+1}_i \right] &<  n \numact \gamma_{T_k^{j-1}} + \frac{\beta}{4}
	\end{align*}
	Since $ \gamma_{T_k^{j-1}} \leq \beta \tfrac{1}{4\n\numact}$, the right hand side is less or equal to $\frac \beta 2$ and satisfies thereby $B_k^j$. Finally, we have 
	\begin{equation*}
		\prob(B_k^j \mid A_k^{j-1}) \geq \prob(Z_{T_k^{j}} - Z_{T_k^{j-1}} 
		< \tfrac{\beta}{4}(T^j_k - T^{j-1}_k) \mid A_k^{j-1}) 
		\geq 1 - \exp \bigg( \frac{- (T^j_k - T^{j-1}_k) \beta^2}{32} \bigg).        
	\end{equation*}
\end{proof}

To apply Proposition \ref{prop:individually-rational-responses} in the proof of Lemma \ref{lem:B-events-conditioned-on-A-events}, we need to make sure that the assumptions hold. This is ensured by the following lemma.
\begin{lemma}
	\label{lem:condition-holds-in-gamma-intervals}
	If $A_k^{j-1}$ holds, the condition $P_{t-1}^i(\Acal^k) > \Tilde{p}$ holds for all $i \in \players$, $t \in \{ T^{j-1}_k+1, \dots, T^j_k \}$, and any $j = 1, \dots, J$.
\end{lemma}

\begin{proof}
	Recall that given $\tilde p$ as defined in Eq.\eqref{eq:p_tilde}, the constants $\beta$ and $c$ in the previous proofs were chosen such that  $0 < \beta < 1- \tilde p$ and $1 < c < \tfrac{1-\beta}{\tilde p}$.
	\begin{align*}
		P^i_{t-1}(\Acal^k) 
		&= \frac{1}{t - 1} \sum_{s = 1}^{t - 1} \I[\forall j \neq i: ~ a_{j,s} \in \Acal^k_{j}] \\
		&= \frac{1}{t - 1} \bigg(
		\sum_{s = 1}^{T_k} \I[\forall j \neq i: ~ a_{j,s} \in \Acal^k_{j}]
		+ \sum_{s = T_k + 1}^{T_k^{j-1}} \I[\forall j \neq i: ~ a_{j,s} \in \Acal^k_{j}]
		+ \sum_{s = T_k^{j-1} + 1}^{t - 1} \I[\forall j \neq i: ~ a_{j,s} \in \Acal^k_{j}]
		\bigg) \\
		&\geq \frac{1}{t - 1} \bigg(
		\underbrace{
			\sum_{s = 1}^{T_k} \I[\forall j: ~ a_{j,s} \in \Acal^k_{j}]
		}_{\text{bounded using } A_k}
		+ \underbrace{
			\sum_{s = T_k + 1}^{T_k^{j-1}} \I[\forall j: ~ a_{j,s} \in \Acal^k_{j}]
		}_{\text{bounded using }(B_k^1, \dots, B_k^{j-1})}
		+ \underbrace{
			\sum_{s = T_k^{j-1} + 1}^{t - 1} \I[\forall j: ~ a_{j,s} \in \Acal^k_{j}]
		}_{\geq 0}
		\bigg) \\
		&\geq \frac{1}{t - 1} \big( T_k \cdot (1 - \beta) + (T_k^{j - 1} - T_k)(1-\tfrac \beta 2) \big) \\
		&> \frac{1}{t - 1} T_k^{j-1}(1 - \beta) \geq \frac{T_k^{j-1}}{T_k^j}(1-\beta) = \tfrac{1-\beta}{c} > \tilde p.
	\end{align*}
	In the last line, we used $1- \tfrac{\beta}{2} > 1 - \beta$ and $t-1 < T_k^j$.
\end{proof}

\subsubsection{Extension: Asymmetric mean-based algorithms}

Up until this point, we assumed that all algorithms are $\gamma_t$-mean-based for the same sequence $\gamma_t$. The following corollary extends our previous results to situations where this assumption does not hold.

\begin{corollary}
	Assume each agent $i$, $i \in \players$, follows a $\gamma_{i,t}$-mean-based algorithm with respect to a monotonically decreasing sequence $\gamma_{i,t}$ (in $t$) such that $\gamma_{i,t} \to 0$ as $t \to \infty$.
	For every sufficiently small $\beta > 0$, there exists a rational constant $ c > 1$ such that
	\begin{equation*}
		\prob(A_K) \geq 1 - \sum_{k = 0}^{K - 1} \sum_{j = 1}^J \exp \left( \frac{- (T^j_k - T^{j-1}_k) \beta^2}{32} \right),
	\end{equation*}
	where $T_k^j = c^{j + Jk} T_0$ and $T_k^0 = T_k$ with a sufficiently large $T_0 \in \N$ and $J = \lceil \log_{c} \tfrac 2 \beta \rceil$.\\
	Moreover, $\prob(A_K)$ approaches 1 as $T_0 \to \infty$.
\end{corollary}

\begin{proof}
	Let $\gamma_t := \max_{i \in \players} \gamma_{i,t}$ be the maximum $\gamma_{i,t}$ at every time step. The pointwise maximum of finitely many nonincreasing sequences is again nonincreasing, so $\gamma_t$ is monotonically nonincreasing with $\gamma_t \to 0$.
	Also, $\gamma_{i,t} \leq \gamma_t$, so each of the algorithms is also $\gamma_t$-mean-based according to Proposition \ref{prop:gamma_t-mean-based} below. The result follows from the application of Proposition \ref{prop:mean-based-convergence}.
\end{proof}

\begin{proposition}
	\label{prop:gamma_t-mean-based}
	Let ALG be a $\gamma_t$-mean-based algorithm. For any sequence $\gamma_t'$ such that $\gamma_t' \geq \gamma_t$ for all $t$, ALG is also $\gamma_t'$-mean-based. 
\end{proposition}

\begin{proof}
	We need to show that, whenever $\alpha_t(a') - \alpha_t(a) > \gamma_t'$, the algorithm selects action $a$ with probability less than or equal to $\gamma_t'$. 
	Assume $\alpha_t(a') - \alpha_t(a) > \gamma_t'$. This implies that $\alpha_t(a') - \alpha_t(a) > \gamma_t$. Since ALG is $\gamma_t$-mean-based, we know that $\prob(a_t = a) \leq \gamma_t$. Thus, $\prob(a_t = a) \leq \gamma_t'$. 
\end{proof}

\subsubsection{Extension: Stochastic utility}
\label{sec:stochastic-utility}

So far, we assumed that utilities are deterministic given the players' actions, giving rise to a game that we call "deterministic" in the following. We now drop this assumption and instead model the utilities (including \textit{counterfactual} utilities) as a stochastic process
\begin{equation*}
	\{ U(i, t, a, \omega): t \in \N, ~ i = 1, \dots, n, ~ a \in \Acal_i ~ \omega \in \Omega \}.
\end{equation*}
for our probability space $(\Omega, \mathcal{F}, \prob)$. $U(i, t, a)$ is the utility that player $i$ would have received at time $t$ if they had played action $a$. The opponents' actions are determined by $\omega$, unless they are specified explicitly by an event. This adds another layer of randomness to the payoffs. In general, there can be rich dependencies between the actions and rewards over time. However, given the concrete realizations $a_t := (a_{1, t}, \dots, a_{n,t})$, we assume that for all $i = 1, \dots, n$, $t \in \N$, and $a \in \Acal_i$,
\begin{enumerate}[label=\roman*)]
	\item the utility $U(i, t, a)$ is conditionally independent from the utilities of other times given $a_t$,
	\item its conditional expectation equals the deterministic game, $\E[U(i, t, a) \mid a_t] = u_i(a, a_{-i, t})$, and
	\item it is bounded, $U(i, t, a) \in [\underline u, \overline u]$, where $M := \overline u - \underline u$ denotes the maximal range.
\end{enumerate}
This definition of random payoffs allows for a variety of different payoff models in the oligopoly settings. For example, we can add a random, bounded noise to the payoffs or demands. This noise may also be dependent (but not correlated) among players as long as its conditional expectation is zero. Alternatively, we could define a model where the utility function determines the probability that an item is sold. Many models are possible, as long as the analysis of the game (i.e., its correlated actions and Nash equilibria) is based on the deterministic game.

In this setting, we provide an alternative version of Proposition \ref{prop:individually-rational-responses}.

\begin{proposition}
	\label{prop:individually-rational-responses-stochastic}
	Given $k \in \{0, \dots, K-1\}$ and $\gamma_t < \frac{\delta}{4}$ where $\delta$ is the competition constant of the game. If $P_{i, t-1}(\Acal^k)$ is sufficiently large for a history of actions $h = ((a_{1,s}, \dots, a_{n, s}))_{s = 1}^{t - 1}$, then 
	$$
	\prob_{i,t}(a \mid H_{t-1} = h) \leq \tilde{\gamma}_t := \gamma_t + \exp \left(- \frac{(t-1) \delta^2}{32 (b - a)^2} \right)
	$$ 
	for all $a \notin \Acal_i^{k+1}$.
\end{proposition}
\noindent Remember that, by $\prob_{i,t}(a \mid H_{t-1} = h)$, we denote the probability that player $i$ with $\gamma_t$-mean-based algorithm chooses action $a$ at time $t$, given that the previous actions were as described by $h$. This alternative version of the proposition can still be used with the proofs in section \ref{section:proof-step-high-prob-induction} by replacing $\gamma_t$ with $\tilde{\gamma}_t$ (a monotonically decreasing sequence) and choosing $T_0$ sufficiently large such that $\gamma_{T_0} < \frac{\delta}{4}$.

\begin{proof}
	We again say that the empirical frequency $P_{i, t-1}(\Acal^k) $ is sufficiently large if 
	\begin{equation*}
		P_{i, t-1}(\Acal^k) \geq \tilde p := \frac{1}{2} \ll \frac{\Delta u}{\delta + \Delta u} + 1 \rr \in (0, 1)
	\end{equation*}
	where $\delta$ is the competition constant of the game and $\Delta u$ is the maximum absolute utility difference between any two actions of player $i$ for any opponent profile $a_{-i}$.
	
	Recall that $a \notin \Acal^{k+1}_i$, and let $a' \in \Acal^{k+1}_i$ be an action which increases player $i$'s utility by at least $\delta$ \textit{in the deterministic game} compared to playing action $a$ if opponents play a mixed strategy with support on $\Acal^k_{-i}$. We want to bound $\prob_{i,t}(a \mid H_{t-1} = h)$.
	
	Consider the (random) counterfactual advantage of action $a'$ over $a$,
	\begin{equation*}
		\alpha_{i, t}(a') - \alpha_{i, t}(a) := \frac{1}{t-1} \sum_{s = 1}^{t-1} U(i, s, a') - U(i, s, a),
	\end{equation*}
	and let 
	\begin{equation*}
		D: \quad \alpha_{i, t}(a') - \alpha_{i, t}(a) \geq \frac{1}{t-1} \sum_{s = 1}^{t-1} [u_i(a', a_{-i,s}) - u_i(a, a_{-i,s})] - \frac{\delta}{4}
	\end{equation*}
	denote the event that it does not undercut the \textit{deterministic} advantage by more than $\delta/4$.
	In the following, we will derive our desired bound by splitting the probability $\prob_{i,t}(a \mid H_{t-1} = h)$ into
	\begin{equation*}
		\begin{aligned}
			\prob_{i,t}(a \mid H_{t-1} = h) 
			&= \prob_{i,t}(a, \neg D \mid H_{t-1} = h) + \prob_{i,t}(a, D \mid H_{t-1} = h) \\
			&\leq \prob_{i, t}(\neg D \mid H_{t-1} = h) + \prob_{i,t}(a \mid D, H_{t-1} = h) 
		\end{aligned}
	\end{equation*}
	and bounding the two probability terms.
	
	Let us start with the first term.
	Since we condition on $h$, the random variables $U(i, s, a') - U(i, s, a)$, $s = 1, \dots, t-1$, have expectation $u_i(a', a_{-i,s}) - u_i(a, a_{-i,s})$, they are independent, and they are bounded by $M$ from above and $-M$ from below.
	By applying Hoeffding's inequality, we find that
	\begin{align*}
		\prob_{i, t}(\neg D \mid H_{t-1} = h)
		&= \prob_{i, t} \left( \sum_{s = 1}^{t-1} [U(i, s, a') - U(i, s, a)] - \sum_{s = 1}^{t-1} [u_i(a', a_{-i,1}) - u_i(a, a_{-i,s})] \leq  - (t-1) \frac{\delta}{4} \right) \\
		&\leq \exp \left(- \frac{2 (t-1)^2 \frac{\delta^2}{16}}{(t-1) (2M)^2} \right) 
		= \exp \left(- \frac{(t-1) \delta^2}{32 M^2} \right)
	\end{align*}
	
	Now, we bound the second term. Again, we consider the counterfactual utility advantage of $a'$ over $a$. We show that it exceeds $\gamma_t$, given that $D$ holds.
	\begin{equation*}
		\begin{aligned}
			\alpha_{i, t}(a') - \alpha_{i, t}(a)
			&= \frac{1}{t-1} \sum_{s = 1}^{t-1} U(i, s, a') - U(i, s, a) \\
			&\geq \frac{1}{t-1} \sum_{s = 1}^{t-1} u_i(a', a_{-i,s}) - u_i(a, a_{-i,s}) - \frac{\delta}{4} \\
			&= \frac{1}{t-1} \sum_{\substack{1 \leq s \leq t -1: \\ a_{-i, s} \not \in \Acal_{-i}^k}} \underbrace{u_i(a', a_{-i,s}) - u_i(a, a_{-i,s})}_{\geq - \Delta u} 
			+ \frac{1}{t-1} \sum_{\substack{1 \leq s \leq t - 1: \\ a_{-i,s} \in \Acal_{-i}^k}} u_i(a', a_{-i,s}) - u_i(a, a_{-i,s}) - \frac{\delta}{4} \\
			&\geq (1 - P_{i, t-1}(\Acal^k)) \cdot (- \Delta u)
			+  P_{i, t-1}(\Acal^k) \cdot \underbrace{\E_{a_{-i} \sim x_{-i}} \left[u_i(a', a_{-i}) - u_i(a, a_{-i}) \right]}_{> \delta}  - \frac{\delta}{4} \\
			&> (1 - P_{i, t-1}(\Acal^k)) (- \Delta u) + \delta P_{i, t - 1}(\Acal^k)  - \frac{\delta}{4}
			= P_{i, t-1}(\Acal^k) (\delta + \Delta u) - \Delta u  - \frac{\delta}{4} \\
			&\geq \tilde{p} (\delta + \Delta u) - \Delta u  - \frac{\delta}{4}
			= \frac{1}{2} (2\Delta u + \delta) - \Delta u  - \frac{\delta}{4}
			= \frac{\delta}{4} 
			> \gamma_t
		\end{aligned}
	\end{equation*}
	The second line makes use of the event $D$, which allows us to relate the random variables with their expected outcomes.
	The fourth line interprets the sum over all times for which the opponents play $a_{-i} \in \Acal^k_{-i}$ as an expectation over a joint action distribution $x_{-i} \in \Delta(\Acal^{k}_{-i})$ with support on $\Acal^k_{-i}$. We have chosen $a' \in \Acal^{k+1}_i$ such that it performs better than $a$ by a margin of at least $\delta$ in the deterministic game. For $a \not \in \Acal^{k+1}_i$, this action exists for any such mixed strategy, according to Equation \eqref{eq:irrational-actions}.
	The last line implies that the probability $\prob_{i, t}(a)$ that the mean-based algorithm chooses actions $a \not\in \Acal^{k+1}_i$ is small:
	\begin{equation*}
		\prob_{i,t}(a \mid H_{t-1} = h, D) \leq \gamma_t
	\end{equation*}
	
	Combining the two bounds shows the result. 
\end{proof}

\begin{remark}
	We can replace the assumption of a bounded support for $U_{i,t}$ by the assumption that it has a finite variance. Using Chebyshev's inequality instead of Hoeffding's inequality would yield a similar result that would still show the main claim. 
\end{remark}

\subsubsection{Extension: Staggered Entry and Biases}
\label{sec:staggered-entry}

In the previous sections, we have assumed that all players start at the same time $t = 0$ and have made no experiences about the payoffs in advance. Now, we will show that we can drop this assumption and still show convergence. The intuition behind this is that all beliefs that have been formed up to some point will become negligible if the interaction continues sufficiently long. 

Consider the following notation. We assume that all agents will have entered the competition at time $t = 1$, meaning that the game will remain stable thereafter (no players leave or join after $t = 0$). However, each agent $i$ might have joined the competition at an earlier point $\tau_i \leq 1$. We will show that the agents (employing mean-based strategies) still converge to the correlated rationalizable actions of the game that is defined by the competition present at $t = 1$.

In this scenario, every agent has $t - \tau_i$ observations at time $t$. For $t = \tau_i, \dots, 0$, the agent observes rewards $\tilde{u}_{i,t}(\cdot, \tilde{a}_{-i,s})$ from a utility function that is potentially different from $u_{i,t}$ and that also depends on the potential actions of other competitors $\tilde{a}_{i, t}$ (these can be different competitors than at $t \geq 0$). We will assume that $\tilde{u}_{i,t}$ is bounded by the interval $[-c, c], c \in \R$. The value estimate of player $i$ for an action $a \in \Acal_i$ at time $t > \tau_i$ becomes
\begin{equation*}
	\alpha_{i,t}(a) = \frac{1}{t - \tau_i} \left( \sum_{s = \tau_i}^{0} \tilde{u}_{i,s}(a, \tilde{a}_{-i, s}) + \sum_{s = 1}^{t-1} u_{i,s}(a, a_{-i, s}) \right) = \frac{1}{t - \tau_i} \left( \tilde{b}_i(a) + \sum_{s = 1}^{t-1} u_{i,s}(a, a_{-i, s}) \right).
\end{equation*}
We have used $\tilde{b}_i(a) = \sum_{s = \tau_i}^{0} \tilde{u}_{i,s}(a, \tilde{a}_{-i, s})$ to denote player $i$'s biases of the game that they developed during rounds $\tau_i, \dots, 0$. Note that these biases can grow with $\tau_i$. These biases could, alternatively, be viewed as prior knowledge that an algorithm designer injects into an algorithm based on their beliefs about the competition.

We again provide an alternative version of Proposition \ref{prop:individually-rational-responses} that shows that an individual player will select a correlated rationalizable response as long as its opponents have played a certain set of actions sufficiently often. This is enough to show convergence to the set of correlated rationalizable actions as in Proposition \ref{prop:mean-based-convergence}.

Recall that we denote the empirical frequency of opponents ($-i$) selecting actions in $\Acal^k_{-i}$ up to time $t$ by
\begin{equation*}
	P_{i,t}(\Acal^k) := \frac{1}{t} \sum_{s = 1}^{t} \I \left[\forall j \neq i: ~ a_{j, s} \in \Acal^k_j \right].
\end{equation*}

\begin{proposition}
	\label{prop:staggered-entry}
	Given $k \in \{0, \dots, K-1 \}$ and $t \geq 1$ such that $\gamma_{(t - \tau_i)} < \frac{-2 c \abs{\tau_i - 1}}{t - \tau_i} + \frac{t-1}{t - \tau_i} \frac{\delta}{2}$, where $\delta$ is the competition constant of the game.
	If $P_{i, t-1}(\Acal^k)$ is sufficiently large for a history of actions $h$, then $\prob_{i, t}(a \mid H_{t-1} = h) \leq \gamma_{(t - \tau_i)}$ for all $a \not\in \Acal^{k+1}_i$.
\end{proposition}

\noindent Remember that, by $\prob_{i,t}(a \mid H_{t-1} = h)$, we denote the probability that player $i$ with $\gamma_t$-mean-based algorithm chooses action $a$ at time $t$. This alternative version can still be used with the proofs in section \ref{section:proof-step-high-prob-induction}: We can replace $\gamma_t$ by $\tilde{\gamma}_t := \max_{i \in \players} \gamma_{(t - \tau_i)}$ in the proof of Proposition \ref{prop:mean-based-convergence}, and we choose $T_0$ sufficiently large such that $\gamma_{(T_0 - \tau_i)} < \frac{-2 c \abs{\tau_i - 1}}{T_0 - \tau_i} + \frac{T_0-1}{T_0 - \tau_i}$ for all active players $i$.

\begin{proof}[Proof of Proposition \ref{prop:staggered-entry}.]
	We say that the empirical frequency $P_{i, t-1}(\Acal^k) $ is sufficiently large if 
	\begin{equation*}
		P_{i, t-1}(\Acal^k) \geq \tilde p := \frac{1}{2} \ll \frac{\Delta u}{\delta + \Delta u} + 1 \rr \in (0, 1)
	\end{equation*}
	where $\delta$ is the competition constant of the game and $\Delta u$ is the maximum absolute utility difference between any two actions of player $i$ for any opponent profile $a_{-i}$.
	Let $a' \in \Acal^{k+1}_i$ be an action which increases player $i$'s utility by at least $\delta$ compared to playing action $a \not\in \Acal^{k+1}_i$ if opponents play a mixed strategy with support on $\Acal^k_{-i}$. 
	Consider the counterfactual advantage of action $a'$ over $a$:
	\begin{equation*}
		\tiny
		\begin{aligned}
			\alpha_{i, t}(a') - \alpha_{i, t}(a)
			&:= \frac{1}{t- \tau_i } \left( \tilde{b}_i(a') - \tilde{b}_i(a) + \sum_{s = 1}^{t-1} u_i(a', a_{-i,s}) - u_i(a, a_{-i,s}) \right) \\
			&= \frac{\tilde{b}_i(a') - \tilde{b}_i(a)}{t- \tau_i} + \frac{t-1}{t - \tau_i} \left( \frac{1}{t-1} \sum_{\substack{1 \leq s \leq t -1: \\ a_{-i, s} \not \in \Acal_{-i}^k}} \underbrace{u_i(a', a_{-i,s}) - u_i(a, a_{-i,s})}_{\geq - \Delta u} 
			+ \frac{1}{t-1} \sum_{\substack{1 \leq s \leq t - 1: \\ a_{-i,s} \in \Acal_{-i}^k}} u_i(a', a_{-i,s}) - u_i(a, a_{-i,s}) \right) \\
			&\geq \frac{-2 c \abs{\tau_i - 1}}{t - \tau_i} + \frac{t-1}{t - \tau_i} \left( (1 - P_{i, t-1}(\Acal^k_{-i})) \cdot (- \Delta u)
			+  P_{i, t-1}(\Acal^k) \cdot \underbrace{\E_{a_{-i} \sim x_{-i}} \left[u_i(a', a_{-i}) - u_i(a, a_{-i}) \right]}_{> \delta} \right) \\
			&> \frac{-2 c \abs{\tau_i - 1}}{t - \tau_i} + \frac{t-1}{t - \tau_i} \left( (1 - P_{i, t-1}(\Acal^k)) (- \Delta u) + \delta P_{i, t - 1}(\Acal^k) \right)
			= \frac{-2 c \abs{\tau_i - 1}}{t - \tau_i} + \frac{t-1}{t - \tau_i} \left( P_{i, t-1}(\Acal^k) (\delta + \Delta u) - \Delta u \right) \\
			&\geq \frac{-2 c \abs{\tau_i - 1}}{t - \tau_i} + \frac{t-1}{t - \tau_i} \left( \tilde{p} (\delta + \Delta u) - \Delta u \right)
			= \frac{-2 c \abs{\tau_i - 1}}{t - \tau_i} + \frac{t-1}{t - \tau_i} \left( \frac{1}{2} (2\Delta u + \delta) - \Delta u \right) \\
			&= \frac{-2 c \abs{\tau_i - 1}}{t - \tau_i} + \frac{t-1}{t - \tau_i} \frac{\delta}{2} 
			> \gamma_{(t - \tau_i)}
		\end{aligned}
	\end{equation*}
	In contrast to the proof of Proposition \ref{prop:individually-rational-responses}, we need to take the biases into account. According to our assumption, we can bind the values for $\tilde{b}_i(a)$ and $\tilde{b}_i(a')$ by $-c \cdot (\tau_i - 1)$ and $c \cdot (\tau_i - 1)$, respectively. We also gain a factor $\frac{t - 1}{t - \tau_i}$.
	The series of values $\frac{-2 c \abs{\tau_i - 1}}{t - \tau_i} + \frac{t-1}{t - \tau_i} \frac{\delta}{2} , t \geq 0$ is strictly monotonously increasing and converges to a value of $\frac{\delta}{2}$. As $\gamma_t$ is monotonously decreasing and converges to zero, there exists $T \geq 0$ such that $\gamma_t$ falls below this bound for all $t \geq T$. \\
	The inequality implies that $\prob_{i, t}(a \mid H_{t-1} = h) \leq \gamma_{(t - \tau_i)}$, which is less than $\tilde{\gamma}_t$.
\end{proof}

\subsubsection{Formal Theorem}
\label{sec:formal-theorem}

Here, we state the formal version of \ref{thm:mean-based-convergence-informal} in the main part, which combines all three extensions introduced above. The  extensions can be combined into a single statement: 
\begin{itemize}
	\item We can use the piecewise maximum of $\gamma_{i,t}$ to unify the mean-based bound of all players. This does not affect stochastic utility or staggered entry.
	\item The bias terms of staggered entry can be included in the derivation of $\alpha_{i,t}(a') - \alpha_{i,t}(a)$ in the proof of Proposition \ref{prop:individually-rational-responses-stochastic}.
	The result, $\tilde{\gamma}_t$, is a scaled version of $\gamma_t$ with an additional term. However, it still monotonously decreases to zero, which is the only requirement for convergence used in the proof of Proposition \ref{prop:mean-based-convergence}.
\end{itemize}

All in all, we obtain the following main result:

\begin{theorem}
	\label{thm:mean-based-convergence}
	Let $n$ players, $i \in \players$, compete in a repeated game with stochastic (counterfactual) payoffs $U(i, t, a)$ as described in section \ref{sec:stochastic-utility}.
	Assume each agent follows a $\gamma_{i,t}$-mean-based algorithm with respect to a monotonically decreasing sequence $\gamma_{i,t}$ such that $\gamma_{i,t} \to 0$ as $t \to \infty$.
	Then, there exists a rational constant $ c > 1$ such that
	\begin{equation*}
		\prob \left( \frac{1}{T_K} \sum_{t = 1}^{T_K} \I \left[ \forall i: ~ a_{i, t} \in \bar{\Acal}_i \right] \geq 1 - \beta \right) \geq 1 - \sum_{k = 0}^{K - 1} \sum_{j = 1}^J \exp \left( \frac{- (T^j_k - T^{j-1}_k) \beta^2}{32} \right)
	\end{equation*}
	for every sufficiently small $\beta > 0$, where $T_k^j = c^{j + Jk} T_0$ and $T_k^0 = T_k$ with a sufficiently large $T_0 \in \N$ and $J = \lceil \log_{c} \tfrac 2 \beta \rceil$.\\
	Moreover, the probability approaches 1 as $T_0 \to \infty$, indicating time-average convergence to $\bar{\Acal} = \bar{\Acal}_1 \times \dots \times \bar{\Acal}_n$, the set of correlated rationalizable actions of the game with rewards $u_i$. 
	The result also holds if some or all of the agents have entered the game at different times $\tau_i \leq 1$.
\end{theorem}

\subsubsection{Last-iterate convergence}
\label{sec:last-iterate-convergence}

Finally, we also provide the proof for Corollary \ref{cor:last-iterate-convergence}. 

\begin{proof}[Proof of Corollary \ref{cor:last-iterate-convergence}.]
	From Lemma \ref{lem:condition-holds-in-gamma-intervals}, we get that $P_{t-1}(\Acal^{K-1}) \geq \tilde{p}$ for sufficiently large $t$. According to Proposition \ref{prop:individually-rational-responses} (Prop. \ref{prop:individually-rational-responses-stochastic}, \ref{prop:staggered-entry}), a player $i$ plays outside of $\bar{\Acal} = \Acal^K$ with probability at most $\gamma_t$ ($\tilde{\gamma}_t$). As the same argument holds for all players, the statement follows from $\gamma_t \to 0$ ($\tilde{\gamma}_t \to 0$) for $t \to \infty$. 
\end{proof}

\FloatBarrier

\section{Additional Experiments}
\label{sec:additional-experiments}

In this section, we provide a brief description of our algorithms and provide additional experiments that shed further light on the algorithmic interaction in oligopoly games.

\subsection{Online Optimization Algorithms}\label{app:algorithms}

\subsubsection{\texorpdfstring{$\epsilon$}{epsilon}-Greedy Bandit}
\label{sec:greedy-bandit}

One of the simplest algorithms is the $\epsilon$-greedy exploration strategy. 
With its internal belief, it keeps track of the empirical mean reward $\bar{u}_t(a)$ of each explored action $a$. With probability $1 - \epsilon$, it selects the most promising action in the current time step, assuming that each action has been chosen at least once. With probability $\epsilon$, it chooses an arbitrary action with uniform probability.
\begin{equation}
	\pi_t(a_t) = \begin{cases}
		1 - \epsilon \cdot \frac{K - 1}{K} & a_t = \arg\max_{a \in \mathcal{A}} \bar{u}_{t-1}(a) \\
		\frac{\epsilon}{K} & \text{else}
	\end{cases}
\end{equation}

\subsubsection{UCB Algorithms}
\label{sec:ucb}
Like $\epsilon$-greedy, the UCB algorithms estimate and update the empirical mean reward of all their actions. 
They balance exploration and exploitation by choosing arms according to a heuristic, the upper confidence bound (UCB) of estimated rewards. 
This value can be computed by adding an exploration bonus $\rho_t$ to the empirical mean of past rewards $\bar{u}_t$:
\begin{equation}
	\text{UCB}_{t}(a) = \bar{u}_{t}(a) + \rho_{t}(a)
\end{equation}
The exploration bonus may depend on past rewards, $u_s(a_s), ~ s \leq t$, the number of times an action was selected, $n_t(a)$, and the time, $t$.
In each step, the algorithm chooses the action $a_t$ which maximizes $\text{UCB}_{t-1}(a)$. 
Similar to \citet{hansen_frontiers_2021}, we use the UCB-Tuned algorithm (\cite{auer2002bandit}) in our experiments. This algorithm is a refined version of UCB1 (\cite{auer2002bandit}) with an exploration bonus 
\begin{equation*}
	\rho_{t}^{UCB-Tuned}(a) = \sqrt{\frac{\ln t}{n_t(a)} \min\left\{\dfrac{1}{4}, V_{t}(a)\right\}}
\end{equation*}
where $V_t(a)$ is a variance estimate:
\begin{equation*}
	V_{t}(a) = \left(\frac{1}{n_{t}(a)}\sum_{\substack{1 \leq s \leq t:\\a_s = a}} u_{s}(a)^2\right) - \bar{u}_{t}(a)^2 + \sqrt{\frac{2\ln t}{n_{t}(a)}}.
\end{equation*}

Unlike \citet{hansen_frontiers_2021}, our implementation of the algorithm does not eliminate any actions that it deems unpromising, and it resolves ties of the UCB heuristic randomly.
The algorithm has vanishing regret in the stochastic multiarm bandit problem under mild assumptions (see, e.g., \cite{lattimore2020bandit}) but may suffer linear regret in the adversarial setting.

\subsubsection{Exp3 Algorithms}
\label{sec:exp3}

Exp3 ("Exponential-weight algorithm for Exploration and Exploitation") is an algorithm introduced by \cite{auer_gambling_1995}, which was designed for the adversarial model. This means that we do not make any assumptions about the distribution of rewards. In contrast, the rewards could also be chosen by an adversary who has full information about the setting \emph{and the agent's algorithm} and who chooses rewards in each round. 

The Exp3 algorithm's action-selection rule assigns weights to each arm, and its update rule adjusts these weights based on the observed rewards. 
Following the slightly modified version by \citet{lattimore2020bandit}, our implementation of the Exp3 algorithm samples its actions $a_t$ from
\begin{equation}
	P_{t}(a) = (1 - \epsilon_t) \cdot \frac{\exp \left( \eta \sum_{s = 1}^{t-1}\hat{u}_{s}(a) \right)}{\sum_{\hat{a} \in \Acal}\exp \left( \eta \sum_{s = 1}^{t-1}\hat{u}_{s}(\hat{a}) \right)} + \frac{\epsilon_t}{\abs{\Acal}}.
\end{equation}
Here, $\epsilon_t$ is an additional exploration rate that encourages the algorithm to choose arbitrary actions from time to time and $\hat{u}_{s}(a)$ is an estimate of the reward. Let $P_s(a)$ be the probability that the algorithm chooses action $a$ at time $s$, given the history until then. Assuming rewards are bounded within $[0, 1]$, we choose
\begin{equation}
	\hat{u}_{s}(a) = 1 - \frac{\mathbb{I}\{ a_s = a\}}{P_{s}(a)} \cdot (1 - u_s).
\end{equation}

The Exp3 algorithm is known to be no-regret for some choices of $\epsilon_t$ in the adversarial setting (\cite{auer_gambling_1995, lattimore2020bandit}) and mean-based for a specific choice of exploration rate $\epsilon_t$ and learning rate $\eta_t$ \citep{braverman_selling_2018}. Our implementation relies on a fixed exploration and learning rate for simplicity.

\subsubsection{Thompson Sampling}
\label{sec:thompson-sampling}

Thompson Sampling is a Bayesian approach to multi-armed bandit algorithms that was named after William R. Thompson based on his papers from \citeyear{thompson_likelihood_1933} and \citeyear{thompson_theory_1935}. The heuristic maintains probability distributions over the rewards of every possible action and uses these to determine the next action.

More specifically, Thompson sampling assumes that the rewards of each action $a \in \Acal$ are sampled independently from a probability distribution $\prob(u_t = r \vert \theta_a)$.\footnote{This also generalizes to probability densities. For simplicity, we stick to the discrete formulation in this introduction.} Bayes' rule allows us to maintain a posterior distribution to update our beliefs on the (unknown) parameters $\theta_a$, and consequently on the distribution of rewards of each action:
\begin{equation*}
	\prob(\theta_a \vert \{u_s\}_{s = 1}^{t}, \theta_a^0) = \frac{\prob \ll \left. \{u_s\}_{1 \leq s \leq t : a_s = a} \right\vert \theta_a\rr \cdot \prob(\theta_a \vert \theta_a^0)}{\prob \ll\{u_s\}_{1 \leq s \leq t : a_s = a}\rr}
\end{equation*}
Here, $\prob(\theta_a \vert \theta_a^0)$ is a prior distribution over reward parameters for action $a$, and $\prob \ll \left. \{u_s\}_{1 \leq s \leq t : a_s = a} \right\vert \theta_a\rr$ is the probability of rewards $u_s$ at iterations for which the selected action was $a$.

In each iteration, Thompson sampling draws parameter samples $\theta_a$ for all actions from their corresponding posterior distribution, introducing randomness. Subsequently, the action with the highest expected reward, based on $\theta_a$, is chosen. After receiving reward feedback, the posterior distribution of the chosen action is updated according to Bayes' rule.

While \citet{thompson_likelihood_1933} described the approach for two actions with Bernoulli distributions, our setting has $K$ actions with continuous, bounded rewards. 
The Bayesian approach can also be used in this setting; a generalization that has been described, e.g., by \citet{russo_tutorial_2020}. 
Our implementation of Thompson sampling assumes Gaussian-distributed rewards\footnote{We note that this does not reflect that the rewards are bounded. The approach is not severely limited by this.} and poses a Gaussian distribution over the parameters as conjugated prior. 

\subsection{Convergence Speed}
\label{sec:convergence-speed}

Aside from the converged prices, we also analyzed the training behavior of the algorithms. As with the final prices, we find significant differences between the \textit{UCB-T} algorithms and the other algorithms, which we summarize as our third result. 

\begin{result}
	\textit{Exp3-$\epsilon$}, \textit{TS}, and \textit{$\epsilon$-Greedy} converge quickly and reliably to the Nash equilibrium prices.
	\textit{UCB-T} combinations erratically learn to charge higher prices with large differences within and between runs.
\end{result}

In Figures \ref{fig:price-evolution-1} and \ref{fig:price-evolution-2}, we provide the price-competition indices for all four algorithms in symmetric duopolies with linear demand (\textit{O2}). 
Indeed, we find that the convergence of \textit{Exp3-$\epsilon$} and \textit{TS} is particularly consistent and quick. The former algorithm sets prices close to the competitive Nash equilibrium after roughly 50,000 steps, while the latter requires even fewer interactions between the agents. 
With \textit{$\epsilon$-Greedy}, the agents arrive at competitive prices quickly, with small variations between runs. We attribute this to the fixed exploration rate of our algorithm configuration for \textit{$\epsilon$-Greedy}.


\begin{figure}
	\centering
	
	\hfill
	\begin{subfigure}{0.44\textwidth}
		\includegraphics[width=\textwidth]{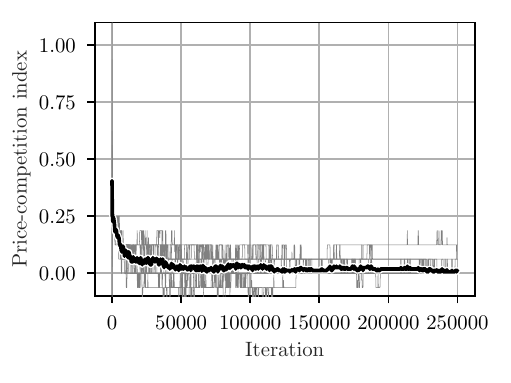}
		\subcaption{\textit{Exp3-$\epsilon$}}
	\end{subfigure}
	\hfill
	\begin{subfigure}{0.44\textwidth}
		\includegraphics[width=\textwidth]{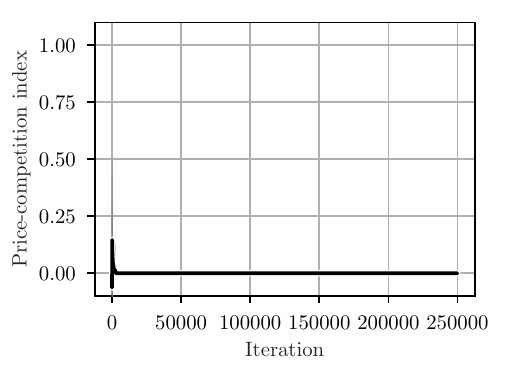}
		\subcaption{\textit{TS}}
	\end{subfigure}
	\hfill
	
	\caption{Evolution of Prices During Training}
	\label{fig:price-evolution-1}
	
	\centering \small
	The figures display the competition indices based on the charged prices (running medians over 1000 steps) for ten runs of the experiment in thin lines. The thick line shows the average of all runs and thus indicates the overall trend.
\end{figure}

\begin{figure}
	\centering
	
	\hfill
	\begin{subfigure}{0.44\textwidth}
		\includegraphics[width=\textwidth]{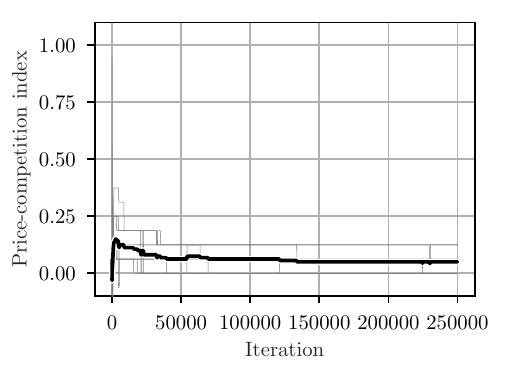}
		\subcaption{\textit{$\epsilon$-Greedy}}
	\end{subfigure}
	\hfill
	\begin{subfigure}{0.44\textwidth}
		\includegraphics[width=\textwidth]{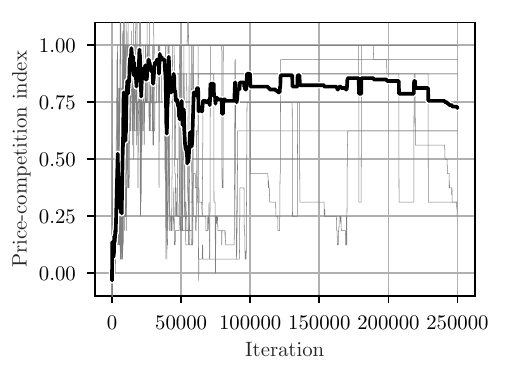}
		\subcaption{\textit{UCB-T}}
	\end{subfigure}
	\hfill
	
	\caption{Evolution of Prices During Training}
	\label{fig:price-evolution-2}
	
	\centering \small
	The figures display the competition indices based on the charged prices (running medians over 1000 steps) for ten runs of the experiment in thin lines. The thick line shows the average of all runs and thus indicates the overall trend.
\end{figure}

In stark contrast to this, the \textit{UCB-T} experiments exhibit a larger variation of prices, both within runs and between runs. The variation within runs indicates that agents do not simply "agree" on a supra-competitive price level. Instead, they erratically raise and lower prices in a seemingly random pattern. Since these agents choose actions by maximizing their heuristic (which is a non-smooth operation), their actions can vary drastically if multiple actions experience values of the heuristics at around the same level. Each run of our \textit{UCB-T} experiments is different, and we observed a strong variation between the runs. We note that the only source of randomness comes from the random action selection if one or more actions have the same heuristic value. The varying outcomes make it hard for firms to reliably predict their returns when employing a \textit{UCB-T} algorithm in this scenario. On average, though, we observe supra-competitive prices, starting from an early point in the training. 

\subsection{Supra-competitive Prices with More Sellers}
\label{sec:experiments-with-more-sellers}

There may be more than two competitors in a market. There are two major questions we answer in this context: 1) Do non-competitive algorithms such as \textit{UCB-T} eventually reach competitive prices if the number of competitors grows large enough? 2) How many competitive agents must be inserted in a system of cooperating agents to break supra-competitive pricing? Our second result sheds light on these questions.

\begin{result}
	With an increasing number of competitors, competition indices decrease. Pure \textit{UCB-T} oligopolies can still achieve supra-competitive prices, but at a lower level. Adding other bandit algorithms reduces price levels.
\end{result}

To answer the first question, we ran experiments with identical algorithms in the linear demand setting (\textit{O2}). Figure \ref{fig:large-oligopolies} shows that with an increasing number of players, \textit{UCB-T} comes closer to the Nash equilibrium prices. However, price levels remain supra-competitive in oligopolies of up to ten players. In line with our previous experiments, the other algorithms remain competitive for any number of players.

We answer the second question by configuring oligopolies where all players commit to the \textit{UCB-T} algorithm except for one. The results are displayed in Figure \ref{fig:UCB-oligopolies}. 
While the price levels decrease compared to the pure \textit{UCB-T} oligopolies, the \textit{UCB-T} algorithms are still able to coordinate on supra-competitive prices.

\begin{figure}[h]
	\centering
	
	\hfill
	\begin{subfigure}{0.44\textwidth}
		\includegraphics[width=\textwidth]{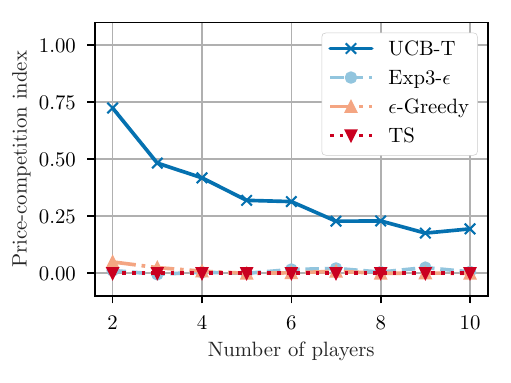}
		\subcaption{Identical algorithms \label{fig:large-oligopolies}}
	\end{subfigure}
	\hfill
	\begin{subfigure}{0.44\textwidth}
		\includegraphics[width=\textwidth]{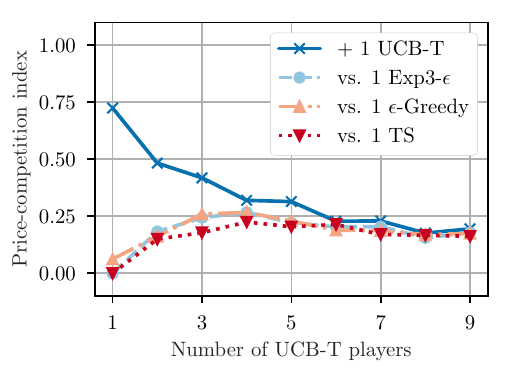}
		\subcaption{\textit{UCB-T} majority \label{fig:UCB-oligopolies}}
	\end{subfigure}
	\hfill
	
	\caption{Joint Profits with Increasing Number of Players}
	\label{fig:oligopolies}
	
	\centering \small
	The plots visualize the price-competition indices at the end of training for different numbers of players in a Bertrand economy with linear demand. \textbf{a)} shows oligopolies with identical algorithms while \textbf{b)} shows the combination of $\n - 1$ \textit{UCB-T} algorithms with one different algorithm. The median price of the last 1000 iteration steps is averaged over all runs to form the index.
\end{figure}

Of course, this coordination is only possible if there are sufficiently many players to support it. We analyze the competition of varying shares of \textit{UCB-T} agents in a 10-player oligopoly with competing \textit{Exp3-$\epsilon$} algorithms with linear demand (\textit{O2}). The result is displayed in Figure \ref{fig:UCB-fraction}. We only observe supra-competitive prices if the majority of agents play \textit{UCB-T}. With increasing numbers, prices increase, too.

\begin{figure}
	\centering
	
	\includegraphics[width=0.44\textwidth]{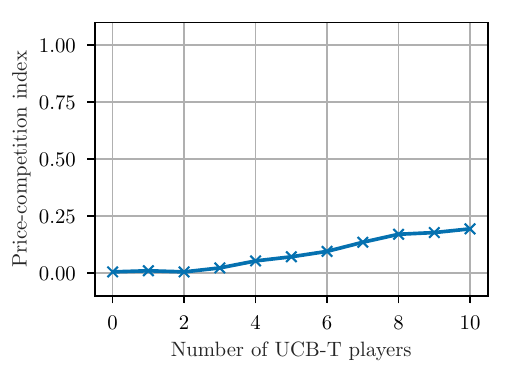}
	
	\caption{Different Numbers of \textit{UCB-T} Agents Against \textit{Exp3-$\epsilon$} Competitors in a 10-player Oligopoly}
	\label{fig:UCB-fraction}
	
	\centering \small
	The plot visualizes the price-competition index in 10-player oligopoly experiments with linear demand (\textit{O2}) for an increasing share of \textit{UCB-T} players. The remaining agents act according to \textit{Exp3-$\epsilon$}. The median price of the last 1000 iteration steps is averaged over all runs to form the index.
\end{figure}

\subsection{Staggered Entry}
\label{sec:staggered-entry-experiments}

We may not always assume that competing agents start at the same time. In this section, we complement the analytical result of Section \ref{sec:staggered-entry} with experiments. Our modified setup is as follows: We simulate the environment for 250,000 steps. In the first 20\% of all iterations, only one player is in the market. Then, the second contestant enters. A logit-demand model is most appropriate here as it ensures that the single-player version of the environment exhibits bounded, positive demand. The parameters are selected as in the symmetric logit-demand duopoly described in Section \ref{sec:experiment-setup}. 

We ran simulations with Exp3 and Thompson Sampling (TS) for 250,000 steps. The results are displayed in Figure \ref{fig:staggered-entry}. In both scenarios, the first competitor learns to charge a high price in the beginning. As soon as the second seller enters the market, the first seller immediately starts lowering prices and adapts to the increased competition. However, convergence is slow, and none of the scenarios fully recovers the competitive price within the given time frame. 

\begin{figure}
	\centering
	
	\hfill
	\begin{subfigure}{0.44\textwidth}
		\includegraphics[width=\textwidth]{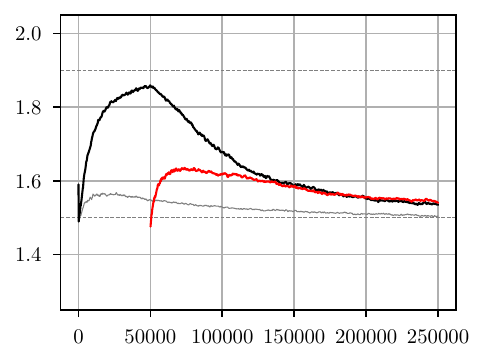}
		\subcaption{\textit{Exp3-$\epsilon$}}
	\end{subfigure}
	\hfill
	\begin{subfigure}{0.44\textwidth}
		\includegraphics[width=\textwidth]{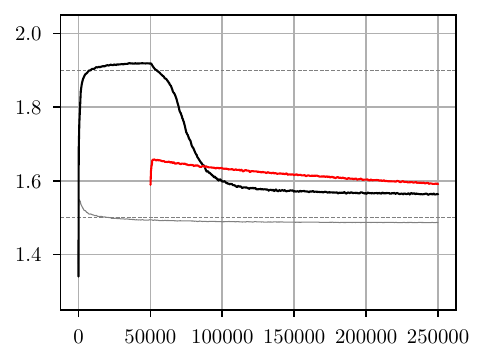}
		\subcaption{\textit{TS}}
	\end{subfigure}
	\hfill
	
	\caption{Staggered-entry experiments where the second contestant entered the market at one-fifth of the iterations.}
	\label{fig:staggered-entry}
	
	\centering \small
	Displayed is the action averaged over ten independent runs and 1000 subsequent steps. The first contestant is displayed in black, the second entrant in red. The light grey curve represents a simultaneous-entry version of the scenario for reference. The (continuous) Nash equilibrium is located at approximately 1.47; the joint profit-maximizing prices and the monopoly price are roughly 1.92. The discrete Nash equilibrium and monopoly price are marked by dashed lines.
\end{figure}

Depending on how long the entry of the second entrant is delayed, \textit{the two agents might not recover their equilibrium outcome} due to a fixation of the first contestant. While they are a monopolist, the first player might start to reject low actions as these don't yield a high reward in a monopolistic market. We could observe that the probability assigned to such actions becomes very small, close enough to be considered virtually equivalent to zero from a computer's perspective. This could lead to shrinking flexibility and thus prohibit competitive play in some scenarios.

\subsection{Random Coefficients Models: The Cereal Market}
\label{sec:cereal-market-simulations}

The examples in the main part of our paper are based on toy scenarios with a few agents and simple market structures. While they are a good choice for analyzing the algorithms' behavior in an isolated environment, they lack the complexity of real-world markets. We investigate an additional market scenario which is based on the more complex random coefficients model introduced by \citet{berry_automobile_1993}. In particular, we build on the work by \citet{nevo_practitioners_2000, nevo_measuring_2001}, which analyzes the prices in the American Cereal Market in the years 1988 to 1992. This market seems particularly suited to our case as it exhibits high margins, raising the question of whether or not real firms compete as expected \citep{nevo_measuring_2001}. We leverage the PyBLP package \citep{conlon_best_2020, conlon_pyblp_2025}, which provides dummy data for the cereal market (as used in \cite{nevo_practitioners_2000, conlon_best_2020}) and code for parameter estimation and counterfactual analysis. Our random coefficients model follows the description provided in the PyBLP documentation.

The data comprises 94 markets, each characterized by a combination of city and quarter. Every market consists of five firms, each holding one to nine products with individual prices. Further increasing the complexity is the interplay between products. Due to the random coefficients model, products show more intricate substitution patterns and price elasticities. Finally, firms with multiple products can coordinate their prices to maximize their aggregated profit.

\begin{algorithm}
	\caption{Code parts for random coefficients experiments}
	\label{alg:random-coefficients}
	\vspace{\baselineskip}
	\footnotesize
	
	\lstdefinestyle{myPython}{%
		language=Python,
		basicstyle=\small\ttfamily,
		keywordstyle=\color{blue},
		stringstyle=\color{orange},
		commentstyle=\color{niceGreen},
		morekeywords={self,as,with}, 
		showstringspaces=false,%
		tabsize=4,
		keepspaces=true,
		backgroundcolor=\color{gray!10},
		escapechar=?,
		moredelim=[is][\color{teal}]{^}{^},  
		upquote=true,
	}
	\lstset{
		frame=tlbr,
		framesep=0.1cm,          
		columns=flexible,
		linewidth=\columnwidth,
		style=myPython,
		framexleftmargin=0.2cm,
		framexrightmargin=0.2cm,
		basicstyle=\fontencoding{T1}\selectfont,
	}
	
	\begin{lstlisting}
		# Create the random coefficients model
		problem = pyblp.Problem(
		product_formulations=(
		pyblp.Formulation("0 + prices", absorb="C(product_ids)"), 
		pyblp.Formulation("1 + prices + sugar + mushy"),
		),
		agent_formulation=pyblp.Formulation(
		"0 + income + age + child" 
		),  # Demographics
		product_data=pd.read_csv(
		pyblp.data.NEVO_PRODUCTS_LOCATION
		), 
		agent_data=pd.read_csv(
		pyblp.data.NEVO_AGENTS_LOCATION
		), 
		)
		
		# Estimate the parameters
		results = problem.solve(
		sigma=np.diag([0.3302, 2.4526, 0.0163, 0.2441]),
		pi=[
		[5.4819,    0,          0       ],
		[15.8935,   -1.2000,    2.6342  ],
		[-0.2506,   0,          0       ],
		[1.2650,    0,          0       ],
		],
		method="1s",  # One-step GMM
		optimization=pyblp.Optimization(
		"bfgs", {"gtol": 1e-5}
		), 
		)
		
		# Compute market shares and profits for varying prices in a specific market
		market_id = "C01Q1"
		prices = np.array([...])
		shares = results.compute_shares(prices=prices, market_id=market_id)
		profits = results.compute_profits(
		prices=prices, shares=shares, market_id=market_id
		)
	\end{lstlisting}
\end{algorithm}

We instantiate the model, estimate its parameters based on the full dataset, and compute market shares for prices in individual markets. In essence, the code is as described in Algorithm \ref{alg:random-coefficients}.\footnote{See also \url{https://pyblp.readthedocs.io/en/stable/_notebooks/tutorial/nevo.html}.} For our market simulations, every product is represented by one agent that learns to set optimal prices for this product. Eligible cooperation between the products of one firm can be implemented by sharing the reward between the respective agents (see below).

The experiments are mostly configured as described in the main part of our paper. However, there are three specific parameters that we need to choose.

\begin{itemize}
	\item \textbf{Market ID}: We can choose from a set of 94 markets, each with a unique combination of city and quarter, that all comprise 24 products. We arbitrarily select "C01Q1" and "C07Q2" as our choices.
	\item \textbf{Agents}: For our experiments, we employ Exp3 and Thompson sampling.
	\item \textbf{Profit model}: As described above, the market consists of multiple firms with one to five products each. We can either let each agent observe the \emph{aggregate profit of all products} of the corresponding firm (profit model "firm") or let every agent only receive the profit of its own product (profit model "product").
\end{itemize}

The agents can choose from 21 actions between 0 and 0.25, which are prices that extend the equilibrium and monopoly prices of the continuous game. The equilibrium prices were computed by PyBLP. For the monopoly prices, we simulated a merger of all firms and evaluated the resulting optimal prices with PyBLP. We scaled the rewards based on the profits with equilibrium and monopoly prices to ensure that our algorithms' default parameters performed reasonable.\footnote{This does not change the strategic nature of the game.}
We ran our simulations for 500,000 steps and report the price competition indices over time. Note that the price competition index is now computed based on the reference prices of the \emph{continuous} game, as the estimation of the discrete equilibria is out of reach for 24 players. All simulations have been repeated 5 times (with seeds 0 to 4).

The results are visualized in Figure \ref{fig:random-coefficient-results-1} and Figure \ref{fig:random-coefficient-results-2}. In the first configuration, we observe a steep price increase in the very beginning, followed by a steady decline in prices. This indicates a consistent behavior with previous results in smaller markets. However, we also observe that convergence is much slower in market C07Q2. In general, slow convergence is to be expected as the large number of players and the complexity of the interaction make the feedback noisy and impede learning. 

With Thompson sampling, our results are less promising. Interestingly, we observe a split between runs: Half of our runs resulted in high competition indices, while the other half was much more benign. A potential reason is, again, slow convergence due to complexity. In market C07Q2 (not displayed), the convergence of Thompson sampling was slightly better, resulting in an average competition index of about 0.25. 

Lastly, the alternative profit model "product", where every agent is evaluated on the profit of its product alone, shows comparably quick convergence to competitive prices. Keep in mind that the \textit{reference price} is still based on the assumption that \textit{a firm coordinates its prices}, which means that competition index values below zero can occur easily with this non-cooperative profit model. 

\begin{figure}
	\centering
	
	\hfill
	\begin{subfigure}{0.44\textwidth}
		\includegraphics[width=\textwidth]{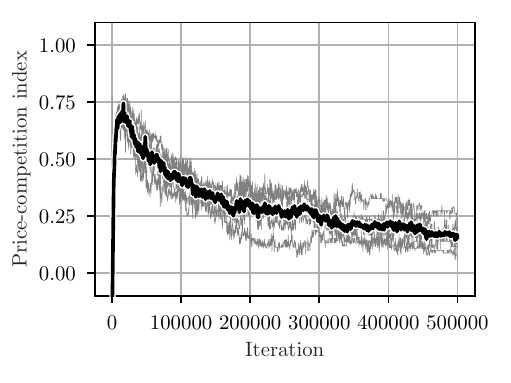}
		\subcaption{C01Q1, Exp3-$\epsilon$, "firm"}
	\end{subfigure}
	\hfill
	\begin{subfigure}{0.44\textwidth}
		\includegraphics[width=\textwidth]{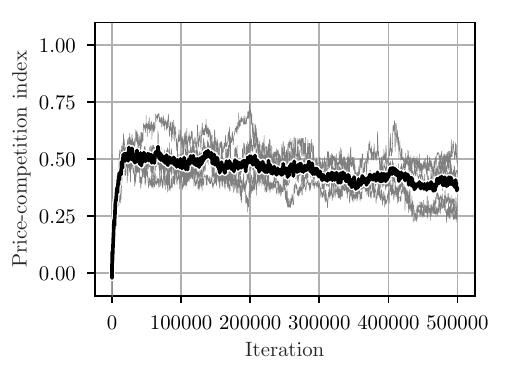}
		\subcaption{C07Q2, Exp3-$\epsilon$, "firm"}
	\end{subfigure}
	\hfill
	
	\caption{Price competition indices in the random coefficients model.}
	\label{fig:random-coefficient-results-1}
	
	\centering \small
	We report the results for markets C01Q1 and C07Q2. The competition indices are based on median price of the previous 1000 steps and are averaged over 24 agents and five runs. The results of the individual runs are displayed in light gray.
\end{figure}

\begin{figure}
	\centering
	
	\hfill
	\begin{subfigure}{0.44\textwidth}
		\includegraphics[width=\textwidth]{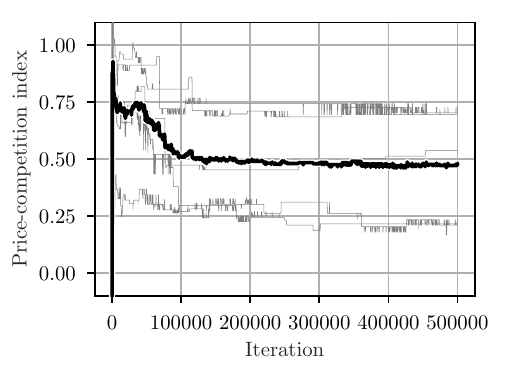}
		\subcaption{C01Q1, TS, "firm"}
	\end{subfigure}
	\hfill
	\begin{subfigure}{0.44\textwidth}
		\includegraphics[width=\textwidth]{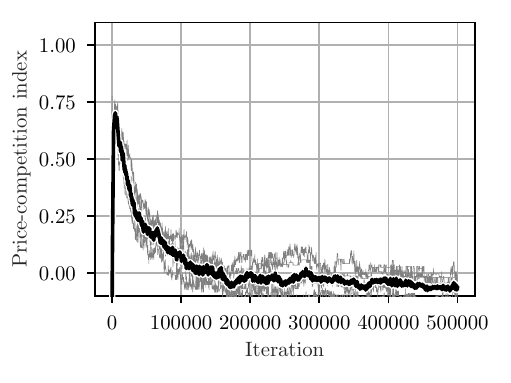}
		\subcaption{C01Q1, Exp3-$\epsilon$, "product"}
	\end{subfigure}
	\hfill
	
	\caption{Price competition indices in the random coefficients model.}
	\label{fig:random-coefficient-results-2}
	
	\centering \small
	We report the results for algorithm "TS", and profit models "product". The competition indices are based on median price of the previous 1000 steps and are averaged over 24 agents and five runs. The results of the individual runs are displayed in light gray.
\end{figure}

\FloatBarrier

\subsection{Bertrand Duopoly Experiments with More Algorithms}\label{app:qlearning}

We also ran experiments with another variant of the UCB algorithms described by \citet{auer2002bandit}, \textbf{UCB1}, and with Q-Learning (\textbf{QL}). Following \citet{Calvano.2020}, our implementation of \textbf{QL} uses an $\epsilon$-greedy exploration strategy with an exponentially decaying exploration rate $\epsilon$. A small exploration rate towards the end seems important for the evolution of supra-competitive prices with Q-Learning \citep{barfuss_intrinsic_2023, schaefer_emergence_2023, lambin2024less}. For all experiments where \textbf{QL} was involved, we iterated 3 million steps (instead of 250,000), and we provided the actions of both players from the previous iteration as the state: $s_t = (a_{1, t-1}, a_{2, t - 1})$. For completeness, we also repeated our experiments with pretrained Q-Learning agents, but could not see substantial differences in the final outcome.

Figure \ref{fig:QL-collusion-price} shows the price competition indices of all algorithm configurations. We observe that, in line with what we observed in our main experiments, non-competition mostly evolves when two algorithms of the same type are paired. The combination of \textbf{UCB-T} and \textbf{UCB1} deviates from this pattern and displays a noticeable level of cooperation. It is likely that these two algorithms can learn to "coordinate" (unconsciously) because of their similarity, which shows that algorithms need not share the exact same implementation and parameters in order to cooperate.

\begin{figure}
	\centering
	\includegraphics[width=0.49\textwidth]{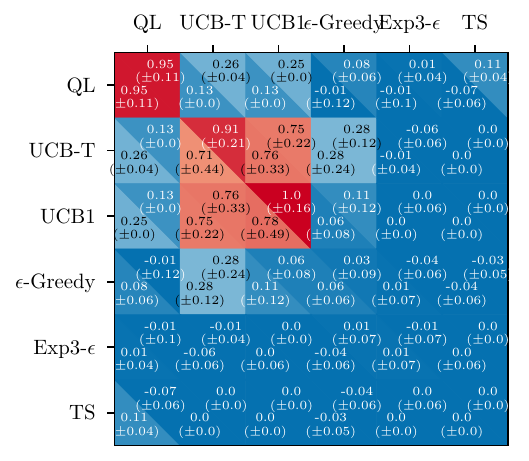}
	\caption{Price-competition indices in duopoly \textbf{O3} (logit demand) \label{fig:QL-collusion-price}}
	
	\centering \small
	We display the price-competition index \textit{by agent} as an average value over all runs. To determine the index, we use the median price over the last 1000 iterations.
\end{figure}

For the Q-learning experiment, Figure \ref{fig:price-evolution-QL} shows the evolution of prices over the experiment iterations. We observe that high prices evolve late and abruptly, which can likely be attributed to the exploration rate becoming sufficiently small.

Let us elaborate a little on the relation between Q-learning and the $\epsilon$-greedy bandit. Both share an explicit exploration term and a similar pattern of estimating average rewards/valuations. However, they differ in two ways: First, Q-learning can base its actions on previous prices, giving it a method to "retaliate" for deviations from a supra-competitive price. This is the pathway that \citet{calvano_algorithmic_2021} use to explain the evolution of algorithmic collusion, although \citet{lambin2024less} argues that stateless Q-learning, which resembles our $\epsilon$-greedy bandit, can reach supra-competitive outcomes as well. Second, the implementation of Q-learning uses a \textit{decaying} exploration term $\epsilon_t$, while our $\epsilon$-greedy bandit uses a fixed rate. Following the argumentation by \citet{den_boer_artificial_2022} and \citet{lambin2024less} that this change of exploration leads to coordination, we conjecture that this explains the observed differences.

\begin{figure}
	\centering
	\includegraphics[width=0.49\textwidth]{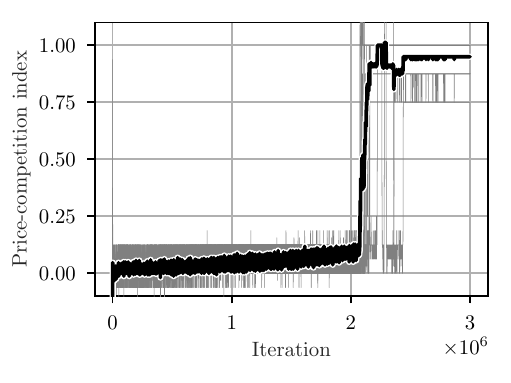}
	
	\caption{Evolution of prices with \textbf{QL} \label{fig:price-evolution-QL}}
	
	\centering \small
	We display the price-competition index of of a \textbf{QL}-duopoly over time. To determine the index, we use the median price over 1000 iterations.
\end{figure}

\subsection{Results by Demand Function}

In the main part of the paper, we averaged the competition indices of the duopolies for all demand functions. Here, we provide the results for the individual demand functions. Figures \ref{fig:price-collusion-symmetric} and \ref{fig:price-collusion-asymmetric} display the price-competition indices for all five games. We find that supra-competitive prices evolve consistently with the \textbf{UCB-T} and \textbf{UCB1} algorithms throughout all settings, with particularly pronounced cooperation in the standard demand setting (\textbf{O1}). Likewise, prices remain mostly competitive for the remaining algorithm combinations. 

\begin{figure}
	\centering
	
	\hfill
	\begin{subfigure}{0.3\textwidth}
		\includegraphics[width=\textwidth]{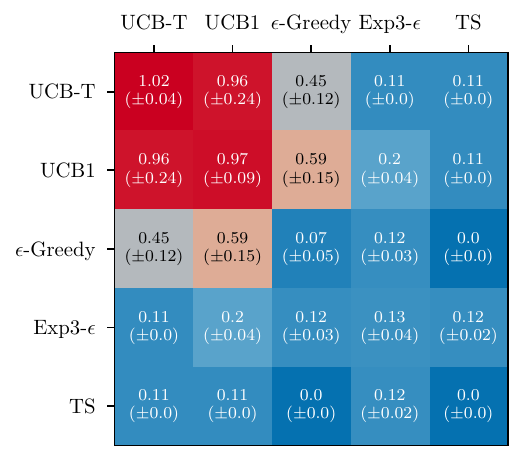}
		\subcaption{\textbf{O1}.}
	\end{subfigure}
	\hfill
	\begin{subfigure}{0.3\textwidth}
		\includegraphics[width=\textwidth]{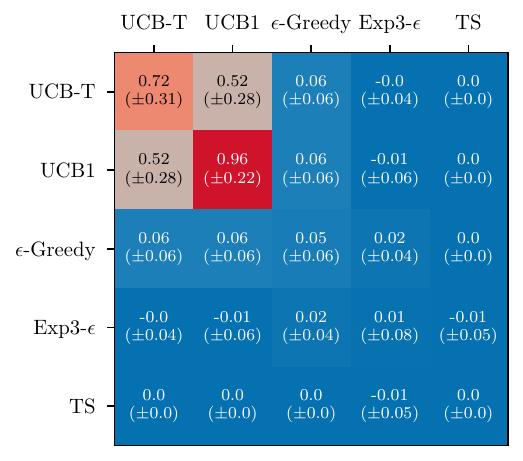}
		\subcaption{\textbf{O2}.}
	\end{subfigure}
	\hfill
	\begin{subfigure}{0.3\textwidth}
		\includegraphics[width=\textwidth]{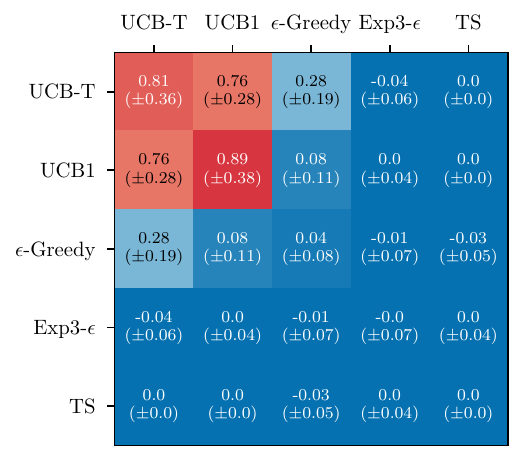}
		\subcaption{\textbf{O3}.}
	\end{subfigure}
	\hfill
	
	\caption{Price-competition indices in symmetric games}
	\label{fig:price-collusion-symmetric}
	
	\centering \small
	We visualize the price-competition indices for the symmetric games and several combinations of algorithms. The displayed numbers are the numeric values of the competition indices; the values in brackets display their corresponding standard deviation. Averages and standard deviations were computed over two agents and 10 runs.
	
\end{figure}

\begin{figure}
	\centering
	
	\hfill
	\begin{subfigure}{0.44\textwidth}
		\includegraphics[width=\textwidth]{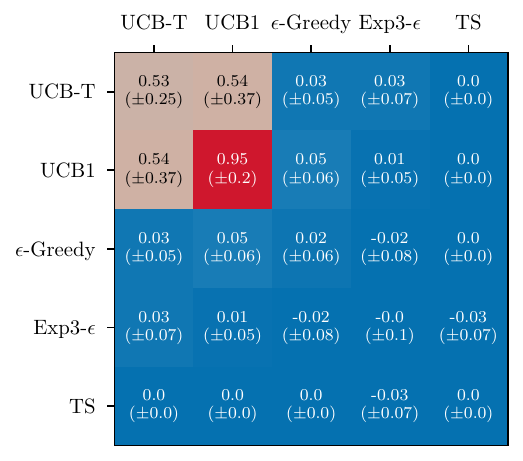}
		\caption{\textbf{O2'}.}
	\end{subfigure}
	\hfill
	\begin{subfigure}{0.44\textwidth}
		\includegraphics[width=\textwidth]{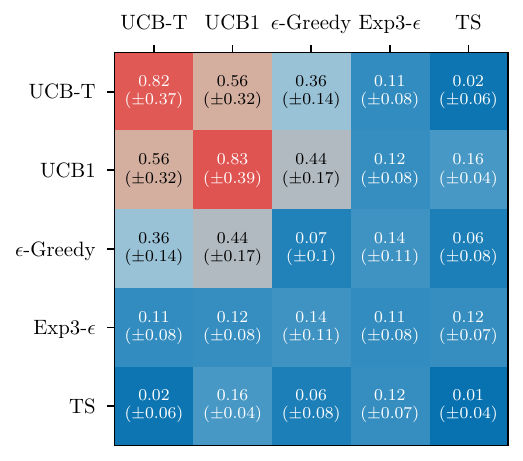}
		\caption{\textbf{O3'}.}
	\end{subfigure}
	\hfill
	
	\caption{Price-competition indices in asymmetric games}
	\label{fig:price-collusion-asymmetric}
	
	\centering \small
	We visualize the price-competition indices for the asymmetric games and several combinations of algorithms. The displayed numbers are the numeric values of the competition indices; the values in brackets display their corresponding standard deviation. Averages and standard deviations were computed over two agents and 10 runs.
\end{figure}

\end{document}